\documentclass[10pt,a4paper,english,preprint,3p,number,sort&compress,times]{elsarticle}
\usepackage[T1]{fontenc}
\usepackage[utf8]{inputenc}
\usepackage{xcolor}
\usepackage{mathrsfs}
\usepackage{mathtools}
\usepackage{bm}
\usepackage{amsmath}
\usepackage{amssymb}
\usepackage{cancel}
\usepackage{graphicx}
\usepackage{esint}

\makeatletter

\pdfpageheight\paperheight
\pdfpagewidth\paperwidth

\usepackage{mathrsfs}
\usepackage{bm}

\usepackage[]{ulem}

\usepackage{mathrsfs}
\usepackage{bm}
\usepackage[]{ulem}

\usepackage{amsthm}
\usepackage{tikz}
\usetikzlibrary{arrows.meta}

\usepackage{braket}
\newtheorem{thm}{Theorem}\newtheorem{lem}[thm]{Lemma}

\newcommand{\Dcausal}{%
\tikz[baseline]{\draw[-{Latex[length=2mm]}] (0,0.5ex)--(3ex,0.5ex) node[right]{};
\draw (2ex,0.5ex)--(4ex,0.5ex)}%
}

\newcommand{\DcausalDotA}{
\tikz[baseline]{%
\draw[-{Latex[length=2mm]}] (0,0.5ex)--(2.2ex,0.5ex) node[right]{};
\draw[black,fill=black] (3ex,0.5ex) circle(.4ex);
\draw (2ex,0.5ex)--(4.5ex,0.5ex)}%
}

\newcommand{\DcausalDotB}{
\tikz[baseline]{%
\draw[black,fill=black] (1.8ex,0.5ex) circle(.4ex);
\draw[-{Latex[length=2mm]}] (2.6ex,0.5ex)--(3.8ex,0.5ex) node[right]{};%
\draw (0,0.5ex)--(4.5ex,0.5ex)}%
}

\newcommand{\GH}{
 \tikz[baseline]{%
	\draw (0.0,0.5ex) sin (0.2ex,1ex) cos (0.4ex,0.5ex) sin(0.6ex,0) cos (0.8ex,0.5ex) sin (1.0ex,1.0ex) cos(1.2ex,0.5ex)
	sin (1.4ex,0ex) cos (1.6ex,0.5ex) sin(1.8ex,1.0ex) cos (2.0ex,0.5ex) sin (2.2ex,0) cos(2.4ex,0.5ex) sin (2.4ex,1.0ex)
	cos (2.6ex,0.5ex) sin(2.8ex,0ex) cos (3.0ex,0.5ex) sin (3.2ex,1.0ex) cos(3.4ex,0.5ex) sin (3.6ex,0) cos(3.8ex,0.5ex)
	sin (4.0ex,1ex) cos (4.2ex,0.5ex) sin(4.4ex,0ex) cos (4.6ex,0.5ex) 	 
	}%
}

\renewcommand{\textcolor}[2]{#2}
\usepackage{babel}
\journal{Annals of Physics}
\usepackage[titles]{tocloft}

\makeatother

\usepackage{babel}
\begin{document}

\begin{frontmatter}{}

\title{Nonequilibrium Steady State and Heat Transport in Nonlinear Open
Quantum Systems: Stochastic Influence Action and Functional Perturbative
Analysis}

\author{Jing Yang}

\ead{jyang75@ur.rochester.edu}

\address{Department of Physics and Astronomy, University of Rochester, Rochester,
New York 14627, USA}

\author{Jen-Tsung Hsiang}

\ead{cosmology@gmail.com}

\address{Center for High Energy and High Field Physics, National Central University,
Chungli 32001, Taiwan}

\address{Center for Particle Physics and Field Theory, Department of Physics,
Fudan University, Shanghai 200438, China}

\author{Andrew N. Jordan}

\ead{jordan@pas.rochester.edu}

\address{Department of Physics and Astronomy, University of Rochester, Rochester,
New York 14627, USA}

\address{Institute for Quantum Studies, Chapman University, 1 University Drive,
Orange, CA 92866, USA}

\author{B. L. Hu}

\ead{blhu@umd.edu}

\address{Joint Quantum Institute and Maryland Center for Fundamental Physics,
University of Maryland, College Park, Maryland 20742, USA}
\begin{abstract}
In this paper, we show that a nonequilibrium steady state (NESS) exists
at late times in open quantum systems with weak nonlinearity by following
its nonequilibrium dynamics with a perturbative analysis. We consider
an oscillator chain containing three-types of anharmonicity: cubic
$\alpha$- and quartic $\beta$-type Fermi-Pasta-Ulam-Tsingou (FPUT)
nearest-oscillator interactions and the on-site (pinned) Klein-Gordon
(KG) quartic self-interaction. Assuming weak nonlinearity, we introduce
a stochastic influence action approach to the problem and obtain the
energy flows in different junctures across the chain. The formal results
obtained here can be used for quantum transport problems in weakly
nonlinear quantum systems. For $\alpha$-type anharmonicity, we observe
that the first-order corrections do not play any role in the thermal
transport in the NESS of the configuration we considered. For KG and
$\beta$-types anharmonicity, we work out explicitly the case of two
weakly nonlinearly coupled oscillators, with results scalable to any
number of oscillators. We examine the late-time energy flows from
one thermal bath to the other via the coupled oscillators, and show
that both the zeroth- and the first-order contributions of the energy
flows become constant in time at late times, signaling the existence
of a late-time NESS to first order in nonlinearity. Our perturbative
calculations provide a measure of the strength of nonlinearity for
nonlinear open quantum systems, which may help control the mesoscopic
heat transport distinct from or close to linear transport. Furthermore,
our results also give a benchmark for the numerical challenge of simulating
heat transport. Our setup and predictions can be implemented and verified
by investigating heat flow in an array of Josephson junctions in the
limit of large Josephson energy with the platform of circuit QED. 
\end{abstract}

\end{frontmatter}{}

\newpage

\tableofcontents{}

\newpage

\section{Introduction}

Inasmuch as the equilibrium state of a system in contact with one
heat bath is of fundamental importance in the conceptualization and
application of the powerful canonical ensemble in statistical thermodynamics,
the existence of a nonequilibrium steady state (NESS) is of similar
importance for the understanding of nonequilibrium processes which
are widely present in Nature. For example NESS is the arena in the
discovery of the fluctuation theorems of various forms from Gallivotti
and Cohen \citep{GalCoh,Evans,Kurchan,LebSpoFluc,Seifert} to Jarzynsky
and Crook \citep{Jar97,Crook,FlucThmRev,JarRev,SubHu} and the basis
for the investigation of quantum thermodynamics (e.g., \citep{Mahler,Kosloff,Anders,QTDQI}).
One common feature of the equilibrium state and nonequilibrium steady
state is their ubiquitous appearance--many, but not all, systems
approach an equilibrium or steady state given sufficient time. While
the exceptions are not so easily captured but of special interest,
the ``norm'' also needs to be proven or demonstrated for different
types of systems and their environments before one can use the many
attractive features of NESS to exert control of nonequilibrium systems
in such states.

For classical many-body systems the existence and uniqueness of NESS
have been studied by mathematical physicists in statistical mechanics
for decades. For Gaussian systems (such as a chain of harmonic oscillators
with two heat baths at the two ends of the chain) \citep{NESSclHOC}
and anharmonic oscillators under general conditions \citep{NESSclAHO}
there are definitive answers in the form of proven theorems. Answering
this question for quantum many-body systems is not so straightforward.
This is the challenge two of the present authors have undertaken \citep{HHAoP},
showing the existence of NESS at late times for a quantum \textit{harmonic}
chain with two ends coupled linearly to two separate heat baths at
different temperatures. (See also \citep{Motz}). The goal of this
work, is to demonstrate the existence of NESS for a system of two
weakly \textit{nonlinearly} coupled (generalizable to a chain of )
quantum oscillators linearly coupled to two heat baths of different
temperatures. The forms of interaction amongst the system oscillators
correspond to the Fermi-Pasta-Ulam-Tsingou (FPUT) $\alpha$ (cubic
nearest neighbor interaction), $\beta$ (quartic nearest neighbor
interaction) \citep{FPU,FPUrev} or the Klein-Gordon (KG)(`pinned',
or quartic self-interaction) \citep{Montroll,Bender} models. The
current work approaches to this problem using a stochastic influence
action which is a part of the functional techniques systematically
developed in \citep{cgea,JH1,GH1,CRV} and applied to transport problems
in \citep{HHAoP,NENL1}. 
A fluctuation-dissipation relation is shown to exist in such open
nonlinear quantum systems by nonperturbative methods in \citep{NLFDR}.
In this paper using the stochastic influence action, we aim to demonstrate
the existence of a NESS at late times for the FPUT and KG models with
weak nonlinearity.

In terms of physical contexts and applications, we give a brief sketch
of the problem of quantum transport in nonlinear systems, motivating
the usefulness of analytic results, perturbative notwithstanding.
We then give a brief summary of the functional perturbative method
developed in \citep{HPZ93} and a description of our approach. We
illustrate this by calculating the energy transfer between the weakly
nonlinear system and its environment, showing that at late times indeed
they balance to first order of nonlinearity, signifying the existence
of a NESS. Further background descriptions can be found in the Introductions
of \citep{HHAoP} and \citep{NENL1}.


\subsection{Quantum transport of nonlinear systems - numerically challenging}

Quantum transport is an important class of problems where the nonequilibrium
evolution of an open quantum system has been followed and applicable
laws examined. A well-known case is the Fourier law for heat conduction
in low-dimensional lattices \citep{Fourier,LLP}. Amongst the vast
literature on this subject suffice it for our purpose here to mention
three reviews \citep{LLPrep,Dhar,LiRen}. As for methodologies closest
to ours, we mention the density matrix approach (e.g., \citep{DSH}),
the nonequilibrium Green function approach~ (e.g., \citep{KadBay,MacWanGuo,Rammer,WALT}),
the stochastic path integral developed in the context of charge transport
in mesoscopic systems \citep{PJSB03,JSP04}, the closed time path
(CTP, Schwinger-Keldysh, or `in-in') ~\citep{ctp,CalHu08,Kamenev},
the influence functional \citep{Weiss,IF,HPZ92} and the stochastic
influence action formalisms \citep{cgea,JH1,GH1,CRV} which are the
methods of choice for two of the present authors.

Transport problems depend heavily on numerical computations (e.g.,
\citep{Zhao}), which enable one to identify qualitatively different
behaviors in different regimes. Many analytical approaches, e.g.,
Boltzmanm molecular dynamics \citep{QBE}, nonlinear kinetic theory
and fluctuating hydrodynamics \citep{Spohn} when combined with numerical
methods give reasonably good predictions for classical systems in
different parameter regimes. However, for quantum many-body systems,
numerical computations are a lot more difficult to carry out. For
example, even casting the Langevin equations for non-linear quantum
systems suitable for numerical analysis poses a challenge \citep{DharPC}.

For this reason analytical results for quantum nonlinear systems,
even perturbative in nature, may serve some function, in bridging
the gaps in the parameter space regions these theories left behind
\citep{ZhaoPC}.

\subsection{Analytical results, even perturbative, are helpful}

The motivation for our investigation using perturbative methods is
to explore how analytical methods and results, no matter how unimpressive
they are in the face of full scale numerical computations, can still
serve some function in clarifying the issues behind energy or charge
transport in nonequilibrium weakly-nonlinear systems.

In recent years, several discoveries highlighted the importance of
even weak nonlinearity in the thermodynamic and transport properties
of classical many-body systems. E.g., it is shown that in the thermodynamic
limit a one-dimensional (1D) nonlinear lattice can always be thermalized
for arbitrarily small nonlinearity \citep{FZZ19}. On the effects
of nonlinear interaction on disordered chains there is a tentative
suggestion that weak nonlinearity can destroy Anderson localization
\citep{WFZZ19}. Thermalization is an important current subject in
quantum many-body physics. It is found that even weak nonlinearity
can play a pivotal role for thermalization. Thus, a different perspective
from nonlinear quantum open systems is desirable. Thermalization via
wave turbulence theory has also generated renewed interest in the
FPUT and KG models \citep{POC18,PCBLO}. 
These recent developments provided good impetus for us to implement
our program on nonequilibrium dynamics of open quantum nonlinear systems.
The first step is to explore the conditions for systems in popular
models to reach a NESS. Because knowing that a system enters such
a state from following its nonequilibrium dynamics can offer great
simplification in the analysis of the nonlinear system's long time
behavior.

\subsection{Functional perturbative method}

We use the same mathematical framework as in \citep{HHAoP}, namely,
the path-integral \citep{ctp,CalHu08}, influence functional \citep{IF,Weiss}
formalism, under which the influence action, the coarse-grained effective
action \citep{cgea} and the stochastic influence action \citep{JH1,GH1}
are defined. The stochastic equations such as the master equation
(see, e.g., \citep{HPZ92}) and the Langevin equations (see, e.g.,
\citep{CRV}) can be obtained from taking the functional variations
of these effective actions. We then invoke the functional perturbative
approach of \citep{HPZ93} developed further in \citep{NENL1} to
treat systems with weak nonlinearity. In this approach we first introduce
external sources to drive a linear (harmonic oscillator) system and
calculate the in-in generating functional. Taking the functional derivatives
with respect to the sources gives the expectation values of the covariance
matrix elements or the two-point functions of the canonical variables.
The perturbative correction to the two-point functions, or the expectation
values of quantum observables due to the nonlinear potential can also
be found order by order by taking the appropriate functional derivatives
of the generating functional of the linear system linked to external
sources.

Going beyond the stochastic path integrals \citep{CDJ13,CJ15} which
unravels the effect of a Markovian bath with continuous weak measurement,
in this work we unravel the full effect of the bath without invoking
the Markov approximation but at the price that the unravelled trajectory
is fictitious which may not correspond to any real measurement. It
has been shown that the stochastic unraveling admits a time-local
master equation \citep{StoMak} even in the presence of the nonlinearity
of the quantum open systems. In our stochastic influence action approach,
we can define the heat flow rigorously among the oscillators for our
cases where the system Hamiltonian contains $\alpha,\,\beta$-FPUT
and KG nonlinearities. Furthermore, we introduce an additional external
source $\bm{h}$, which allows one to compute correlation functions
involving momentum operators using the stochastic influence action.
These new ingredients allow us to compute \textit{steady-state} heat
currents in a more efficient and economic way than \citep{NENL1},
which provides fully nonequilibrium evolution from the transient to
the relaxation to a steady state of an open nonlinear system.

\subsection{Power balance and stationarity}

Instead of seeking mathematical proofs for the existence and uniqueness
of NESS in quantum nonlinear systems, which to us is a daunting challenge,
we aim at solving for the nonequilibrium dynamics of the model systems
mentioned above, examine how they evolve in time and see if one or
more NESS exist(s) at late times. This is the same approach used earlier
by two of the present authors in \citep{HHAoP}. Namely, taking the
functional variation of the stochastic influence action yields the
quantum stochastic equations. The Langevin equation is then used to
obtain expressions for the energy flow from one bath to another through
the nonlinear system. We can check if these two energy fluxes reach
a steady state, and an energy flow (power) balance relation exists.
The main result of this work is that at late times at least to first
order in the nonlinear interaction between the oscillators, perturbations
do not grow any faster or stronger so as to overtake and disrupt the
contributions of the linear (harmonic) order. Thus the NESS is explicitly
demonstrated to exist for the weakly nonlinear systems studied in
our models.

\subsection{Findings and outline }

Based on the perturbative formalism in the framework of quantum open
systems, we give formal expressions for the energy currents up to
the first-order of the nonlinear couplings for an FPUT and KG chain.
It consists of nonlinearly coupled anharmonic oscillators via the
$\alpha$-, $\beta$-FPUT and KG nonlinearities. Each oscillator also
linearly couples to its own private bath. We find for $\alpha$-FPUT
nonlinearity, the first-order corrections to the energy currents vanish.
We demonstrate that for the two-oscillator chain, the NESS is established
in the late time up to the first-order of the nonlinear coupling constant.
This procedure can be straightforwardly extended to a chain of arbitrary
length. In addition, we devise diagrammatic representations for the
energy currents in the NESS, which may provide an intuitive understanding
of the heat transport at NESS. As a physical application, the ratio
between the first-order and zeroth correction of the energy currents
in the steady-state provides a measure of the strength of the nonlinearity
in the context of thermal transport, which may be applied to controlling
thermal transport either distinct from or close to the linear transport.
Our analytical results suggest that when the coupling between the
chain and the baths is strong or when the temperature bias across
the chain is large, the heat transport behaves like a linear one in
the presence of the nonlinearities, offering a wider latitude to manipulate
nonlinear couplings. Our results can provide a useful benchmark for
the numerical simulations of anomalous heat transport in low dimensions
in the weakly nonlinear regime \citep{LLPrep}. Our setup can be implemented
by e.g. engineering an array of Josephson junctions in circuit quantum
electrodynamics in the large Josephson energy limit \citep{JJA-NatComm,JJA-Nature}.

This paper is organized as follows: in Sec. \ref{sec:The-general-formalsim}
we first present the general formalism of the influence functional
and functional perturbation, aiming at the problem of heat transport
in a network of oscillators and baths. As a first application of the
functional perturbation formalism, in Sec. \ref{sec:zeroth-order},
we apply the functional perturbation formalism to compute the zeroth-order
steady-state energy currents in a chain configuration with two oscillators
and two baths. In Sec. \ref{sec:The-first-order-corrections}, we
compute the first-order correction for the KG, $\beta$-FPUT and $\alpha$-FPUT
models respectively. Upon replacing causal propagator for two-oscillator
and two-bath configuration with the one for the general $N$-oscillator
and $N$-bath configuration, the results in Sec. \ref{sec:zeroth-order}
and \ref{sec:The-first-order-corrections} can be generalized straightforwardly
to the case of a chain of anharmonic oscillators. In Sec. \ref{sec:Diagrammatic-representations},
we give the diagrammatic representation of all the energy currents,
which provides an intuitive understanding of the heat transport up
to the first-order. In Sec. \ref{sec:NESS}, we prove for the two-oscillator
two-bath configuration that, the NESS exists at the late times for
both KG and $\beta$-FPUT nonlinearities. In Sec. \ref{sec:Conclusion},
we discuss the possible physical applications of our perturbative
results and summarize the findings.

\section{\label{sec:The-general-formalsim}The general formalism}

\subsection{The Feynman-Vernon formalism for a network of oscillators and baths}

We consider $N$ coupled oscillators interacting with their private
thermal baths, as shown in Fig.~\ref{fig:setup}. Each oscillator
couples to only its nearest neighbors via the linear and the nonlinear
mutual couplings, so that these oscillators form a general FPUT chain
threading through thermal baths of various temperatures. The dynamics
of this entire system can be described by the action 
\begin{equation}
S_{\text{tot}}=S+S_{\phi}+S_{I}\,.
\end{equation}
The action $S_{\chi}$ for the chained oscillators contains \textit{linear}
and \textit{nonlinear} mutual couplings between the neighboring oscillators
\begin{equation}
S_{\chi}[\bm{\chi}]=\int\!dt\,\biggl[\sum_{n=1}^{N}L_{n}-\sum_{n=1}^{N-1}V_{n,\,n+1}\biggr],\label{E:gbkgjs}
\end{equation}
in which the Lagrangian $L_{n}$ of the $n^{\text{th}}$ oscillator
of mass $m_{n}$ takes the form 
\begin{equation}
L_{n}=\frac{1}{2m_{n}}\dot{\chi}_{n}^{2}-V(\chi_{n})\,,\label{eq:L-n}
\end{equation}
and $\chi_{n}$ denote the oscillator's displacement. The site potential
$V(\chi_{n})$ may include \textit{nonlinear self-interaction} of
the oscillator. The mutual coupling between two neighboring oscillators
$n$ and $n+1$, ${V}_{n,\,n+1}$, takes the form 
\begin{equation}
V_{n,\,n+1}=\sum_{\eta=2}\frac{\lambda_{\eta}}{\eta}(\chi_{n}-\chi_{n+1})^{\eta}.\label{eq:V-n}
\end{equation}
If $V(\chi_{n})$ contains only the harmonic site potential, then
the $\eta=2$ case in fact describes a linear chain, whose nonequilibrium
dynamics and the existence of the NESS have been discussed with great
details in~\citep{HHAoP}. Thus, Eq.~\eqref{E:gbkgjs} generalize
the previous consideration to include a nonlinear site potential,
as well as nonlinear intra-oscillator couplings.

\begin{figure}
\begin{centering}
\includegraphics[scale=0.4]{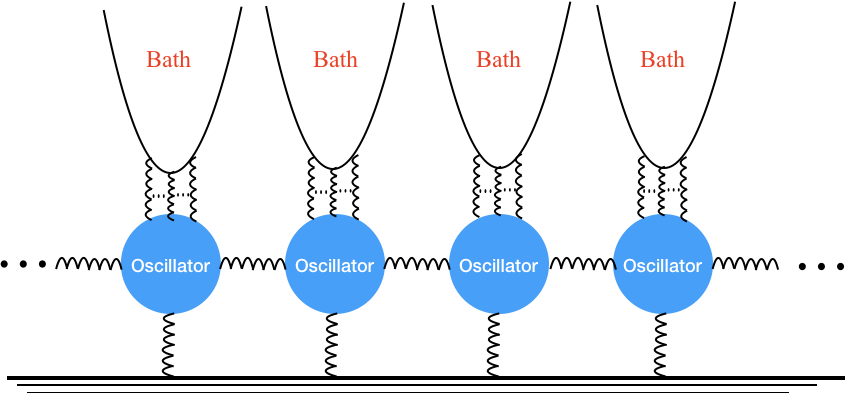} 
\par\end{centering}
\caption{\label{fig:setup}$N$-oscillators represented by blue solid circles
that are coupled to their individual thermal baths represented by
curved lines, which are theoretically modeled as massless scalar field
theory. The inter-oscillator coupling in general could be both linear
and nonlinear while the coupling between the chain and the bath is
assumed to be linear. The coupling between the oscillators with the
substrate contains the linear as well as nonlinear site potential. }
\end{figure}

The private baths associated with each oscillator are modeled by independent
free massless scalar fields, initially prepared in their respective
thermal states. They can have different initial temperatures in general,
and their action is given by 
\begin{equation}
S_{\phi}[\bm{\phi}]=\int\!d^{4}x\;\sum_{n=1}^{N}\partial_{\mu}\phi_{n}(\bm{x},\,t)\partial^{\mu}\phi_{n}(\bm{x},\,s)\,,
\end{equation}
where $\phi_{n}$ denotes the private bath attached to the $n^{\text{th}}$
oscillator. The interaction between the oscillator and its bath takes
the bilinear form 
\begin{align}
S_{I}[\bm{\chi},\,\bm{\phi}] & =\int\!dt\;\sum_{n=1}^{N}e_{n}\,\chi_{n}(t)\phi_{n}(\bm{z}(t),\,t)\,.
\end{align}
The parameter $\bm{z}(t)$ describes the trajectory of each oscillator,
but we assume that the oscillators are held at fixed position in our
subquent treatment. The coupling strength $e_{n}$ between the $n^{\text{th}}$
oscillator and its private bath may take on different values and are
not restricted to weak coupling. In addition, to avoid cluster of
notations, we introduce the matrix representation such that, say,
$\bm{\chi}=(\chi_{1},\dots,\chi_{N})^{T}$ to describe $N$ oscillator
as a whole.

We will investigate the nonequilibrium dynamics of this anharmonic
chain in the framework of quantum open systems, where the oscillator
chain is the (reduced) system of our interest and the private baths
serve as the environment. The dynamics of the system is fully governed
by the reduced density matrix $\rho$, whose evolution can be described
by the propagating function $J(\bm{\chi}_{f},\,\bm{\chi}_{f}^{\prime},\,t_{f};\bm{\chi}_{i},\,\bm{\chi}_{i}^{\prime},\,t_{i})$.
It essentially maps the reduced density matrix of the initial configuration
at time $t_{i}$ to the reduced density matrix at final time $t_{f}$
according to 
\begin{equation}
\rho(\bm{\chi}_{f},\,\bm{\chi}_{f}^{\prime},\,t_{f})=\int\!d\bm{\chi}_{i}d\bm{\chi}_{i}^{\prime}\;J(\bm{\chi}_{f},\,\bm{\chi}_{f}^{\prime},\,t_{f};\bm{\chi}_{i},\,\bm{\chi}_{i}^{\prime},\,t_{i})\,\rho(\bm{\chi}_{i},\,\bm{\chi}_{i}^{\prime},\,t_{i})\,.\label{E:kfghyrd}
\end{equation}
where $\rho(\bm{\chi},\,\bm{\chi}^{\prime},\,t)\equiv\braket{\bm{\chi}\big|\hat{\rho}(t)\big|\bm{\chi}^{\prime}}$
is the elements of the reduced density matrix in the position representation.
With the help of the Feynman-Vernon influence functional formalism
and the Hubbard-Stratonovich transformation {[}see \ref{sec:FV-formalism}{]},
the propagating function ${J}(\bm{\chi}_{f},\,\bm{\chi}_{f}^{\prime},\,t_{f};\bm{\chi}_{i},\,\bm{\chi}_{i}^{\prime},\,t_{i})$
can be formally expressed as the ensemble average 
\begin{align}
J(\bm{\chi}_{f},\,\bm{\chi}_{f}^{\prime},\,t_{f};\bm{\chi}_{i},\,\bm{\chi}_{i}^{\prime},\,t_{i})=\int\!\mathcal{D}\bm{\xi}\;P[\bm{\xi}]\,\mathcal{J}_{\bm{\xi}}(\bm{r}_{f},\,\bm{q}_{f},\,t_{f};\,\bm{r}_{i},\,\bm{q}_{i},\,t_{i})\,,\label{eq:Jx}
\end{align}
where $P[\bm{\xi}]$ is the probability distribution functional of
the stochastic trajectory, which may be interpreted as the manifestation
of quantum fluctuations of the bath fields, and $\mathcal{J}_{\bm{\xi}}(\bm{r}_{f},\,\bm{q}_{f},\,t_{f};\,\bm{r}_{i},\,\bm{q}_{i},\,t_{i})$
is the propagating function for each realization of the stochastic
noise $\bm{\xi}$. Here we have introduced the center-of-mass coordinate
$\bm{r}$ and relative coordinate $\bm{q}$ by 
\begin{align}
\bm{q}(s) & =\bm{\chi}(s)-\bm{\chi}^{\prime}(s)\,, & \bm{r}(s) & =\frac{1}{2}[\bm{\chi}(s)+\bm{\chi}^{\prime}(s)]\,,
\end{align}
respectively. From now on, we will set $t_{i}=0$ without loss of
generality and suppress the argument $t_{i}$ in relevant quantities.

The stochastic trajectory $\bm{\xi}(s)$ satisfies the Gaussian statistics,
and its first two moments are 
\begin{align}
\langle\!\langle\xi_{k}(s)\rangle\!\rangle & =0\,, & \langle\!\langle\xi_{k}(s)\xi_{l}(s^{\prime})\rangle\!\rangle & =G_{H}^{(k)}(s,\,s^{\prime})\delta_{kl}\,,\label{eq:xik-xil-ave}
\end{align}
where $G_{H}^{(n)}(s,\,s^{\prime})$ is the noise kernel of the $n^{\text{th}}$
private thermal bath. The stochastic propagation function thus enables
us to introduce the stochastic density operator in the same fashion
as \eqref{E:kfghyrd} 
\begin{equation}
\varrho_{\bm{\xi}}(\bm{r}_{f},\,\bm{q}_{f},\,t_{f})\equiv\int\!d\bm{r}_{i}d\bm{q}_{i}\;\mathcal{J}_{\bm{\xi}}(\bm{r}_{f},\,\bm{q}_{f},\,t_{f};\,\bm{r}_{i},\,\bm{q}_{i})\,\varrho(\bm{r}_{i},\,\bm{q}_{i})\,,
\end{equation}
such that 
\begin{align}
\mathcal{J}(\bm{r}_{f},\,\bm{q}_{f},\,t_{f};\,\bm{r}_{i},\,\bm{q}_{i}) & =\langle\!\langle\mathcal{J}_{\bm{\xi}}(\bm{r}_{f},\,\bm{q}_{f},\,t_{f};\,\bm{r}_{i},\,\bm{q}_{i})\rangle\!\rangle\,,\\
\varrho(\bm{r}_{f},\,\bm{q}_{f},\,t_{f}) & =\langle\!\langle\varrho_{\bm{\xi}}(\bm{r}_{f},\,\bm{q}_{f},\,t_{f})\rangle\!\rangle\,,
\end{align}
where $\langle\!\langle\bullet\rangle\!\rangle$ denotes the ensemble
average over the probability functional $P[\bm{\xi}]$.

In the context of the thermal transport along the oscillator chain,
we are primarily concerned with the nature of the thermal flow along
the chain, and the energy exchange between the chain and the bathes.
In next section, we will discuss the physical quantities of our interest
for the current configuration.

\subsection{The energy exchange}

We use the master equation developed in \citep{StoMak}, which is
time-local for general system Hamiltonian, to define the energy exchange
between an oscillator and its private bath or its nearest neighbors.
The Hamiltonian operator of the chain, corresponding to the action
$S_{\chi}$, is 
\begin{equation}
\hat{H}_{\chi}=\sum_{n=1}^{N}\hat{H}_{n}+\sum_{n=1}^{N-1}\hat{V}_{n,\,n+1}\,,
\end{equation}
where the Hamiltonian operator $\hat{H}_{n}$ at site $n$ takes the
form 
\begin{equation}
\hat{H}_{n}=\frac{\hat{p}_{n}^{2}}{2m_{n}}+V(\hat{\chi}_{n})\,.\label{eq:H-n}
\end{equation}
For Ohmic baths, it can be shown that, the evolution of the chain
associated with each trajectory $\bm{\xi}$ can be described by the
stochastic Liouville-von Neumann equation \citep{StoMak} 
\begin{align}
\dot{\hat{\rho}}_{\bm{\xi}}(t_{f}) & =-i\bigl[\hat{H}_{\chi},\,\hat{\rho}_{\bm{\xi}}(t_{f})\bigr]+i\sum_{n=1}^{N}\xi_{n}(t_{f})\bigl[\hat{\chi}_{n},\,\hat{\rho}_{\bm{\xi}}(t_{f})\bigr]-i\sum_{n=1}^{N}\gamma_{n}\,\bigl[\hat{\chi}_{n},\,\{\hat{p}_{n},\,\hat{\rho}_{\bm{\xi}}(t_{f})\}\bigr]
\end{align}
where $\gamma_{n}=e_{n}^{2}/(8\pi m_{n})$. We will calculate $\operatorname{Tr}\{\dot{\hat{\rho}}_{\bm{\xi}}(t_{f})H_{n}\}/\operatorname{Tr}\{\hat{\rho}(t_{f})\}$,
which is the the change of internal energy~of oscillator $n$ associated
with the trajectory $\bm{\xi}$, explicitly to identify the all the
relevant heat exchanges at site $n$. To this end, we find 
\begin{align}
\operatorname{Tr}\Bigl\{\dot{\hat{\rho}}_{\bm{\xi}}(t_{f})\hat{H}_{n}\Bigr\} & =i\,\operatorname{Tr}\Bigl\{\bigl[\hat{H}_{\chi},\,\hat{H}_{n}\bigr]\,\hat{\rho}_{\bm{\xi}}(t_{f})\Bigr\}-i\,\xi_{n}(t_{f})\operatorname{Tr}\Bigl\{\bigl[\hat{\chi}_{n},\,\hat{H}_{n}\bigr]\,\hat{\rho}_{\bm{\xi}}(t_{f})\Bigr\}+i\,\gamma_{n}\operatorname{Tr}\Bigl\{\bigl\{\hat{p}_{n},\,\bigl[\hat{\chi}_{n},\,\hat{H}_{n}\bigr]\bigr\}\,\hat{\rho}_{\bm{\xi}}(t_{f})\Bigr\}\,,\label{eq:Heat}
\end{align}
where we have used the following identities 
\begin{align}
\operatorname{Tr}\Bigl\{[A,\,B]\,C\Bigr\} & =-\operatorname{Tr}\Bigl\{[A,\,C]\,B\Bigr\}\,, & \operatorname{Tr}\Bigl\{[C,\{A,\,B\}]\,D\Bigr\} & =-\operatorname{Tr}\Bigl\{\{A,\,[C,\,D]\}\,B\Bigr\}\,.
\end{align}
Carrying out the commutators and anti-commutators in Eq.~\eqref{eq:Heat},
we find 
\begin{align}
\operatorname{Tr}\Bigl\{\bigl[\hat{\chi}_{n},\,\hat{H}_{n}\bigr]\,\hat{\rho}_{\bm{\xi}}(t_{f})\Bigr\} & =\frac{i}{m_{n}}\operatorname{Tr}\Bigl\{\hat{p}_{n}\,\hat{\rho}_{\bm{\xi}}(t_{f})\Bigr\}\,,\label{eq:comm1}\\
\operatorname{Tr}\Bigl\{\bigl\{\hat{p}_{n},\,\bigl[\hat{\chi}_{n},\,H_{n}\bigr]\bigr\}\,\hat{\rho}_{\bm{\xi}}(t_{f})\Bigr\} & =\frac{2i}{m_{n}}\operatorname{Tr}\Bigl\{\hat{p}_{n}^{2}\,\hat{\rho}_{\bm{\xi}}(t_{f})\Bigr\}\,,\\
\operatorname{Tr}\Bigl\{[\hat{H}_{\chi},\,\hat{H}_{n}]\,\hat{\rho}_{\bm{\xi}}(t_{f})\Bigr\} & =\frac{1}{2m_{n}}\operatorname{Tr}\Bigl\{\bigl[\hat{V}_{n,\,n-1},\,\hat{p}_{n}^{2}\bigr]\,\hat{\rho}_{\bm{\xi}}(t_{f})\Bigr\}+\frac{1}{2m_{i}}\operatorname{Tr}\Bigl\{\bigl[\hat{V}_{n,\,n+1},\,\hat{p}_{n}^{2}\bigr]\,\hat{\rho}_{\bm{\xi}}(t_{f})\Bigr\}\,.\label{eq:comm3}
\end{align}
Substituting Eqs. (\ref{eq:comm1}-\ref{eq:comm3}) into Eq.~(\ref{eq:Heat})
yields 
\begin{equation}
\frac{1}{\mathcal{Z}}\,\operatorname{Tr}\Bigl\{\dot{\hat{\rho}}_{\bm{\xi}}(t_{f})\hat{H}_{n}\Bigr\}=P_{\xi_{n}}(t_{f},\,\bm{\xi}]+P_{\gamma_{n}}(t_{f},\,\bm{\xi}]+P_{n-1\to n}(t_{f},\,\bm{\xi}]+P_{n+1\to n}(t_{f},\,\bm{\xi}]\label{eq:dEdt}
\end{equation}
where $\mathcal{Z}\equiv\operatorname{Tr}\bigl\{\hat{\rho}(t_{f})\bigr\}$.
We can identify 
\begin{equation}
P_{\xi_{n}}(t_{f},\,\bm{\xi}]\equiv\mathcal{Z}^{-1}\frac{\xi_{n}(t_{f})}{m_{n}}\,\operatorname{Tr}\Bigl\{\hat{p}_{n}\,\hat{\rho}_{\bm{\xi}}(t_{f})\Bigr\}\,,\label{eq:P-xi-traj}
\end{equation}
as the energy flow into the oscillator $n$ from its private bath
and 
\begin{equation}
P_{\gamma_{n}}(t_{f},\,\bm{\xi}]\equiv-\mathcal{Z}^{-1}\frac{2\gamma_{n}}{m_{n}}\,\operatorname{Tr}\Bigl\{\hat{p}_{n}^{2}\hat{\rho}_{\bm{\xi}}(t_{f})\Bigr\}\,,\label{eq:P-gamma-traj}
\end{equation}
the rate of the energy dissipated back to the bath. We note that Eqs.
(\ref{eq:P-xi-traj}, \ref{eq:P-gamma-traj}) coincides with the definition
of the oscillator-bath heat exchange in ~\citep{HHAoP} through the
semiclassical Langevin equation. The remaining two expressions in
\eqref{eq:dEdt} describe the energy exchange between the neighboring
oscillators 
\begin{align}
P_{n-1\to n}(t_{f},\,\bm{\xi}] & \equiv-\mathcal{Z}^{-1}\frac{i}{2m_{n}}\,\operatorname{Tr}\Bigl\{\bigl[\hat{p}_{n}^{2},\,\hat{V}_{n,\,n-1}\bigr]\,\hat{\rho}_{\bm{\xi}}(t_{f})\Bigr\}\,,\label{eq:Pn-left}\\
P_{n+1\to n}(t_{f},\,\bm{\xi}] & \equiv-\mathcal{Z}^{-1}\frac{i}{2m_{n}}\,\operatorname{Tr}\Bigl\{\bigl[\hat{p}_{n}^{2},\,\hat{V}_{n,\,n+1}\bigr]\,\hat{\rho}_{\bm{\xi}}(t_{f})\Bigr\}\,.\label{eq:Pn-right}
\end{align}
Since $\operatorname{Tr}\{\dot{\hat{\rho}}_{\bm{\xi}}(t_{f})\hat{H}_{n}\}/\mathcal{Z}$
is the change of the trajectory-wise internal energy~of oscillator
$n$, we then interpret $P_{n-1\to n}(t_{f},\,\bm{\xi}]$ as the power
delivered to oscillator $n$ by oscillator $n-1$ to it, while $P_{n+1\to n}(t_{f},\,\bm{\xi}]$
as the power by oscillator $n+1$. The mutual coupling between the
oscillators consists of various components as shown in \eqref{eq:V-n},
it proves convenient to further decompose the energy exchange between
the oscillators into the sum of the contributions from these components
\begin{equation}
P_{\nu\to n}(t_{f},\,\bm{\xi}]=\sum_{\eta=2}^{4}P_{\nu\to n}^{(\eta)}(t_{f},\,\bm{\xi}]\,,
\end{equation}
where $\nu=n\pm1$ and the contribution from each components is defined
by 
\begin{equation}
P_{\nu\to n}^{(\eta)}(t_{f},\,\bm{\xi}]\equiv-\mathcal{Z}^{-1}\frac{\lambda_{\eta}}{\eta}\frac{i}{2m_{n}}\,\operatorname{Tr}\Bigl\{\bigl[\hat{p}_{n}^{2},\,\bigl(\hat{\chi}_{n}-\hat{\chi}_{\nu}\bigr)^{\eta}\bigr]\,\hat{\rho}_{\bm{\xi}}(t_{f})\Bigr\}\,.\label{eq:Pn-eta-nu}
\end{equation}
At this moment, for each trajectory $\bm{\xi}$, we have denoted the
thermodynamic quantities $O$ by the expectation value $O(t_{f},\,\bm{\xi}]$,
which could be, e.g., the heat exchange between the bath and the primary
system, heat flow inside the primary system, local kinetic and potential
energy, etc. Their nonequilibrium dynamics can be investigated once
we get hold of the stochastic density matrix $\hat{\rho}_{\bm{\xi}}(t)$
at any moment $t$. However, $O(t_{f},\,\bm{\xi}]$ will not correspond
to the observable we measure. It is too fine-grained, corresponding
to each realization of the stochastic noise $\bm{\xi}$, and we still
need to perform the appropriate ensemble average. For example, in
Eqs.~\eqref{eq:P-xi-traj}, \eqref{eq:P-gamma-traj}, and \eqref{eq:Pn-right},
we have used the first and the second moments of the momentum operator
at the $n^{\text{th}}$ site 
\begin{align}
p_{n}(t_{f},\,\bm{\xi}] & =\operatorname{Tr}\Bigl\{\hat{p}_{n}\,\hat{\rho}_{\bm{\xi}}(t_{f})\Bigr\}\,, & p_{n}^{2}(t_{f},\,\bm{\xi}] & =\operatorname{Tr}\Bigl\{\hat{p}_{n}^{2}\,\hat{\rho}_{\bm{\xi}}(t_{f})\Bigr\}\,,\label{eq:pn2}
\end{align}
and the quantity 
\begin{equation}
w_{\nu\to n}^{(\eta)}(t_{f},\,\bm{\xi}]=-\frac{i}{2\eta}\operatorname{Tr}\Bigl\{\bigl[\hat{p}_{n}^{2},\,\bigl(\hat{\chi}_{n}-\hat{\chi}_{\nu}\bigr)^{\eta}\bigr)]\,\hat{\rho}_{\bm{\xi}}(t_{f})\Bigr\}\,.\label{eq:wn-nu-eta}
\end{equation}
Thus after performing the ensemble average introduced in \eqref{eq:Jx},
we can write the averaged \textcolor{purple}{heat current} respectively
associated with \eqref{eq:P-xi-traj}, \eqref{eq:P-gamma-traj}, and
\eqref{eq:Pn-right} as 
\begin{align}
P_{\xi_{n}}(t_{f}) & =\mathcal{Z}^{-1}\frac{1}{m_{n}}\,\langle\!\langle\xi_{n}(t_{f})\,p_{n}(t_{f},\,\bm{\xi}]\rangle\!\rangle\,,\label{eq:P-xi-ave}\\
P_{\gamma_{n}}(t_{f}) & =-\mathcal{Z}^{-1}\frac{2\gamma_{n}}{m_{n}}\,\langle\!\langle p_{n}^{2}(t_{f},\,\bm{\xi}]\rangle\!\rangle\,,\label{eq:P-gamma-ave}\\
P_{\nu\to n}^{(\eta)}(t_{f}) & =\mathcal{Z}^{-1}\frac{\lambda_{\eta}}{m_{n}}\,\langle\!\langle w_{\nu\to n}^{(\eta)}(t_{f},\,\bm{\xi}]\rangle\!\rangle\,.\label{eq:P-nu-eta-ave}
\end{align}

Now with the functional method introduced in ~\ref{sec:SPF}, the
results we obtained earlier can be concisely expressed as the functional
derivatives of a generating functional 
\begin{align}
p_{n}(t_{f},\,\bm{\xi}] & =-i\frac{\partial}{\partial h_{n}}\mathcal{Z}_{\bm{\xi}}[\bm{j},\bm{j}^{\prime},\bm{h})\,\bigg|_{\bm{j}=\bm{j}^{\prime}=\bm{h}=0}\label{eq:pn-Z}\\
p_{n}^{2}(t_{f},\,\bm{\xi}] & =-\frac{\partial^{2}}{\partial h_{n}^{2}}\mathcal{Z}_{\bm{\xi}}[\bm{j},\bm{j}^{\prime},\bm{h})\,\bigg|_{\bm{j}=\bm{j}^{\prime}=\bm{h}=0}\label{eq:pn2-Z}\\
w_{\nu\to n}^{(\eta)}(t_{f},\,\bm{\xi}] & =-\bigl(-i\bigr)^{\eta}\frac{\partial}{\partial h_{n}}\left[\frac{\delta}{\delta j_{n}(t_{f})}-\frac{\delta}{\delta j_{\nu}(t_{f})}\right]^{\eta-1}\mathcal{Z}_{\bm{\xi}}[\bm{j},\bm{j}^{\prime},\bm{h})\,\bigg|_{\bm{j}=\bm{j}^{\prime}=\bm{h}=0}\,,\label{eq:wn-nu-eta-Z}
\end{align}
and in particular, $\operatorname{Tr}\bigl\{\rho(t_{f})\bigr\}=\langle\!\langle\mathcal{Z}_{\bm{\xi}}[0]\rangle\!\rangle$,
where the generating functional $\mathcal{Z}_{\bm{\xi}}[\bm{j},\,\bm{j}^{\prime},\bm{h})$
of central importance is defined by 
\begin{equation}
\mathcal{Z}_{\bm{\xi}}[\bm{j},\,\bm{j}^{\prime},\bm{h})=\int\!d\bm{r}_{i}d\bm{q}_{i}\int\!d\bm{r}_{f}\int_{\bm{r}(0)=\bm{r}_{i}}^{\bm{r}(t_{f})=\bm{r}_{f}}\!\mathcal{D}\bm{r}\int_{\bm{q}(0)=\bm{q}_{i}}^{\bm{q}(t_{f})=\bm{h}}\!\mathcal{D}\bm{q}\;e^{i\,\mathcal{S}_{\bm{\xi}}[\bm{r},\,\bm{q};\,\bm{j},\,\bm{j}^{\prime}]}\varrho(\bm{r}_{i},\,\bm{q}_{i})\,,\label{eq:Z}
\end{equation}
with the action given by 
\begin{equation}
\mathcal{S}_{\bm{\xi}}[\bm{r},\,\bm{q};\,\bm{j},\,\bm{j}^{\prime}]=\mathcal{S}_{\bm{\xi}}[\bm{r},\,\bm{q}]+\int_{0}^{t_{f}}\!ds\;\Bigl\{\bm{j}(s)\cdot\bm{r}(s)+\bm{j}^{\prime}(s)\cdot\bm{q}(s)\Bigr\}\,.\label{eq:s-xi}
\end{equation}
We have introduced the shorthand notation $\mathcal{Z}_{\bm{\xi}}[0]$
for $\mathcal{Z}_{\bm{\xi}}[\bm{j}=0,\bm{j}^{\prime}=0,\bm{h}=0)$.
Two external, nondynamical currents $\bm{j}$, $\bm{j}'$ and one
external parameter $\bm{h}$ are introduced in the context of the
in-in formalism to facilitate the functional method.

As of now, we have kept the discussions very general: (a) We allow
general nonlinearities in the chain. Nonetheless, when nonlinearities
are presented, the generating functional cannot be evaluated exactly
and the perturbative treatment is needed, and (b) We allow each oscillator
linearly coupled their own private linear bath, but do not restrict
the coupling strength. In what follows, we shall carry out the detailed
calculations specified by Eqs. (\ref{eq:P-xi-ave}-\ref{eq:wn-nu-eta-Z})
and address the existence of the NESS across a nonlinear chain made
of two nonlinear oscillators, which are put in contact with their
private thermal baths.

\section{\label{sec:zeroth-order} The zeroth-order of the heat exchange between
an oscillator and their private baths}

With the functional perturbative approach sketched in Sec. \ref{sec:The-general-formalsim},
we are in a position to compute the heat flow across the nonlinear
chain in the weak nonlinearity regime. For the sake of simplicity,
we will demonstrate in detail only for the case of two anharmonic
oscillators. The argument generalizes to $N>2$ anharmonic oscillators.
We outline the steps as follows: 
\begin{enumerate}
\item We show in \ref{sec:Z0-SE} that thanks to the natures of the linear
dynamics at late time, $t_{f}\to\infty$, the stochastic generating
functional $\mathcal{Z}_{0\bm{\xi}}[\bm{j},\,\bm{j}^{\prime},\,\bm{h})$
for the linear chain reduces a simple form {[}see Eqs. (\ref{eq:Z0-xi}-\ref{eq:S0-frak}){]}.
These can be straightforwardly generalized to the case of $N>2$ linear
harmonic oscillators in the chain. One only needs to replace the causal
propagator corresponding to two-oscillator and two-bath configuration
with the one under investigation. 
\item With the help this simplified stochastic generating functional and
the Wick contraction, we compute the perturbative corrections for
the quantities of our interest at the late time to the first order
in the nonlinear coupling constants. 
\item We show that the contribution from the initial state of the chain
at the late time $t_{f}\to\infty$ vanishes due to the damping effect.
Furthermore, the internal energy of each oscillator is stationary
and therefore NESS may be argued, at least up to the first order in
nonlinear coupling constants. 
\end{enumerate}
We have already seen in Eqs. (\ref{eq:pn2-Z}-\ref{eq:wn-nu-eta-Z})
that the generating functional Eq.~\eqref{eq:Z} is of crucial importance
to obtain various thermodynamic quantities. To illustrate the stochastic
functional perturbative approach, we consider three types of nonlinearities:
$\alpha$-FPUT, $\beta$-FPUT and KG types of anharmonicity in the
chain, consisting of only two oscillators. We assume the oscillator-bath
coupling constants are equal $e_{1}=e_{2}=e$ and therefore the retarded
Green's functions of the two private baths take the same form, i.e.,
$G_{R}^{(1)}=G_{R}^{(2)}$. Furthermore, for the sake of simplicity,
we assume uniform mass distribution across the chain, so we set mass
of each oscillator to unity without loss of generality. We can always
recover the mass dependence in the relevant energy fluxes through
dimensional consideration or redefining parameters and displacements
in the Lagrangian.

In the context of the perturbative treatment, one splits the chain
Lagrangian into two parts: $L_{\chi}=L_{0}+L_{\textsc{NL}}$, where
$L_{0}$ describes the linearly coupled oscillators 
\begin{equation}
L_{0}[\bm{\chi}]=\frac{1}{2}\sum_{n=1}^{2}[\dot{\chi}_{n}^{2}-\omega_{0}^{2}\chi_{n}^{2}]-\frac{\lambda_{2}}{2}\bigl(\chi_{1}-\chi_{2}\bigr)^{2}\,,
\end{equation}
and $L_{\text{NL}}$ accounts for the nonlinear self- and intra-oscillator
couplings, 
\begin{equation}
L_{\textsc{NL}}[\bm{\chi}]=-\left[\frac{\lambda_{3}}{3}\bigl(\chi_{1}-\chi_{2}\bigr)^{3}+\frac{\lambda_{4}}{4}\bigl(\chi_{1}-\chi_{2}\bigr)^{4}+\frac{\lambda_{\textsc{kg}}}{4}\bigl(\chi_{1}^{4}+\chi_{2}^{4}\bigr)\right]\,.\label{E:fbrjht}
\end{equation}
In \eqref{E:fbrjht}, the first two terms on the righthand side describe
nonlinear intra-oscillator couplings: the cubic term for the $\alpha$-FPUT
coupling and the quartic term for the $\beta$-FPUT coupling. The
last term is the KG quartic on-site pin potential.

To the first order of the coupling constants, the stochastic propagating
function Eq.~\eqref{eq:Jxi} in this case generically takes the form
\begin{align}
\mathcal{J}_{\bm{\xi}}(\bm{r}_{f},\,\bm{h},\,t_{f};\,\bm{r}_{i},\,\bm{q}_{i};\,\bm{j},\,\bm{j}^{\prime}) & =\int_{\bm{r}(0)=\bm{r}_{i}}^{\bm{r}(t_{f})=\bm{r}_{f}}\!\mathcal{D}\bm{r}\int_{\bm{q}(0)=\bm{q}_{i}}^{\bm{q}(t_{f})=\bm{h}}\!\mathcal{D}\bm{q}\;\exp\Bigl\{ i\,\mathcal{S}_{\bm{\xi}}^{(0)}[\bm{r},\,\bm{q};\,\bm{j},\,\bm{j}^{\prime}]\Bigr\}\nonumber \\
 & \qquad\times\biggl\{1+i\sum_{klmr}\int_{0}^{t_{f}}ds\left[\sigma_{klmr}r_{k}(s)q_{l}(s)q_{m}(s)q_{r}(s)+\mu_{klmr}r_{k}(s)r_{l}(s)r_{m}(s)q_{r}(s)\right]\biggr.\nonumber \\
 & \qquad\qquad\quad\biggl.+\sum_{klm}\int_{0}^{t_{f}}ds\left[\sigma_{klm}r_{k}(s)q_{l}(s)q_{m}(s)+\mu_{klm}r_{k}(s)r_{l}(s)q_{m}(s)\right]\biggr\}\,,\label{eq:calJ-xi}
\end{align}
where the nonlinear coupling constants $\lambda_{\text{KG}}$, $\lambda_{4}$
and $\lambda_{3}$ have been absorbed into $\sigma_{klmr}$, $\mu_{klmr}$
and $\sigma_{klm}$ and $\mu_{klm}$ respectively. Thus the stochastic
generating functional Eq. \eqref{eq:Z} is then given by 
\begin{align}
\mathcal{Z}_{\bm{\xi}}[\bm{j},\,\bm{j}^{\prime},\,\bm{h}) & =\int\!d\bm{r}_{i}d\bm{q}_{i}\int\!d\bm{r}_{f}\;\mathcal{J}_{\bm{\xi}}(\bm{r}_{f},\,\bm{h},\,t_{f};\,\bm{r}_{i},\,\bm{q}_{i};\,\bm{j},\,\bm{j}^{\prime})\,\varrho(\bm{r}_{i},\,\bm{q}_{i})\,.\label{E:yrytd}
\end{align}
Since we will perform the perturbative calculations based on the zeroth-order
results, we will denote $\langle\!\langle\bullet\rangle\!\rangle_{0}$
as the ensemble average with respect to the zeroth-order stochastic
generating functional $Z_{\bm{\xi}}^{(0)}$, that is, the linear chain,
which takes the same form as \eqref{E:yrytd} with $\mathcal{J}_{\bm{\xi}}$
replaced by $\mathcal{J}_{\bm{\xi}}^{(0)}$. The latter does not contain
any nonlinear contribution in \eqref{eq:calJ-xi}. The stochastic
action $\mathcal{S}_{\bm{\xi}}^{(0)}$ for the linear chain is 
\begin{align}
\mathcal{S}_{\bm{\xi}}^{(0)}[\bm{r},\,\bm{q},\,\bm{j},\,\bm{j}^{\prime}] & =\int_{0}^{t_{f}}\!ds\,\biggl\{\dot{\bm{q}}^{T}(s)\cdot\dot{\bm{r}}(s)-\bm{q}^{T}(s)\cdot\bm{\Omega}^{2}\cdot\bm{r}(s)+\bm{q}^{T}(s)\cdot\bm{j}^{\prime}(s)+\bm{r}^{T}\cdot\bm{j}(\bm{s})\biggr.\nonumber \\
 & \qquad\qquad\qquad\qquad+\biggl.\bm{q}^{T}(s)\cdot\bm{\xi}(s)+\int_{0}^{t}\!ds^{\prime}\;\bm{q}^{T}(s)\bm{G}_{R}(s,\,s^{\prime})\bm{r}(s^{\prime})\biggr\}\,,\label{eq:S0SE}
\end{align}
where $\sqrt{\omega_{0}^{2}+\lambda_{2}}$ is understood as the bare
frequency of the linear oscillator, and $\lambda_{2}$ is the mutual
coupling between two oscillators, 
\begin{equation}
\bm{\Omega}^{2}\equiv\begin{bmatrix}\omega_{0}^{2}+\lambda_{2} & -\lambda_{2}\\
-\lambda_{2} & \omega_{0}^{2}+\lambda_{2}
\end{bmatrix}\,.\label{eq:Omega-bare}
\end{equation}
This action essentially describes the dynamics of the coupled, dissipative
oscillators, driven by the quantum fluctuations of their individual
baths~\citep{HHAoP}.

In Eq. \eqref{sec:Z0-SE}, we find in the limit $t_{f}\to\infty$,
the stochastic generating functional $\mathcal{Z}_{\bm{\xi}}^{(0)}$
can be factored into 
\begin{equation}
\mathcal{Z}_{\bm{\xi}}^{(0)}[\bm{j},\,\bm{j}^{\prime},\,\bm{h})=\mathscr{Z}^{(0)}[\varrho_{i},\,\bm{j}]\exp\Bigl\{ i\,\bar{\mathscr{S}}_{\pmb{\xi}}^{(0)}[\bm{j},\,\bm{h}]+i\,\bar{\mathfrak{S}}^{(0)}[\bm{j},\,\bm{j}^{\prime},\,\bm{h})\Bigr\},\label{eq:Z0-xi}
\end{equation}
according to their dependence on the initial state and the external
sources, where 
\begin{align}
\mathscr{Z}^{(0)}[\varrho_{i},\,\bm{j}] & =\int\!d\bm{r}_{i}\;\exp\Bigl\{ i\int_{0}^{t_{f}}\!ds\;\bm{j}^{T}(s)\cdot\bm{D}_{1}(s)\cdot\bm{r}_{i}\Bigr\}\,\varrho(\,\bm{r}_{i},\,\int_{0}^{t_{f}}\!ds\;\bm{D}_{2}(s)\cdot\bm{j}(s)\,)\,,\label{eq:Z0-scr}\\
\mathscr{\bar{S}}_{\bm{\xi}}^{(0)}[\bm{j},\,\bm{h}) & =\bm{h}^{T}\cdot\int_{0}^{t_{f}}\!\!ds\;\;\dot{\bm{D}}_{2}(t_{f}-s)\cdot\bm{\xi}(s)+\int_{0}^{t_{f}}\!ds\int_{0}^{t_{f}}\!ds^{\prime}\;\bm{j}^{T}(s)\cdot\bm{\mathfrak{D}}(s-s^{\prime})\cdot\bm{\xi}(s^{\prime})\,,\label{eq:S0cal-xi}\\
\bar{\mathfrak{S}}^{(0)}[\bm{j},\,\bm{j}^{\prime},\,\bm{h}] & =\int_{0}^{t_{f}}\!ds\int_{0}^{t_{f}}\!ds^{\prime}\;\bm{j}^{T}(s)\cdot\bm{\mathfrak{D}}(s-s^{\prime})\cdot\bm{j}^{\prime}(s^{\prime})+\bm{h}^{T}\cdot\int_{0}^{t_{f}}\!ds\;\dot{\bm{D}}_{2}(t_{f}-s)\cdot\bm{j}^{\prime}(s)\,.\label{eq:S0-frak}
\end{align}
The matrices $\bm{D}_{i}(s)$ are the homogenous solutions to Eq.~\eqref{eq:r-diff-gamma},
with the initial conditions $\bm{D}_{1}(0)=\mathbb{I}$, $\dot{\bm{D}}_{1}(0)=0$,
$\bm{D}_{2}(0)=0$ and $\dot{\bm{D}}_{2}(0)=\mathbb{I}$, and $\bm{\mathfrak{D}}(s)=\Theta(s)\bm{D}_{2}(s)$
is the causal (retarded) Green's function of the chain. In the case
of two oscillators, the Laplace transforms of $\bm{D}_{i}(s)$ are
given by Eqs.~\eqref{eq:D1-s-def} and \eqref{eq:D2-s-def}.

The first factor on the right hand side of Eq. \eqref{eq:Z0-xi} contains
the contributions due to the initial state of the coupled oscillators.
One would naturally ask whether the NESS thermodynamic quantities
and functional derivatives generated by Eq. (\ref{eq:Z0-xi}) depends
on the initial state? For linear chain or the zeroth-order of nonlinear
chain, all the heat fluxes are related to functional derivatives $\delta/\delta\bm{j}(t_{f})$
{[}see Eqs. (\ref{eq:pn-Z}-\ref{eq:wn-nu-eta-Z}){]}. According to
Theorem \ref{thm:dNZscrdNj} in \ref{sec:Functional derivatives},
in the late time limit, the NESS heat fluxes for the linear chain
or at the zeroth-order for nonlinear chain do not depend on the initial
state of the chain. For an anharmonic chain with the $\alpha$, $\beta$-FPUT
and KG types of nonlinearities, we will show in Sec. \ref{sec:The-first-order-corrections},
by the virtues of theorems in  \ref{sec:Functional derivatives},
at least to the first-order of the nonlinear coupling, the NESS exists,
and is independent of the initial states.

We conjecture that this conclusion holds to arbitrary order of the
nonlinear coupling constants for $\alpha$, $\beta$-FPUT and KG types
of nonlinearities, as long as the coupling is weak such that our perturbation
calculation is still valid. The reason is the following: Physically,
it is plausible that because of the damping, the initial information
about the chain will finally decay for an arbitrary order of the nonlinear
coupling constants. Mathematically, this implies that contribution
of the steady-state heat flow from the functional derivation of the
partition function $\mathscr{Z}^{(0)}[\varrho_{i},\,\bm{j}]$ {[}see
Eq. \eqref{eq:Z0-xi}{]} will eventually vanish, and requires that
the contribution from each higher order is smaller than that from
previous order to ensure the validity of the perturbative treatment.
The observation that the initial information is irrelevant in the
NESS effectively amounts to setting $\mathscr{Z}^{(0)}[\varrho_{i},\,\bm{j}=0]=1$
from the very beginning of the computation of the steady-state heat
exchange. That is to say, we can work with the new stochastic generating
functional $\check{\mathcal{Z}}_{\bm{\xi}}^{(0)}$, 
\begin{equation}
\check{\mathcal{Z}}_{\bm{\xi}}^{(0)}[\bm{j},\,\bm{j}^{\prime},\,\bm{h})=\exp\Bigl\{ i\,\mathscr{\bar{S}}_{\bm{\xi}}^{(0)}[\bm{j},\,\bm{h})+i\,\bar{\mathfrak{S}}^{(0)}[\bm{j},\,\bm{j}^{\prime},\,\bm{h})\Bigr\}\,,
\end{equation}
where the check mark indicates the initial state contribution is ignored.
It is then straightforward to compute the high order correction of
the steady-state heat flows across the chain. We speculate the heat
flow should satisfy their respective heat balance order by order for
a chain with either one of $\alpha$, $\beta$-FPUT and KG types of
nonlinearities.

Next we first apply the approach we have developed so far to the zeroth-order
contributions in the quantities defined in Eqs.~\eqref{eq:pn-Z}--\eqref{eq:wn-nu-eta-Z},
which corresponds to the heat exchange across the linear chain. In
this case the generating functional only contains one single term:
${\displaystyle \exp\bigl\{ i\,\bm{h}^{T}\cdot\int_{0}^{t_{f}}\!ds\;\dot{\bm{D}}_{2}(t_{f}-s)\cdot\bm{\xi}(s)\bigr\}}$.
Thus we find that for the $n^{\text{th}}$ oscillator, its first moment
of momentum at the zeroth order is 
\begin{align}
p_{n}^{(0)}(t_{f},\,\bm{\xi}] & =-i\,\frac{\partial}{\partial h_{n}}\mathcal{Z}_{\bm{\xi}}^{(0)}[\bm{j},\,\bm{j}^{\prime},\,\bm{h})\,\bigg|_{\bm{j}=\bm{j}^{\prime}=\bm{h}=0}=\int_{0}^{t_{f}}\!ds\;\Bigl[\dot{\bm{D}}_{2}(t_{f}-s)\cdot\bm{\xi}(s)\Bigr]_{n}\,,\label{eq:pn-zeroth-SPF}
\end{align}
so that the energy flow pumped in by its private bath is 
\begin{align}
P_{\xi_{n}}^{(0)}(t_{f}) & =\sum_{m=1}^{2}\int_{0}^{t_{f}}\!ds\;\Bigl[\dot{\bm{D}}_{2}(t_{f}-s)\Bigr]_{nm}\Bigl[\bm{G}_{H}(t_{f}-s)\Bigr]_{mn}\nonumber \\
 & =\int_{0}^{t_{f}}\!ds\;\Bigl[\dot{\bm{D}}_{2}(s)\cdot\bm{G}_{H}(s)\Bigr]_{nn}=\int_{-\infty}^{t_{f}}\!ds\;\Bigl[\dot{\bm{\mathfrak{D}}}_{2}(s)\cdot\bm{G}_{H}(s)\Bigr]_{nn}\,,\label{eq:Pxi0}
\end{align}
according to \eqref{eq:P-xi-ave} and \eqref{eq:pn-Z}. Similarly,
one can show the second moment of the momentum is given by 
\begin{align}
p_{n}^{2(0)}(t_{f},\,\bm{\xi}] & =-\frac{\partial^{2}Z_{0\bm{\xi}}[\bm{j}=0,\,\bm{j}^{\prime}=0,\,\bm{h})}{\partial h_{n}^{2}}\bigg|_{\bm{h}=0}\nonumber \\
 & =\int_{0}^{t_{f}}ds\int_{0}^{t_{f}}ds^{\prime}[\dot{\bm{D}}_{2}(s)\bm{G}_{H}(s^{\prime}-s)\dot{\bm{D}}_{2}(s^{\prime})]_{nn}\nonumber \\
 & =\int_{-\infty}^{t_{f}}\!ds\int_{-\infty}^{t_{f}}\!ds^{\prime}\;\Bigl[\dot{\bm{\mathfrak{D}}}(s)\cdot\bm{G}_{H}(s^{\prime}-s)\cdot\dot{\bm{\mathfrak{D}}}_{2}(s^{\prime})\Bigr]_{nn}\label{eq:pn2-zeroth}
\end{align}
The energy flow dissipated back into the thermal bath is 
\begin{equation}
P_{\gamma_{n}}^{(0)}(t_{f})\,=-2\gamma\int_{-\infty}^{t_{f}}\!ds\int_{-\infty}^{t_{f}}\!ds^{\prime}\;\Bigl[\dot{\bm{\mathfrak{D}}}(s)\cdot\bm{G}_{H}(s^{\prime}-s)\cdot\dot{\bm{\mathfrak{D}}}_{2}(s^{\prime})\Bigr]_{nn}.\label{eq:P-gamma-zeroth}
\end{equation}
For the energy flow between two oscillators at the zeroth order, we
only need to consider the $w_{\nu\to n}^{(0,\,2)}(t_{f},\,\bm{\xi}]$
term, 
\begin{equation}
w_{\nu\to n}^{(0,\,2)}(t_{f},\,\bm{\xi}]=\frac{\partial}{\partial h_{n}}\Bigl[\frac{\delta}{\delta j_{n}(t_{f})}-\frac{\delta}{\delta j_{\nu}(t_{f})}\Bigr]\,\mathcal{Z}_{\bm{\xi}}^{(0)}[\bm{j},\,\bm{j}^{\prime},\,\bm{h})\,\bigg|_{\bm{j}=\bm{j}^{\prime}=\bm{h}=0}\,.
\end{equation}
Here, we note that notational difference: From Eq. (\ref{eq:wn-nu-eta}),
the first superscript $a$ in $w_{\nu\to n}^{(a,\,\eta)}(t_{f},\,\bm{\xi}]$
denotes the perturbative order of the nonlinear coupling constants,
$\lambda_{3}$,$\lambda_{4}$, $\lambda_{\text{KG}}$, as one will
see in next section. From Eqs.~(\ref{eq:Z0-xi}-\ref{eq:S0-frak}),
we observe that only $\mathscr{Z}^{(0)}[\varrho_{i},\,\bm{j}]$ and
$\exp\bigl\{ i\,\bar{\mathscr{S}}_{\bm{\xi}}^{(0)}[\bm{j},\,\bm{h})\bigr\}$
in the generating functional contribute. Moreover, from the discussions
in Theorem~\ref{thm:dNZscrdNj}, in the late time limit, the only
non-vanishing results in $w_{\nu\to n}^{(0,\,2)}$ come from acting
the functional derivatives $\delta/\delta j_{n}(t_{f})$, $\delta/\delta j_{\nu}(t_{f})$
of $\bar{\mathscr{S}}_{0\bm{\xi}}[\bm{j},\,\bm{h})$. Therefore, we
obtain 
\begin{align}
w_{\nu\to n}^{(0,\,2)}(t_{f},\,\bm{\xi}] & =-\int_{0}^{t_{f}}ds[\dot{\bm{D}}_{2}(t_{f}-s)\bm{\xi}(s)]_{n}\int_{0}^{t_{f}}ds^{\prime}\left([\bm{\mathfrak{D}}(t_{f}-s^{\prime})\bm{\xi}(s^{\prime})]_{n}-[\bm{\mathfrak{D}}(t_{f}-s^{\prime})\bm{\xi}(s^{\prime})]_{\nu}\right)\nonumber \\
 & =-\sum_{kl}\int_{0}^{t_{f}}ds\int_{0}^{t_{f}}ds^{\prime}[\dot{\bm{D}}_{2}(t_{f}-s)]_{nk}[\bm{\xi}(s)]_{k}\left([\bm{\mathfrak{D}}(t_{f}-s^{\prime})]_{nl}[\bm{\xi}(s^{\prime})]_{l}-[\bm{\mathfrak{D}}(t_{f}-s^{\prime})]_{\nu l}\bm{\xi}(s^{\prime})]_{l}\right)
\end{align}
Using Eq. \eqref{eq:xik-xil-ave}, we find 
\begin{align}
w_{\nu\to n}^{(0,\,2)}(t_{f}) & =-\int_{0}^{t_{f}}ds\int_{0}^{t_{f}}ds^{\prime}\left([\dot{\bm{D}}_{2}(s)\bm{G}_{H}(s^{\prime}-s)\bm{\mathfrak{D}}(s^{\prime})]_{nn}-[\dot{\bm{D}}_{2}(s)\bm{G}_{H}(s^{\prime}-s)\bm{\mathfrak{D}}(s^{\prime})]_{n\nu}\right)\nonumber \\
 & =-\int_{-\infty}^{t_{f}}ds\int_{-\infty}^{t_{f}}ds^{\prime}\left([\dot{\bm{\mathfrak{D}}}(s)\bm{G}_{H}(s^{\prime}-s)\bm{\mathfrak{D}}(s^{\prime})]_{nn}-[\dot{\bm{\mathfrak{D}}}(s)\bm{G}_{H}(s^{\prime}-s)\bm{\mathfrak{D}}(s^{\prime})]_{n\nu}\right)\label{eq:W-n2-zeroth}
\end{align}
It can be straightforwardly shown that for $t_{f}\gg\gamma^{-1}$,
Eqs.~(\ref{eq:Pxi0}, \ref{eq:pn2-zeroth} \ref{eq:W-n2-zeroth})
approach time-independent constants because of damping. The diagrammatic
representation of Eqs.~(\ref{eq:Pxi0}, \ref{eq:pn2-zeroth} \ref{eq:W-n2-zeroth})
is given in Fig. \ref{fig:zeroth-order-time}. Further illustrations
of these diagrams will be discussed in Sec. \ref{sec:Diagrammatic-representations}

\begin{figure}
\begin{picture}(450,100) 
\begin{centering}
\put(0,0){\includegraphics[scale=0.75]{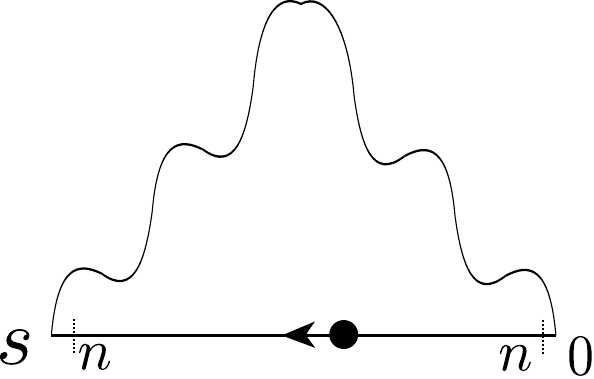}}\put(150,0){\includegraphics[scale=0.55]{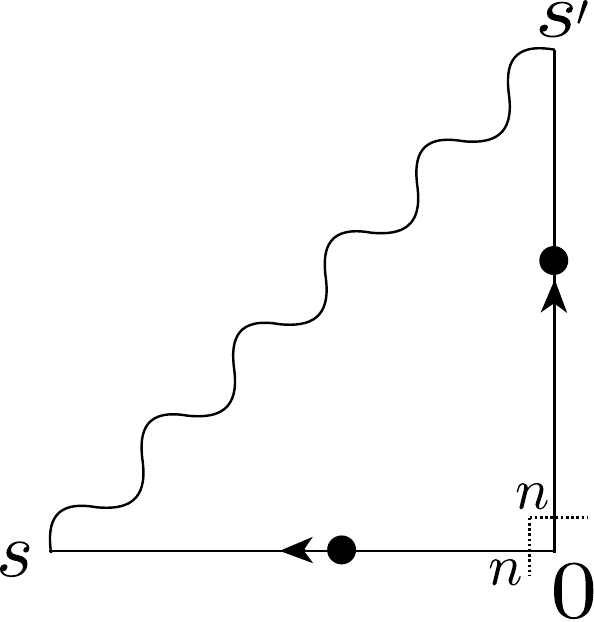}}\put(270,0){\includegraphics[scale=0.55]{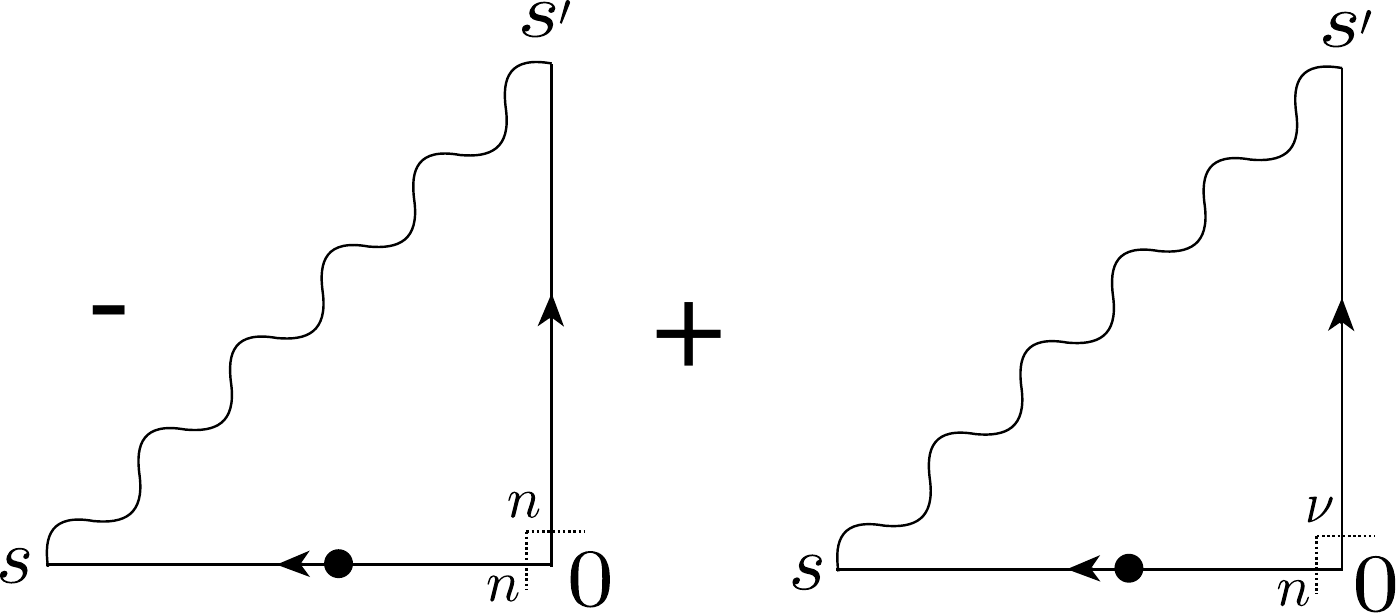}} 
\par\end{centering}
\put(50,30){\footnotesize{}(a)} \put(200,30){\footnotesize{}(b)}
\put(320,30){\footnotesize{}(c)} \put(440,30){\footnotesize{}(c)}
\end{picture}

\caption{\label{fig:zeroth-order-time}Time-domain diagrammatic representation
of the zeroth-order energy flux including (a) $P_{\xi_{n}}$ (b) $p_{n}^{2(0)}$
(c) $w_{\nu\to n}^{(0,\,2)}$. As we show Sec. \ref{subsec:Properties-of-causal-propagator},
the first digram in (c) actually vanishes in the late time limit.}
\end{figure}
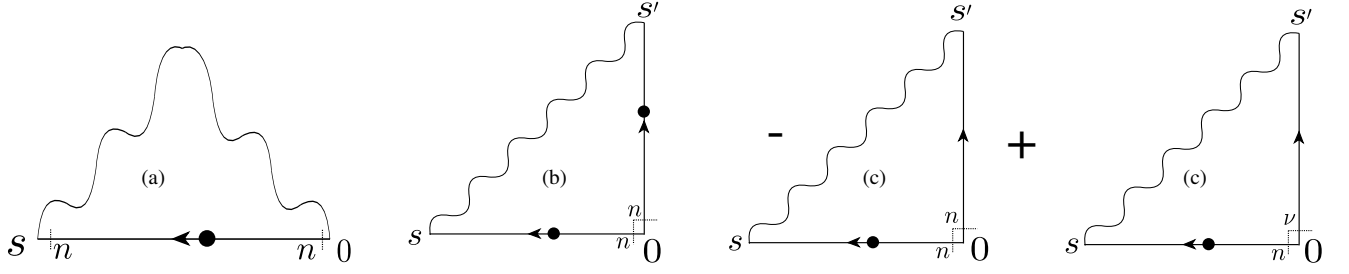

\section{\label{sec:The-first-order-corrections}The first-order corrections}

From \eqref{eq:calJ-xi}, we observe that for the first-order corrections,
the three types of nonlinearities contribute independently to the
first-order corrections, so in what follows, we will first discuss
the contributions from the $\beta$-FPUT and KG nonlinearities, and
then the $\alpha$-FPUT nonlinearity. In the end, we put them together
to form the overall first-order corrections, and examine the effects
of these nonlinearities on the heat transport in NESS.

\subsection{Stochastic generating functional $\mathcal{Z}_{\bm{\xi}}^{(1)}(0)$}

We start with the the first-order correction to the stochastic generating
functional. From Eq. (\ref{eq:calJ-xi}), we can easily show that
\begin{align}
\mathcal{Z}_{\bm{\xi}}^{(1)}[0] & =i\int_{0}^{t_{f}}\!ds\;\sum_{klmr}\Bigl[\sigma_{klmr}\,\langle r_{k}(s)q_{l}(s)q_{m}(s)q_{r}(s)\rangle_{0}+\mu_{klmr}\langle r_{k}(s)r_{l}(s)r_{m}(s)q_{r}(s)\rangle_{0}\Bigr]\,\bigg|_{\bm{j}=\bm{j}'=\bm{h}=0}\nonumber \\
 & \qquad\qquad+i\int_{0}^{t_{f}}\!ds\;\sum_{klm}\Bigl[\sigma_{klm}\langle r_{k}(s)q_{l}(s)q_{m}(s)\rangle_{0}+\mu_{klm}\langle r_{k}(s)r_{l}(s)q_{m}(s)\rangle_{0}\Bigr]\,\bigg|_{\bm{j}=\bm{j}'=\bm{h}=0}\,,
\end{align}
where 
\begin{align}
\langle r_{k}(s)q_{l}(s)q_{m}(s)\rangle_{0} & =i\,\frac{\delta^{3}\mathcal{Z}_{\bm{\xi}}^{(0)}[\bm{j},\,\bm{j}^{\prime},\,\bm{h})}{\delta j_{k}(s)\delta j_{l}^{\prime}(s)\delta j_{m}^{\prime}(s)}\,, & \langle r_{k}(s)q_{l}(s)q_{m}(s)q_{r}(s)\rangle_{0}\rangle_{0} & =\frac{\delta^{4}\mathcal{Z}_{\bm{\xi}}^{(0)}[\bm{j},\,\bm{j}^{\prime},\,\bm{h})}{\delta j_{k}(s)\delta j_{l}^{\prime}(s)\delta j_{m}^{\prime}(s)\delta j_{r}^{\prime}(s)}\,,\\
\langle r_{k}(s)r_{l}(s)q_{m}(s)\rangle_{0} & =i\,\frac{\delta^{3}\mathcal{Z}_{\bm{\xi}}^{(0)}[\bm{j},\,\bm{j}^{\prime},\,\bm{h})}{\delta j_{k}(s)\delta j_{l}(s)\delta j_{m}^{\prime}(s)}\,, & \langle r_{k}(s)r_{l}(s)r_{m}(s)q_{r}(s)\rangle_{0} & =\frac{\delta^{4}\mathcal{Z}_{\bm{\xi}}^{(0)}[\bm{j},\,\bm{j}^{\prime},\,\bm{h})}{\delta j_{k}(s)\delta j_{l}(s)\delta j_{m}(s)\delta j_{r}^{\prime}(s)}\,.
\end{align}
From Theorem \ref{thm:same-s} in  \ref{sec:Functional derivatives},
the above four ensemble averages all vanish, so the generating functional
$\mathcal{Z}_{\bm{\xi}}^{(1)}[0]$ does not have the first-order correction,
that is, $\mathcal{Z}_{\bm{\xi}}^{(1)}[0]=0$.

\subsection{Heat exchanges $P_{\xi_{n}}^{(1)}(t_{f})$, $P_{\gamma_{n}}^{(1)}(t_{f})$
between the oscillator and its private bath}

Here we will compute the first-order correction to the rate of the
energy exchange between the bath to the oscillator it belongs to.

\subsubsection{\label{subsec:Pxi-gamma-betaKG}KG and $\beta$-FPUT nonlinearities}

In \eqref{eq:P-xi-ave} and \eqref{eq:P-gamma-ave}, we have shown
that the first-order corrections of the rates of the energy exchange
$P_{\xi_{n}}^{(1)}(t_{f})$, $P_{\gamma_{n}}^{(1)}(t_{f})$ between
oscillator and its private bath, are proportional to $\langle\!\langle\xi_{n}(t_{f})p_{n}^{(1)}(t_{f},\,\bm{\xi}]\rangle\!\rangle$
and $\langle\!\langle p_{n}^{2(1)}(t_{f},\,\bm{\xi}]\rangle\!\rangle$,
respectively. From Eqs.~\eqref{eq:pn-Z} and \eqref{eq:pn2-Z}, the
first-order corrections of $p_{n}(t_{f},\,\bm{\xi}]$ and $p_{n}^{2(1)}(t_{f},\,\bm{\xi}]$
are found to be 
\begin{equation}
p_{n}^{(1)}(t_{f},\,\bm{\xi}]=\int_{0}^{t_{f}}ds\sum_{klmr}\left[\sigma_{klmr}\partial_{n}\langle r_{k}(s)q_{l}(s)q_{m}(s)q_{r}(s)\rangle_{0}+\mu_{klmr}\partial_{n}\langle r_{k}(s)r_{l}(s)r_{m}(s)q_{r}(s)\rangle_{0}\right]\,|_{\bm{j}=\bm{j}'=\bm{h}=0},
\end{equation}
\begin{equation}
p_{n}^{2(1)}(t_{f},\,\bm{\xi}]=-i\int_{0}^{t_{f}}ds\sum_{klmr}\left[\sigma_{klmr}\partial_{n}^{2}\langle r_{k}(s)q_{l}(s)q_{m}(s)q_{r}(s)\rangle_{0}+\mu_{klmr}\partial_{n}^{2}\langle r_{k}(s)r_{l}(s)r_{m}(s)q_{r}(s)\rangle_{0}\right]\,|_{\bm{j}=\bm{j}'=\bm{h}=0},
\end{equation}
where $\partial_{n}^{k}\equiv\partial^{k}/\partial h_{n}^{k}$ and
\begin{eqnarray}
\partial_{n}\langle r_{k}(s)q_{l}(s)q_{m}(s)q_{r}(s)\rangle_{0}|_{\bm{j}=\bm{j}^{\prime}=\bm{h}=0} & = & \frac{\partial\delta^{4}\mathcal{Z}_{0\bm{\xi}}[\bm{j},\,\bm{j}^{\prime},\,\bm{h})}{\partial h_{n}\delta j_{k}(s)\delta j_{l}^{\prime}(s)\delta j_{m}^{\prime}(s)\delta j_{r}^{\prime}(s)}\bigg|_{\bm{j}=\bm{j}^{\prime}=\bm{h}=0}\label{eq:dh-1r-3q}\\
\partial_{n}\langle r_{k}(s)r_{l}(s)r_{m}(s)q_{r}(s)\rangle_{0}|_{\bm{j}=\bm{j}^{\prime}=\bm{h}=0} & = & \frac{\partial\delta^{4}\mathcal{Z}_{0\bm{\xi}}[\bm{j},\,\bm{j}^{\prime},\,\bm{h})}{\partial h_{n}\delta j_{k}(s)\delta j_{l}(s)\delta j_{m}(s)\delta j_{r}^{\prime}(s)}\bigg|_{\bm{j}=\bm{j}^{\prime}=\bm{h}=0}\label{eq:dh-3r-1q}\\
\partial_{n}^{2}\langle r_{k}(s)q_{l}(s)q_{m}(s)q_{r}(s)\rangle_{0}|_{\bm{j}=\bm{j}^{\prime}=\bm{h}=0} & = & \frac{\partial\delta^{4}\mathcal{Z}_{0\bm{\xi}}[\bm{j},\,\bm{j}^{\prime},\,\bm{h})}{\partial h_{n}^{2}\delta j_{k}(s)\delta j_{l}^{\prime}(s)\delta j_{m}^{\prime}(s)\delta j_{r}^{\prime}(s)}\bigg|_{\bm{j}=\bm{j}^{\prime}=\bm{h}=0}\label{eq:d2h-1r-3q}\\
\partial_{n}^{2}\langle r_{k}(s)r_{l}(s)r_{m}(s)q_{r}(s)\rangle_{0}|_{\bm{j}=\bm{j}^{\prime}=\bm{h}=0} & = & \frac{\partial^{2}\delta^{4}\mathcal{Z}_{0\bm{\xi}}[\bm{j},\,\bm{j}^{\prime},\,\bm{h})}{\partial^{2}h_{n}\delta j_{k}(s)\delta j_{l}(s)\delta j_{m}(s)\delta j_{r}^{\prime}(s)}\bigg|_{\bm{j}=\bm{j}^{\prime}=\bm{h}=0}\label{eq:d2h-3r-1q}
\end{eqnarray}
Not all of the terms in the integrand will contribute. From Theorem
\ref{thm:same-s}, one can easily see that Eqs. (\ref{eq:dh-1r-3q},
\ref{eq:d2h-1r-3q}) vanish. Hence, in order to find $\langle\!\langle\xi_{n}(t_{f})p_{n}^{(1)}(t_{f},\,\bm{\xi}]\rangle\!\rangle$
and $\langle\!\langle p_{n}^{2(1)}(t_{f},\,\bm{\xi}]\rangle\!\rangle$,
we evaluate the following expressions 
\begin{align}
\Gamma_{klm}^{nr}(t_{f}) & \equiv\lim_{t_{f}\to\infty}\int_{0}^{t_{f}}\!ds\;\langle\!\langle\xi_{n}(t_{f})\,\partial_{n}\langle r_{k}(s)r_{l}(s)r_{m}(s)q_{r}(s)\rangle_{0}\rangle\!\rangle\,|_{\bm{h}=0}\,,\label{eq:Gamma-nr-klm-def}\\
\tilde{\Gamma}_{klm}^{nr}(t_{f}) & \equiv-i\lim_{t_{f}\to\infty}\int_{0}^{t_{f}}\!ds\;\langle\!\langle\partial_{n}^{2}\langle r_{k}(s)r_{l}(s)r_{m}(s)q_{r}(s)\rangle_{0}\rangle\!\rangle\,|_{\bm{h}=0}\,.\label{eq:Gamma-tild-nr-klm-def}
\end{align}
Here to illustrate the implementation of the functional method, we
will present the calculations in greater details.

Let us compute $\Gamma_{klm}^{nr}(t_{f})$ first. By Eqs.~\eqref{eq:Z0-xi}--\eqref{eq:S0-frak},
it is straightforward to find 
\begin{equation}
\frac{\partial\delta\mathcal{Z}_{\bm{\xi}}^{(0)}[\bm{j},\,\bm{j}^{\prime},\,\bm{h})}{\partial h_{n}\delta j_{r}^{\prime}(s)}\bigg|_{\bm{j}^{\prime}=\bm{h}=0}=i\,\Bigl[\dot{\bm{D}}_{2}(t_{f}-s)\Bigr]_{nr}\mathscr{Z}^{(0)}[\varrho_{i},\,\bm{j}]\exp\Bigl\{ i\,\bar{\mathscr{S}}_{\bm{\xi}}^{(0)}[\bm{j},\,\bm{h}=0)\Bigr\}\,.
\end{equation}
Thus Eq. (\ref{eq:Gamma-nr-klm-def}) becomes 
\begin{equation}
\Gamma_{klm}^{nr}(t_{f})=i\int_{0}^{t_{f}}\!ds\;\langle\!\langle\Bigl[\dot{\bm{D}}_{2}(t_{f}-s)\Bigr]_{nr}\,\xi_{n}(t_{f})\,\frac{\delta^{3}\mathscr{Z}^{(0)}[\varrho_{i},\,\bm{j}]\exp\Bigl\{ i\,\bar{\mathscr{S}}_{\bm{\xi}}^{(0)}[\bm{j},\,\bm{h})\Bigr\}}{\delta j_{k}(s)\delta j_{l}(s)\delta j_{m}(s)}\rangle\!\rangle\,\bigg|_{\bm{j}=\bm{h}=0}
\end{equation}
Theorem \ref{thm:decaying} in \ref{sec:Functional derivatives} shows
that the only nonzero contribution in the late time limit $t_{f}\to\infty$
comes from the case when all three functional derivatives acts on
$\exp\bigl\{ i\,\bar{\mathscr{S}}_{\bm{\xi}}^{(0)}[\bm{j},\,\bm{h})\bigr\}$
only. Therefore, we arrive at 
\begin{align}
\Gamma_{klm}^{nr}(t_{f})=\sum_{pqj}\dotsint_{0}^{t_{f}}\!dsds_{1}ds_{2}ds_{3}\; & \langle\!\langle[\dot{\bm{D}}_{2}(t_{f}-s)]_{nr}\xi_{n}(t_{f})[\bm{\mathfrak{D}}(s-s_{1})]_{kp}\xi_{p}(s_{1})[\bm{\mathfrak{D}}(s-s_{2})]_{lq}\xi_{q}(s_{2})\nonumber \\
 & \qquad\times[\bm{\mathfrak{D}}(s-s_{3})]_{mj}\xi_{j}(s_{3})\rangle\!\rangle dsds_{1}ds_{2}ds_{3}.\label{eq:Gamma-nr-klm}
\end{align}
The Wick's theorem gives 
\begin{align}
\langle\!\langle\xi_{n}(t_{f})\xi_{p}(s_{1})\xi_{q}(s_{2})\xi_{j}(s_{3})\rangle\!\rangle & =[\bm{G}_{H}(t_{f}-s_{1})]_{np}[\bm{G}_{H}(s_{2}-s_{3})]_{qj}+[\bm{G}_{H}(t_{f}-s_{2})]_{nq}[\bm{G}_{H}(s_{1}-s_{3})]_{pj}\nonumber \\
 & +[\bm{G}_{H}(t_{f}-s_{3})]_{nj}[\bm{G}_{H}(s_{1}-s_{2})]_{pq}
\end{align}
and then \eqref{eq:Gamma-nr-klm} becomes 
\begin{align}
\Gamma_{klm}^{nr}(t_{f}) & =\dotsint_{0}^{t_{f}}dsds_{1}ds_{2}ds_{3}[\dot{\bm{D}}_{2}(s)]_{nr}\left\{ [\bm{G}_{H}(s_{1})\bm{\mathfrak{D}}(s_{1}-s)]_{nk}[\bm{\mathfrak{D}}(s_{2}-s)\bm{G}_{H}(s_{3}-s_{2})\bm{\mathfrak{D}}(s_{3}-s)]_{lm}\right.\nonumber \\
 & +[\bm{G}_{H}(s_{2})\bm{\mathfrak{D}}(s_{2}-s)]_{nl}[\bm{\mathfrak{D}}(s_{1}-s)\bm{G}_{H}(s_{3}-s_{2})\bm{\mathfrak{D}}(s_{3}-s)]_{km}\nonumber \\
 & \left.+[\bm{G}_{H}(s_{3})\bm{\mathfrak{D}}(s_{3}-s)]_{nm}[\bm{\mathfrak{D}}(s_{1}-s)\bm{G}_{H}(s_{2}-s_{1})\bm{\mathfrak{D}}(s_{2}-s)]_{kl}\right\} \label{eq:Gamma-nr-klm2}
\end{align}
after we make the changes of variables $t_{f}-s\to s$, $t_{f}-s^{\prime}\to s^{\prime}$,
and $t_{f}-s_{i}\to s_{i}$. We use the short hand notation $\left[\cdots\right]$
to represent each of the three terms in the integrand in Eq. \eqref{eq:Gamma-nr-klm2}.
Because of the property $\bm{\mathfrak{D}}(s_{i}-s)=0$ for $s_{i}\le s$,
we may write 
\begin{equation}
\dotsint_{0}^{t_{f}}\left[\cdots\right]dsds_{1}ds_{2}ds_{3}=\int_{0}^{t_{f}}ds\int_{-\infty}^{t_{f}}ds_{1}\int_{0}^{t_{f}}ds_{2}\int_{0}^{t_{f}}ds_{3}\left[\cdots\right]=\int_{0}^{t_{f}}ds\int_{-\infty}^{t_{f}}\prod_{i=1}^{3}ds_{i}\left[\cdots\right].\label{eq:integration-limit}
\end{equation}
Note that the lower limit of the integration over $s_{i},\,(i=1,\,2,\,3)$
in Eq. \eqref{eq:Gamma-nr-klm2} is zero while in Eq. \eqref{eq:integration-limit}
is $-\infty$. Finally, one can change the lower limit for the $s$
integral in Eq. \eqref{eq:Gamma-nr-klm2} from $0$ to $-\infty$
by replacing $[\dot{\bm{D}}_{2}(s)]_{nr}$ with $[\bm{\mathfrak{D}}(s)]_{nr}$.
Hence it leads to, 
\begin{align}
\Gamma_{klm}^{nr}(t_{f}) & =\dotsint_{-\infty}^{t_{f}}dsds_{1}ds_{2}ds_{3}[\dot{\bm{\mathfrak{D}}}(s)]_{nr}\left\{ [\bm{G}_{H}(s_{1})\bm{\mathfrak{D}}(s_{1}-s)]_{nk}[\bm{\mathfrak{D}}(s_{2}-s)\bm{G}_{H}(s_{3}-s_{2})\bm{\mathfrak{D}}(s_{3}-s)]_{lm}\right.\nonumber \\
 & +[\bm{G}_{H}(s_{2})\bm{\mathfrak{D}}(s_{2}-s)]_{nl}[\bm{\mathfrak{D}}(s_{1}-s)\bm{G}_{H}(s_{3}-s_{2})\bm{\mathfrak{D}}(s_{3}-s)]_{km}\nonumber \\
 & \left.+[\bm{G}_{H}(s_{3})\bm{\mathfrak{D}}(s_{3}-s)]_{nm}[\bm{\mathfrak{D}}(s_{1}-s)\bm{G}_{H}(s_{2}-s_{1})\bm{\mathfrak{D}}(s_{2}-s)]_{kl}\right\} \label{eq:Gamma-nr-klm-final}
\end{align}

Next we compute $\tilde{\Gamma}_{klm}^{nr}(t_{f})$. Since 
\begin{equation}
\frac{\partial^{2}\delta\mathcal{Z}_{\bm{\xi}}^{(0)}[\bm{j},\,\bm{j}^{\prime},\,\bm{h})}{\partial h_{n}^{2}\delta j_{r}^{\prime}(s)}\bigg|_{\bm{j}^{\prime}=\bm{h}=0}=-2\sum_{p}\int_{0}^{t_{f}}\!ds'\;\Bigl[\dot{\bm{D}}_{2}(t_{f}-s)\Bigr]_{nr}\Bigl[\dot{\bm{D}}_{2}(t_{f}-s^{\prime})\Bigr]_{np}\,\xi_{p}(s^{\prime})\mathscr{Z}^{(0)}[\varrho_{i},\,\bm{j}]\exp\Bigl\{ i\,\bar{\mathscr{S}}_{\bm{\xi}}^{(0)}[\bm{j},\,\bm{h}=0)\Bigr\}
\end{equation}
where the factor $2$ comes from the fact the derivatives $\partial h_{n}$
can be applied to either $\bar{\mathscr{S}}_{\bm{\xi}}^{(0)}[\bm{j},\,\bm{h})$
or $\mathscr{Z}^{(0)}[\varrho_{i},\,\bm{j}]$. Then we find 
\begin{equation}
\tilde{\Gamma}_{klm}^{nr}(t_{f})=2i\sum_{p}\int_{0}^{t_{f}}\!ds\int_{0}^{t_{f}}\!ds'\;\langle\!\langle\Bigl[\dot{\bm{D}}_{2}(t_{f}-s)\Bigr]_{nr}\Bigl[\dot{\bm{D}}_{2}(t_{f}-s^{\prime})\Bigr]_{np}\,\xi_{p}(s^{\prime})\,\frac{\delta^{3}\mathscr{Z}^{(0)}[\varrho_{i},\,\bm{j}]\exp\Bigl\{ i\,\bar{\mathscr{S}}_{\bm{\xi}}^{(0)}[\bm{j},\,\bm{h}=0)\Bigr\}}{\delta j_{k}(s)\delta j_{l}(s)\delta j_{m}(s)}\,\bigg|_{\bm{j}=0}\rangle\!\rangle
\end{equation}
Similar to the evaluation of $\Gamma_{klm}^{nr}(t_{f})$, Theorem
\ref{thm:decaying-2D-3j} in  \ref{sec:Functional derivatives} shows
that the only nonzero contribution in the late time limit $t_{f}\to\infty$
comes from the case when all three functional derivatives act only
on $\exp\bigl\{ i\,\bar{\mathscr{S}}_{\bm{\xi}}^{(0)}[\bm{j},\,\bm{h}=0)\bigr\}$.
Therefore, following the similar steps in Eqs. (\ref{eq:Gamma-nr-klm}-\ref{eq:Gamma-nr-klm-final}),
we find 
\begin{align}
\tilde{\Gamma}_{klm}^{nr}(t_{f})=2\dotsint_{-\infty}^{t_{f}} & \!dsds^{\prime}ds_{1}ds_{2}ds_{3}\;\Bigl[\dot{\bm{\mathfrak{D}}}(s)\Bigr]_{nr}\nonumber \\
 & \qquad\qquad\times\left\{ \Bigl[\dot{\bm{\mathfrak{D}}}(s^{\prime})\bm{G}_{H}(s_{1}-s')\bm{\mathfrak{D}}(s_{1}-s)\Bigr]_{nk}\Bigl[\bm{\mathfrak{D}}(s_{2}-s)\bm{G}_{H}(s_{3}-s_{2})\bm{\mathfrak{D}}(s_{3}-s)\Bigr]_{lm}\right.\nonumber \\
 & \qquad\qquad\;\;\,+\Bigl[\dot{\bm{\mathfrak{D}}}(s^{\prime})\bm{G}_{H}(s_{2}-s')\bm{\mathfrak{D}}(s_{2}-s)\Bigr]_{nl}\Bigl[\bm{\mathfrak{D}}(s_{1}-s)\bm{G}_{H}(s_{3}-s_{2})\bm{\mathfrak{D}}(s_{3}-s)\Bigr]_{km}\nonumber \\
 & \qquad\qquad\;\;\,+\left.\dot{\bm{\mathfrak{D}}}(s^{\prime})\bm{G}_{H}(s_{3}-s')\bm{\mathfrak{D}}(s_{3}-s)\Bigr]_{nm}\Bigl[\bm{\mathfrak{D}}(s_{1}-s)\bm{G}_{H}(s_{2}-s_{1})\bm{\mathfrak{D}}(s_{2}-s)\Bigr]_{kl}\right\} \,.\label{eq:Gamma-tilde-nr-klm-final}
\end{align}
In the late time limit $t_{f}\to\infty$, we then find that the first-order
corrections from the $\beta$- and KG-type of nonlinearities to the
rates of the energy exchange between the $n^{\text{th}}$ oscillator
and its bath are given by 
\begin{eqnarray}
P_{\xi_{n}}^{(1)}(\infty) & = & \sum_{klmr}\mu_{klmr}\Gamma_{klm}^{nr}(\infty),\label{eq:Pxi-first-order}\\
P_{\gamma_{n}}^{(1)}(\infty) & = & -2\gamma\sum_{klmr}\mu_{klmr}\tilde{\Gamma}_{klm}^{nr}(\infty).\label{eq:Pgamma-first-order}
\end{eqnarray}
Both heat current approach time-independent constants according to
Eq. (\ref{eq:Gamma-nr-klm-final}, \ref{eq:Gamma-tilde-nr-klm-final})
in the late time limit, which signals the existence of the NESS. The
diagrammatic representation of $\Gamma_{klm}^{nr}(\infty)$ and $\tilde{\Gamma}_{klm}^{nr}(\infty)$
is shown in Fig. \ref{fig:Gamma-GammaTilde-time}.

\begin{figure}
\begin{picture}(450,100) 
\begin{centering}
\put(0,0){\includegraphics[scale=0.4]{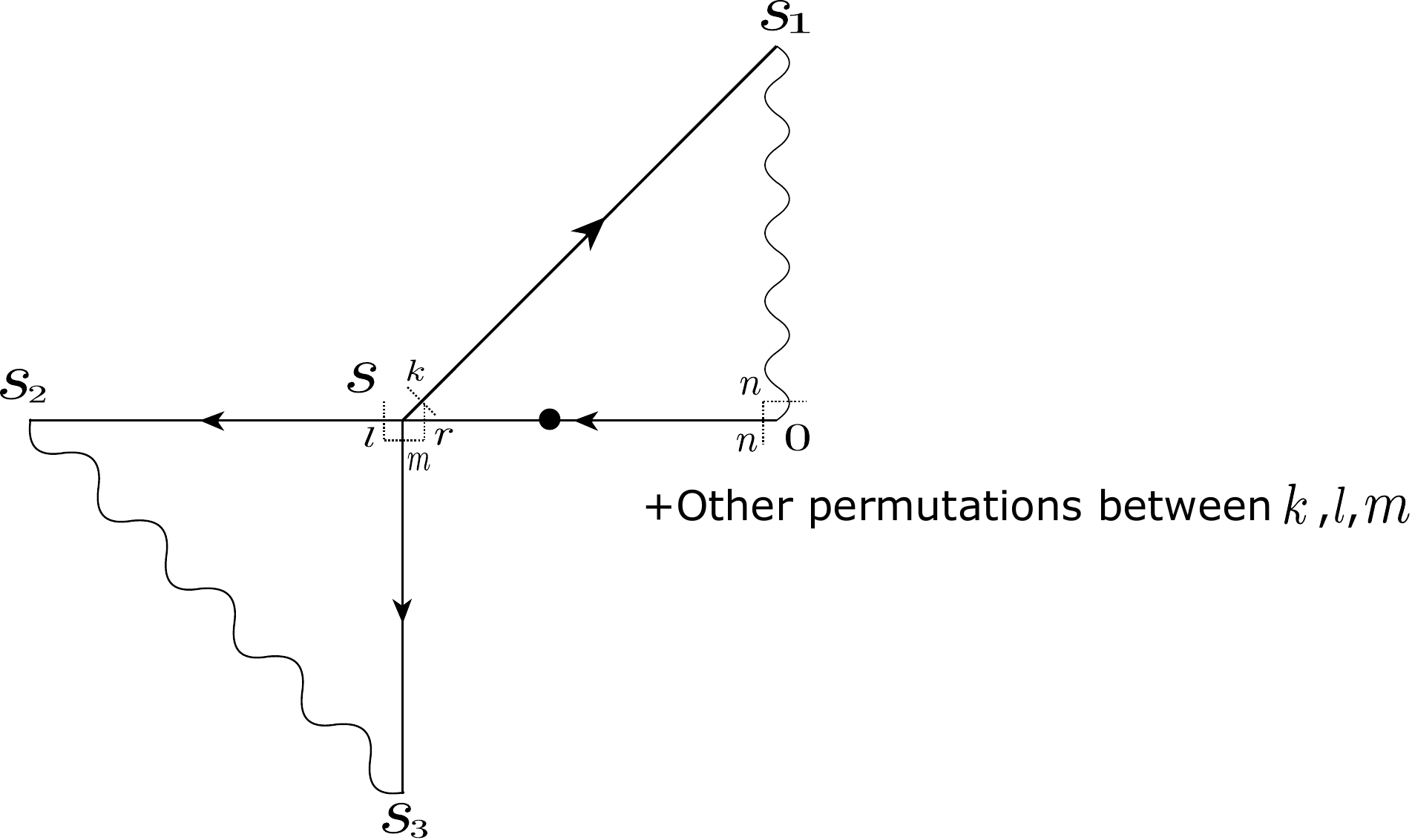}}\put(230,0){\includegraphics[scale=0.4]{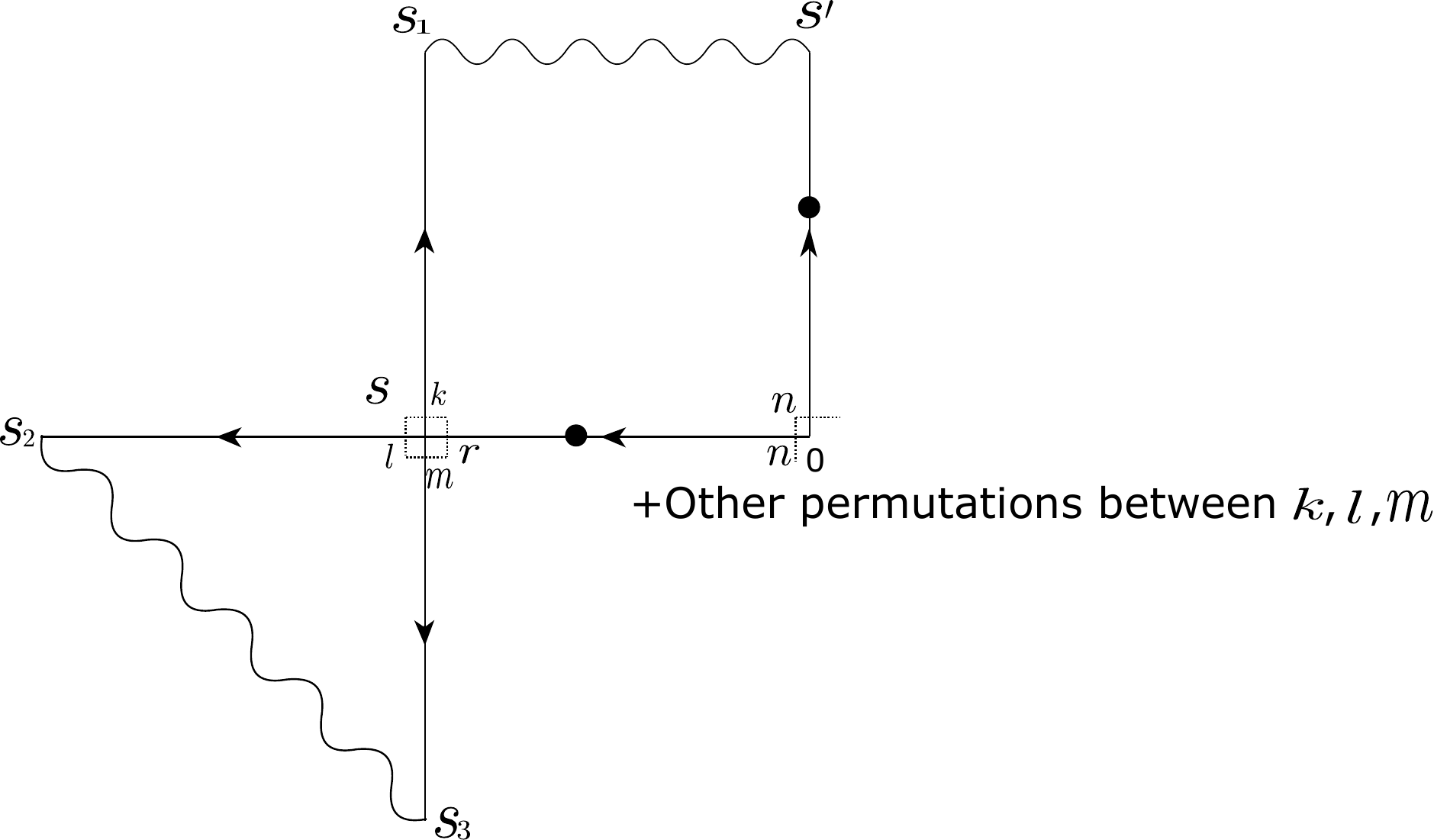}} 
\par\end{centering}
\put(50,80){\footnotesize{}(a)} \put(250,80){\footnotesize{}(b)}
\end{picture} \centering{}\caption{\label{fig:Gamma-GammaTilde-time}Time-domain diagrammatic representations
of $\Gamma_{klm}^{nr}(\infty)$ and $\tilde{\Gamma}_{klm}^{nr}(\infty)$
are shown in (a) and (b) respectively. Note that the factor of $2$
in Eq. \eqref{eq:Gamma-tilde-nr-klm-final} is not shown in (b).}
\end{figure}

\subsubsection{$\alpha$-FPUT nonlinearity}

Here we consider the corrections caused by the $\alpha$-type of nonlinearity
to $P_{\xi_{n}}$ and $P_{\gamma_{n}}$. By Theorem \ref{thm:same-s}
in  \ref{sec:Functional derivatives}, we recognize 
\begin{equation}
\partial_{n}\langle r_{k}(s)q_{l}(s)q_{m}(s)\rangle|_{\bm{h}=0}=0\,.
\end{equation}
Furthermore, following similar steps in Sec. \ref{subsec:Pxi-gamma-betaKG},
we find that for $\alpha$-FPUT nonlinearity, 
\begin{align}
\Gamma_{kl}^{nm}(t_{f}) & \equiv\lim_{t_{f}\to\infty}\int_{0}^{t_{f}}\!ds\;\langle\!\langle\xi_{n}(t_{f})\,\partial_{n}\langle r_{k}(s)r_{l}(s)q_{m}(s)\rangle_{0}\,\rangle\!\rangle\,\bigg|_{\bm{h}=0}\nonumber \\
 & =\int_{0}^{t_{f}}\!ds\;\langle\!\langle\Bigl[\dot{\bm{D}}_{2}(t_{f}-s)\Bigr]_{nm}\,\xi_{n}(t_{f})\frac{\delta^{2}\mathscr{Z}^{(0)}[\varrho_{i},\,\bm{j}]\exp\Bigl\{ i\,\bar{\mathscr{S}}_{\bm{\xi}}^{(0)}[\bm{j},\,\bm{h}=0)\Bigr\}}{\delta j_{k}(s)\delta j_{l}(s)}\,\rangle\!\rangle\bigg|_{\bm{j}=0}\,,\label{eq:gamma-nm-kl}\\
\tilde{\Gamma}_{k}^{nlm}(t_{f}) & \equiv-i\lim_{t_{f}\to\infty}\int_{0}^{t_{f}}ds\langle\!\langle\partial_{n}^{2}\langle r_{k}(s)q_{l}(s)q_{m}(s)\rangle_{0}\big|_{\bm{h}=0}\rangle\!\rangle\nonumber \\
 & =-2\int_{0}^{t_{f}}\int_{0}^{t_{f}}ds^{\prime}ds\langle\!\langle[\dot{\bm{D}}_{2}(t_{f}-s)]_{nl}\dot{\bm{D}}_{2}(t_{f}-s^{\prime})]_{nm}\frac{\delta\left(\mathscr{Z}^{(0)}[\varrho_{i},\,\bm{j}]\exp\{i\bar{\mathscr{S}}_{\bm{\xi}}^{(0)}[\bm{j},\,\bm{h}=0)\}\right)}{\delta j_{k}(s)}\bigg|_{\bm{j}=0}\rangle\!\rangle,\\
\bar{\Gamma}_{kl}^{nm}(t_{f}) & \equiv-i\lim_{t_{f}\to\infty}\int_{0}^{t_{f}}ds\langle\!\langle\partial_{n}^{2}\langle r_{k}(s)r_{l}(s)q_{m}(s)\rangle_{0}\big|_{\bm{h}=0}\rangle\!\rangle\nonumber \\
 & =2\sum_{p}\int_{0}^{t_{f}}\int_{0}^{t_{f}}ds^{\prime}ds\langle\!\langle[\dot{\bm{D}}_{2}(t_{f}-s)]_{nm}\dot{\bm{D}}_{2}(t_{f}-s^{\prime})]_{np}\xi_{p}(s^{\prime})\frac{\delta\left(\mathscr{Z}^{(0)}[\varrho_{i},\,\bm{j}]\exp\{i\bar{\mathscr{S}}_{\bm{\xi}}^{(0)}[\bm{j},\,\bm{h}=0)\}\right)}{\delta j_{k}(s)\delta j_{l}(s)}\bigg|_{\bm{j}=0}\rangle\!\rangle.\label{eq:gamma-bar-nm-kl}
\end{align}
By Theorems \ref{thm:decaying-2D-3j}-\ref{thm:decaying-2D-1j} in
 \ref{sec:Functional derivatives}, we observe that, the contributions
from the cases where at least one functional derivatives with respect
to $\bm{j}$ is applied to $\mathscr{Z}^{(0)}[\varrho_{i},\,\bm{j}]$
vanish in the late time limit $t_{f}\to\infty$. The only possible
nonzero contributions in Eqs. (\ref{eq:gamma-nm-kl}-\ref{eq:gamma-bar-nm-kl})
can only come from the situation that all functional derivatives with
respect to $\bm{j}$ are applied to $\exp\{i\bar{\mathscr{S}}_{\bm{\xi}}^{(0)}[\bm{j},\,\bm{h}=0)\}$.
However, this will generate odd powers in $\bm{\xi}$. Since $\bm{\xi}$
is Gaussian, their ensemble averages give zero. Therefore, we arrive
\begin{equation}
\Gamma_{kl}^{nm}(\infty)=\tilde{\Gamma}_{k}^{nlm}(\infty)=\bar{\Gamma}_{kl}^{nm}(\infty)=0\,.
\end{equation}
In other words, the powers $P_{\xi_{n}}(\infty)$ and $P_{\gamma_{n}}(\infty)$
do not have first-order corrections in $\alpha$-type nonlinearity.

\subsection{First order correction to the energy current $P_{\nu\to n}^{(2)}(t_{f})$ }

Here we will discuss the nonlinear correction to the inter-oscillator
energy flow, contributed by the quadratic coupling ($\eta=2$) between
oscillator $\nu$ and $n$, i.e., between oscillator 2 and oscillator
1. The correction actually stems from the nonlinearity in the generating
function. We first discuss the $\beta$- and KG-type nonlinearities.
From \eqref{eq:P-nu-eta-ave}, we see that we need to compute the
the first-order correction, $w_{\nu\to n}^{(1,\,2)}$, of $w_{\nu\to n}^{(2)}$.

\subsubsection{KG and $\beta$-FPUT nonlinearities}

According to Eqs. (\ref{eq:wn-nu-eta-Z}, \ref{eq:calJ-xi}), the
first-order corrections of $w_{\nu\to n}^{(2)}(t_{f},\,\bm{\xi}]$
due to the $\beta$- and KG-type nonlinearities are given by 
\begin{align}
w_{\nu\to n}^{(1,\,2)}(t_{f},\,\bm{\xi}] & =-\sum_{klmr}\int_{0}^{t_{f}}ds\left[\sigma_{klmr}\left(\partial_{n}\langle r_{n}(t_{f})r_{k}(s)q_{l}(s)q_{m}(s)q_{r}(s)\rangle_{0}-\partial_{n}\langle r_{\nu}(t_{f})r_{k}(s)q_{l}(s)q_{m}(s)q_{r}(s)\rangle_{0}\right)\big|_{\bm{h}=0}\right.\nonumber \\
 & \left.+\mu_{klmr}\left(\partial_{n}\langle r_{n}(t_{f})r_{k}(s)q_{l}(s)q_{m}(s)q_{r}(s)\rangle_{0}-\partial_{n}\langle r_{\nu}(t_{f})r_{k}(s)r_{l}(s)r_{m}(s)q_{r}(s)\rangle_{0}\right)\big|_{\bm{h}=0}\right],
\end{align}
where 
\begin{eqnarray}
\partial_{n}\langle r_{\nu}(t_{f})r_{k}(s)q_{l}(s)q_{m}(s)q_{r}(s)\rangle_{0}\big|_{\bm{h}=0} & = & -i\frac{\partial\delta^{5}\mathcal{Z}_{\bm{\xi}}^{(0)}[\bm{j},\,\bm{j}^{\prime},\,\bm{h})}{\partial h_{n}\delta j_{\nu}(t_{f})\delta j_{k}(s)\delta j_{l}^{\prime}(s)\delta j_{m}^{\prime}(s)\delta j_{r}^{\prime}(s)}\bigg|_{\bm{j}=\bm{j}^{\prime}=\bm{h}=0},\label{eq:dh-1rtf-1r-3q}\\
\partial_{n}\langle r_{\nu}(t_{f})r_{k}(s)r_{l}(s)r_{m}(s)q_{r}(s)\rangle_{0}\big|_{\bm{h}=0} & = & -i\frac{\partial\delta^{5}\mathcal{Z}_{\bm{\xi}}^{(0)}[\bm{j},\,\bm{j}^{\prime},\,\bm{h})}{\partial h_{n}\delta j_{\nu}(t_{f})\delta j_{k}(s)\delta j_{l}(s)\delta j_{m}(s)\delta j_{r}^{\prime}(s)}\bigg|_{\bm{j}=\bm{j}^{\prime}=\bm{h}=0}.\label{eq:dh-1rtf-3r-1q}
\end{eqnarray}
Theorem \ref{thm:same-s-tf} indicates that Eq. (\ref{eq:dh-1rtf-1r-3q})
vanishes, so we need to evaluate only Eq. (\ref{eq:dh-1rtf-3r-1q}).
So one only needs to evaluate Eq. \eqref{eq:dh-1rtf-3r-1q}. After
performing the derivatives $\partial\delta^{2}/\partial h_{n}\delta j_{\nu}(t_{f})\delta j_{r}^{\prime}(s)$,
we find 
\begin{gather}
\partial_{n}\langle r_{\nu}(t_{f})r_{k}(s)r_{l}(s)r_{m}(s)q_{r}(s)\rangle_{0}\big|_{\bm{h}=0}=-i\times\nonumber \\
\frac{\delta^{3}}{\delta j_{k}(s)\delta j_{l}(s)\delta j_{m}(s)}\left(\frac{\delta\left(\mathscr{Z}^{(0)}[\varrho_{i},\,\bm{j}]\exp\{i\bar{\mathscr{S}}_{\bm{\xi}}^{(0)}[\bm{j},\,\bm{h})\}\right)}{\delta j_{\nu}(t_{f})}\frac{\partial\delta\left(\exp\{i\bar{\mathfrak{S}}_{0\bm{\xi}}[\bm{j},\,\bm{j}^{\prime}\,\bm{h})\}\right)}{\partial h_{n}\delta j_{r}^{\prime}(s)}\right.\nonumber \\
+\left.\frac{\partial\left(\mathscr{Z}^{(0)}[\varrho_{i},\,\bm{j}]\exp\{i\bar{\mathscr{S}}_{\bm{\xi}}^{(0)}[\bm{j},\,\bm{h})\}\right)}{\partial h_{n}}\frac{\delta^{2}\left(\exp\{i\bar{\mathfrak{S}}_{0\bm{\xi}}[\bm{j},\,\bm{j}^{\prime}\,\bm{h})\}\right)}{\delta j_{\nu}(t_{f})\delta j_{r}^{\prime}(s)}+\cdots\right)\bigg|_{\bm{j}=\bm{j}^{\prime}=\bm{h}=0},
\end{gather}
where all other omitted terms vanish in the end due to according Theorems
\ref{thm:NKlessM} and \ref{thm:same-s}. As one can see from Theorem
\ref{thm:dNZscrdNj} in  \ref{sec:Functional derivatives}, if the
functional derivative $\delta/\delta j_{\nu}(t_{f})$ in the first
term of the right hand side acts on $\mathscr{Z}^{(0)}[\varrho_{i},\,\bm{j}]$,
it will generates terms that will vanish in the late time limit $t_{f}\to\infty$.
Therefore, the only possible non-vanishing terms come from when $\delta/\delta j_{\nu}(t_{f})$
acts on $\exp\{i\bar{\mathscr{S}}_{\bm{\xi}}^{(0)}[\bm{j},\,\bm{h}=0)\}$.
With this observation, we find 
\begin{gather}
\partial_{n}\langle r_{\nu}(t_{f})r_{k}(s)r_{l}(s)r_{m}(s)q_{r}(s)\rangle_{0}\big|_{\bm{h}=0}=i\times\nonumber \\
\left\{ \sum_{p}\int_{0}^{t_{f}}ds^{\prime}[\dot{\bm{D}}_{2}(t_{f}-s)]_{nr}\bm{\mathfrak{D}}(t_{f}-s^{\prime})]_{\nu p}\xi_{p}(s^{\prime})\frac{\delta^{3}\left(\mathscr{Z}^{(0)}[\varrho_{i},\,\bm{j}]\exp\{i\bar{\mathscr{S}}_{\bm{\xi}}^{(0)}[\bm{j},\,\bm{h}=0)\}\right)}{\delta j_{k}(s)\delta j_{l}(s)\delta j_{m}(s)}\right.\nonumber \\
+\left.\sum_{p}\int_{0}^{t_{f}}ds^{\prime}[\bm{\mathfrak{D}}(t_{f}-s)]_{\nu r}[\dot{\bm{D}}_{2}(t_{f}-s^{\prime})]_{np}\xi_{p}(s^{\prime})\frac{\delta^{3}\left(\mathscr{Z}^{(0)}[\varrho_{i},\,\bm{j}]\exp\{i\bar{\mathscr{S}}_{\bm{\xi}}^{(0)}[\bm{j},\,\bm{h}=0)\}\right)}{\delta j_{k}(s)\delta j_{l}(s)\delta j_{m}(s)}\right\} .
\end{gather}
Let us define 
\begin{equation}
\Upsilon_{\nu klm}^{nr}(t_{f})\equiv\lim_{t_{f}\to\infty}i\sum_{p}\int_{0}^{t_{f}}\int_{0}^{t_{f}}dsds^{\prime}\langle\!\langle[\dot{\bm{D}}_{2}(t_{f}-s)]_{nr}\bm{\mathfrak{D}}(t_{f}-s^{\prime})]_{\nu p}\xi_{p}(s^{\prime})\frac{\delta^{3}\left(\mathscr{Z}^{(0)}[\varrho_{i},\,\bm{j}]\exp\{i\bar{\mathscr{S}}_{\bm{\xi}}^{(0)}[\bm{j},\,\bm{h}=0)\}\right)}{\delta j_{k}(s)\delta j_{l}(s)\delta j_{m}(s)}\rangle\!\rangle,\label{eq:Ups-nr-klm}
\end{equation}
\begin{equation}
\tilde{\Upsilon}_{\nu klm}^{nr}(t_{f})\equiv\lim_{t_{f}\to\infty}i\sum_{p}\int_{0}^{t_{f}}\int_{0}^{t_{f}}dsds^{\prime}\langle\!\langle[\bm{\mathfrak{D}}(t_{f}-s)]_{\nu r}[\dot{\bm{D}}_{2}(t_{f}-s^{\prime})]_{np}\xi_{p}(s^{\prime})\frac{\delta^{3}\left(\mathscr{Z}^{(0)}[\varrho_{i},\,\bm{j}]\exp\{i\bar{\mathscr{S}}_{\bm{\xi}}^{(0)}[\bm{j},\,\bm{h}=0)\}\right)}{\delta j_{k}(s)\delta j_{l}(s)\delta j_{m}(s)}\rangle\!\rangle,\label{eq:Ups-tilde-nr-klm}
\end{equation}
then 
\begin{equation}
\lim_{t_{f}\to\infty}\int_{0}^{t_{f}}ds\langle\!\langle\partial_{n}\langle r_{\nu}(t_{f})r_{k}(s)r_{l}(s)r_{m}(s)q_{r}(s)\rangle_{0}\big|_{\bm{h}=0}\rangle\!\rangle=\Upsilon_{\nu klm}^{nr}(t_{f})+\tilde{\Upsilon}_{\nu klm}^{nr}(t_{f}).
\end{equation}
We have justified in Theorems \ref{thm:decaying-2D-3j} in  \ref{sec:Functional derivatives}
that in the late time limit $t_{f}\to\infty$, the only non-vanishing
terms come from the contribution when all the three functional derivatives
with respect to $\bm{j}$ acts on $\exp\{i\bar{\mathscr{S}}_{\bm{\xi}}^{(0)}[\bm{j},\,\bm{h}=0)\}$.
Therefore, similar with the calculations in Eqs. (\ref{eq:Gamma-nr-klm}-\ref{eq:Gamma-nr-klm2}),
we obtain 
\begin{align}
\Upsilon_{\nu klm}^{nr}(t_{f}) & =\idotsint_{0}^{t_{f}}dsds^{\prime}ds_{1}ds_{2}ds_{3}\dot{\bm{D}}_{2}(t_{f}-s)]_{nr}\nonumber \\
 & \times\left\{ [\bm{\mathfrak{D}}(t_{f}-s^{\prime})\bm{G}_{H}(s^{\prime}-s_{1})\bm{\mathfrak{D}}(s-s_{1})]_{\nu k}[\bm{\mathfrak{D}}(s-s_{2})\bm{G}_{H}(s_{2}-s_{3})\bm{\mathfrak{D}}(s-s_{3})]_{lm}\right.\nonumber \\
 & +[\bm{\mathfrak{D}}(t_{f}-s)\bm{G}_{H}(s^{\prime}-s_{2})\bm{\mathfrak{D}}(s-s_{2})]_{\nu l}[\bm{\mathfrak{D}}(s-s_{1})\bm{G}_{H}(s_{1}-s_{3})\bm{\mathfrak{D}}(s-s_{3})]_{km}\nonumber \\
 & +\left.[\bm{\mathfrak{D}}(t_{f}-s)\bm{G}_{H}(s^{\prime}-s_{3})\bm{\mathfrak{D}}(s-s_{3})]_{\nu m}[\bm{\mathfrak{D}}(s-s_{1})\bm{G}_{H}(s_{1}-s_{2})\bm{\mathfrak{D}}(s-s_{2})]_{kl}\right\} ,\label{eq:Ups-nr-klm-form1}
\end{align}
\begin{align}
\tilde{\Upsilon}_{\nu klm}^{nr}(t_{f}) & =\idotsint_{0}^{t_{f}}dsds^{\prime}ds_{1}ds_{2}ds_{3}[\bm{\mathfrak{D}}(t_{f}-s)]_{\nu r}\nonumber \\
 & \times\left\{ [\dot{\bm{D}}_{2}(t_{f}-s^{\prime})\bm{G}_{H}(s^{\prime}-s_{1})\bm{\mathfrak{D}}(s-s_{1})]_{nk}[\bm{\mathfrak{D}}(s-s_{2})\bm{G}_{H}(s_{2}-s_{3})\bm{\mathfrak{D}}(s-s_{3})]_{lm}\right.\nonumber \\
 & +[\dot{\bm{D}}_{2}(t_{f}-s)]\bm{G}_{H}(s^{\prime}-s_{2})\bm{\mathfrak{D}}(s-s_{2})]_{nl}[\bm{\mathfrak{D}}(s-s_{1})\bm{G}_{H}(s_{1}-s_{3})\bm{\mathfrak{D}}(s-s_{3})]_{km}\nonumber \\
 & +\left.[\dot{\bm{D}}_{2}(t_{f}-s)\bm{G}_{H}(s^{\prime}-s_{3})\bm{\mathfrak{D}}(s-s_{3})]_{nm}[\bm{\mathfrak{D}}(s-s_{1})\bm{G}_{H}(s_{1}-s_{2})\bm{\mathfrak{D}}(s-s_{2})]_{kl}\right\} .\label{eq:Ups-tilde-nr-klm-form1}
\end{align}
Making change of variables $t_{f}-s\to s$, $t_{f}-s^{\prime}\to s^{\prime}$,
and $t_{f}-s_{i}\to s_{i}$ and employing the trick given in \eqref{eq:integration-limit},
we arrive at 
\begin{align}
\Upsilon_{\nu klm}^{nr}(t_{f}) & =\idotsint_{-\infty}^{t_{f}}dsds^{\prime}ds_{1}ds_{2}ds_{3}[\dot{\bm{\mathfrak{D}}}(s)]_{nr}\nonumber \\
 & \times\left\{ [\bm{\mathfrak{D}}(s^{\prime})\bm{G}_{H}(s_{1}-s^{\prime})\bm{\mathfrak{D}}(s_{1}-s)]_{\nu k}[\bm{\mathfrak{D}}(s_{2}-s)\bm{G}_{H}(s_{3}-s_{2})\bm{\mathfrak{D}}(s_{3}-s)]_{lm}\right.\nonumber \\
 & +[\bm{\mathfrak{D}}(s)\bm{G}_{H}(s_{2}-s^{\prime})\bm{\mathfrak{D}}(s_{2}-s)]_{\nu l}[\bm{\mathfrak{D}}(s_{1}-s)\bm{G}_{H}(s_{3}-s_{1})\bm{\mathfrak{D}}(s_{3}-s)]_{km}\nonumber \\
 & +\left.[\bm{\mathfrak{D}}(s)\bm{G}_{H}(s_{3}-s^{\prime})\bm{\mathfrak{D}}(s_{3}-s)]_{\nu m}[\bm{\mathfrak{D}}(s_{1}-s)\bm{G}_{H}(s_{2}-s_{1})\bm{\mathfrak{D}}(s_{2}-s)]_{kl}\right\} ,\label{eq:Ups-nr-klm-final}
\end{align}
\begin{align}
\tilde{\Upsilon}_{\nu klm}^{nr}(t_{f}) & =\idotsint_{-\infty}^{t_{f}}dsds^{\prime}ds_{1}ds_{2}ds_{3}[\bm{\mathfrak{D}}(s)]_{\nu r}\nonumber \\
 & \times\left\{ [\dot{\bm{\mathfrak{D}}}(s^{\prime})\bm{G}_{H}(s_{1}-s^{\prime})\bm{\mathfrak{D}}(s_{1}-s)]_{nk}[\bm{\mathfrak{D}}(s_{2}-s)\bm{G}_{H}(s_{3}-s_{2})\bm{\mathfrak{D}}(s_{3}-s)]_{lm}\right.\nonumber \\
 & +[\dot{\bm{\mathfrak{D}}}(s)\bm{G}_{H}(s_{2}-s^{\prime})\bm{\mathfrak{D}}(s_{2}-s)]_{nl}[\bm{\mathfrak{D}}(s_{1}-s)\bm{G}_{H}(s_{3}-s_{1})\bm{\mathfrak{D}}(s_{3}-s)]_{km}\nonumber \\
 & +\left.[\dot{\bm{\mathfrak{D}}}(s)\bm{G}_{H}(s_{2}-s^{\prime})\bm{\mathfrak{D}}(s_{3}-s)]_{nm}[\bm{\mathfrak{D}}(s_{1}-s)\bm{G}_{H}(s_{2}-s_{1})\bm{\mathfrak{D}}(s_{2}-s)]_{kl}\right\} .\label{eq:Ups-tilde-nr-klm-final}
\end{align}
Now we take the limit $t_{f}\to\infty$, we find 
\begin{equation}
w_{\nu\to n}^{(1,\,2)}(\infty)=-\sum_{klmr}\mu_{klmr}[\Upsilon_{nklm}^{nr}(\infty)+\tilde{\Upsilon}_{nklm}^{nr}(\infty)-\Upsilon_{\nu klm}^{nr}(\infty)-\tilde{\Upsilon}_{\nu klm}^{nr}(\infty)].\label{eq:wn2-first-order}
\end{equation}
This is one of the the expressions we need to find the corrections
to the energy flow contributed by the quadratic coupling ($\eta=2$)
between oscillator $\nu$ and $n$. Next, we discuss the corresponding
contribution from the $\alpha$-type nonlinearity in the generating
functional. 
\begin{figure}
\begin{picture}(450,100) 
\begin{centering}
\put(0,0){\includegraphics[scale=0.4]{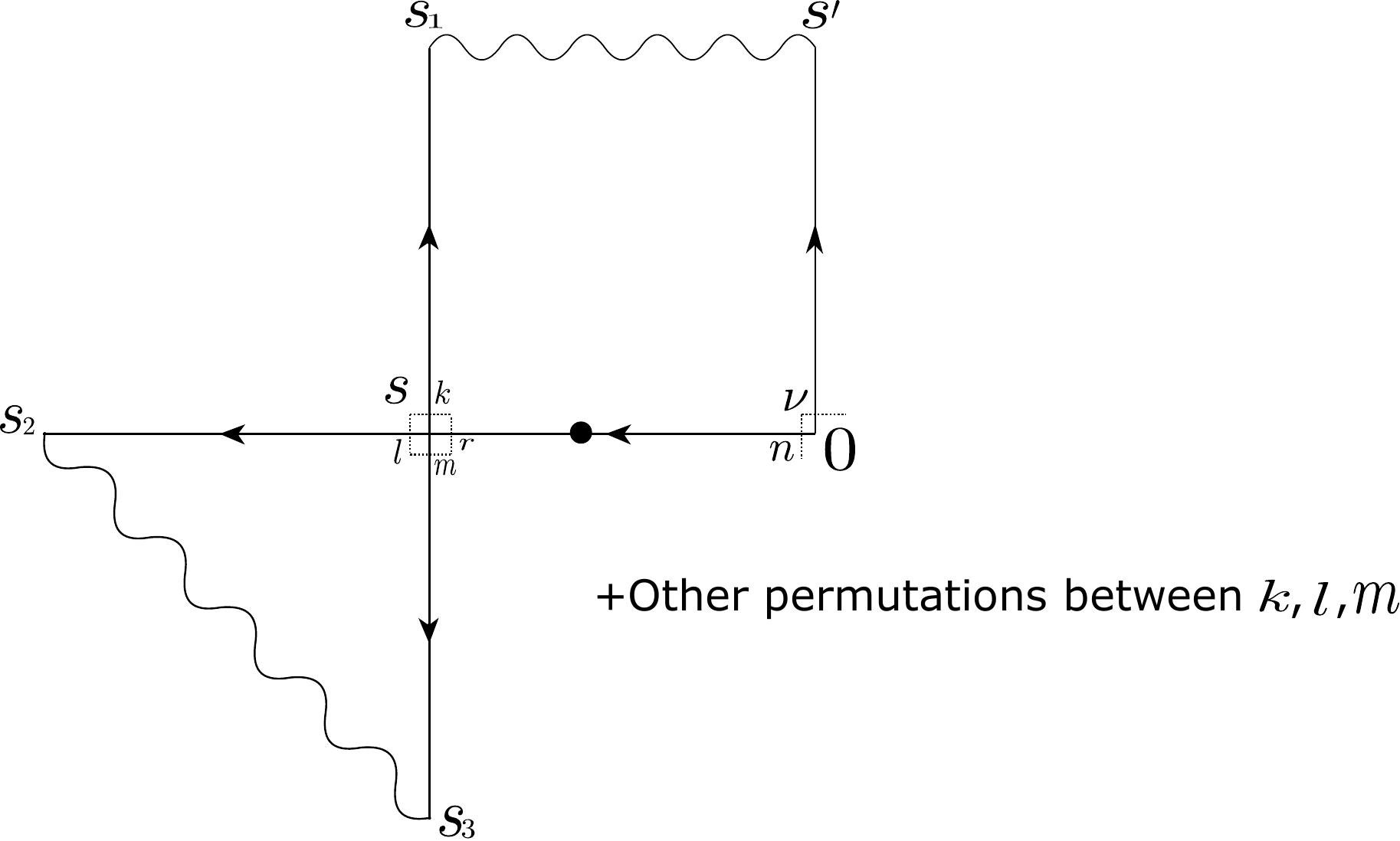}}\put(230,0){\includegraphics[scale=0.4]{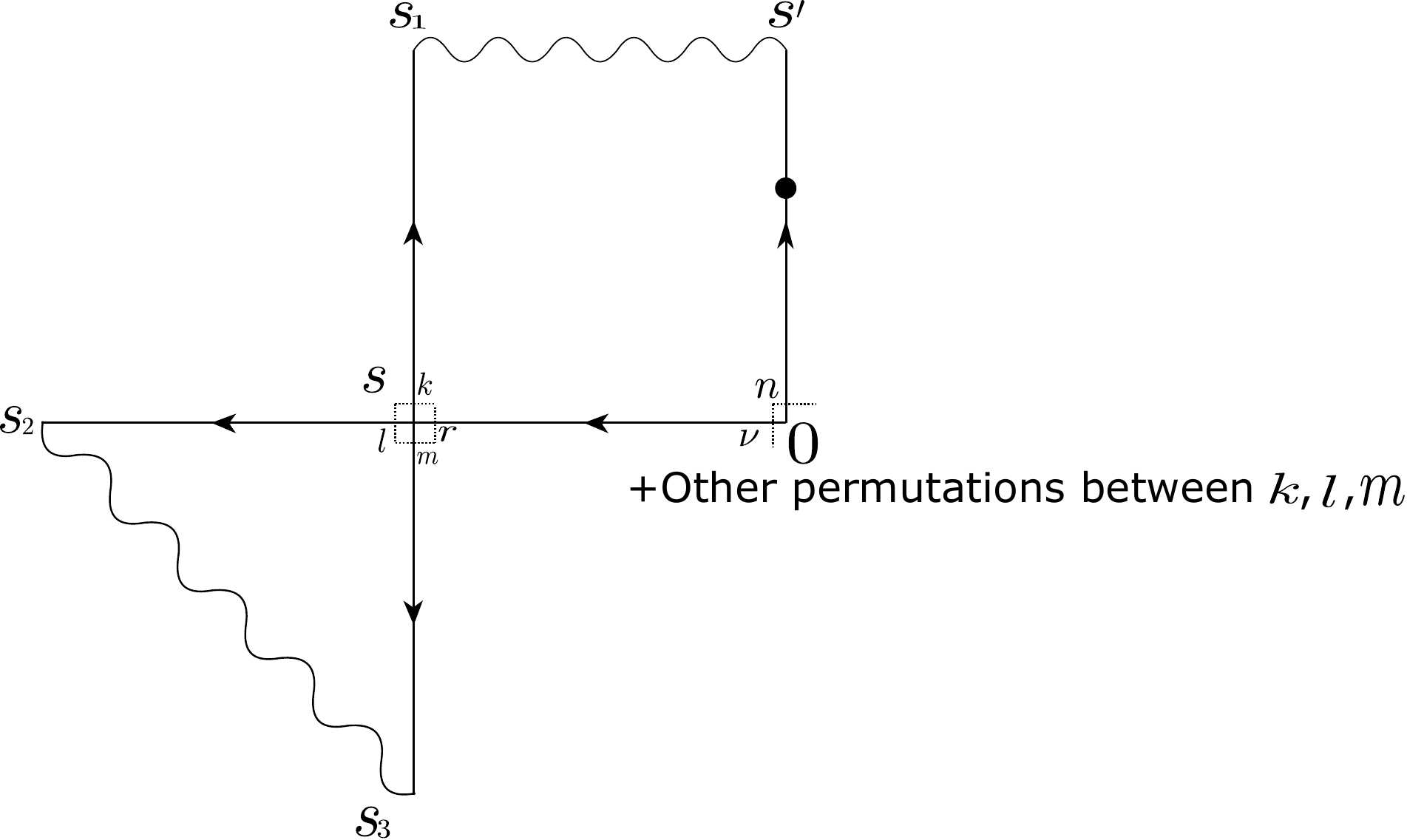}} 
\par\end{centering}
\put(25,80){\footnotesize{}(a)} \put(250,80){\footnotesize{}(b)}
\end{picture} \centering{}\caption{\label{fig:Ups-UpsTilde}Time-domain diagrammatic representation of
$\Upsilon_{klm}^{nr}(\infty)$ and $\tilde{\Upsilon}_{\nu klm}^{nr}(\infty)$}
\end{figure}

\subsubsection{$\alpha$-FPUT nonlinearity}

To compute the first-order correction to $w_{\nu\to n}^{(2)}(t_{f},\,\bm{\xi}]$
due to $\alpha$-FPUT nonlinearity in the generating functional, we
need to evaluate the following quantities 
\begin{align}
\partial_{n}\langle r_{\nu}(t_{f})r_{k}(s)q_{l}(s)q_{m}(s)\rangle_{0}\bigg|_{\bm{h}=0} & =\frac{\partial\delta^{4}\mathcal{Z}_{\bm{\xi}}^{(0)}[\bm{j},\,\bm{j}^{\prime},\,\bm{h})}{\partial h_{n}\delta j_{\nu}(t_{f})\delta j_{k}(s)\delta j_{l}^{\prime}(s)\delta j_{m}^{\prime}(s)}\bigg|_{\bm{j}=\bm{j}^{\prime}=\bm{h}=0},\label{eq:dh-1rtf-1r-2q}\\
\partial_{n}\langle r_{\nu}(t_{f})r_{k}(s)r_{l}(s)q_{m}(s)\rangle_{0}\bigg|_{\bm{h}=0} & =\frac{\partial\delta^{4}\mathcal{Z}_{\bm{\xi}}^{(0)}[\bm{j},\,\bm{j}^{\prime},\,\bm{h})}{\partial h_{n}\delta j_{\nu}(t_{f})\delta j_{k}(s)\delta j_{l}(s)\delta j_{m}^{\prime}(s)}\bigg|_{\bm{j}=\bm{j}^{\prime}=\bm{h}=0}.\label{eq:dh-1rtf-2r-1q}
\end{align}
Carrying out the derivatives gives 
\begin{align}
 & \partial_{n}\langle r_{\nu}(t_{f})r_{k}(s)q_{l}(s)q_{m}(s)\rangle_{0}\bigg|_{\bm{h}=0}\nonumber \\
= & \frac{\delta}{\delta j_{k}(s)}\left(\mathscr{Z}_{0}[\varrho_{i},\,\bm{j}]\exp\{i\bar{\mathscr{S}}_{\bm{\xi}}^{(0)}[\bm{j},\,\bm{h})\}\frac{\partial\delta^{3}\left(\exp\{i\bar{\mathfrak{S}}_{0\bm{\xi}}[\bm{j},\,\bm{j}^{\prime}\,\bm{h})\}\right)}{\partial h_{n}\delta j_{\nu}(t_{f})\delta j_{l}(s^{\prime})\delta j_{m}^{\prime}(s)}+\cdots\right)\bigg|_{\bm{j}=\bm{j}^{\prime}=\bm{h}=0},\\
 & \partial_{n}\langle r_{\nu}(t_{f})r_{k}(s)r_{l}(s)q_{m}(s)\rangle_{0}\bigg|_{\bm{h}=0}\nonumber \\
= & \frac{\delta^{2}}{\delta j_{k}(s)\delta j_{l}(s)}\left(\frac{\delta\left(\mathscr{Z}^{(0)}[\varrho_{i},\,\bm{j}]\exp\{i\bar{\mathscr{S}}_{\bm{\xi}}^{(0)}[\bm{j},\,\bm{h})\}\right)}{\delta j_{\nu}(t_{f})}\frac{\partial\delta\left(\exp\{i\bar{\mathfrak{S}}_{0\bm{\xi}}[\bm{j},\,\bm{j}^{\prime}\,\bm{h})\}\right)}{\partial h_{n}\delta j_{m}^{\prime}(s)}\right.\biggr.\nonumber \\
 & \qquad\qquad\qquad\qquad+\biggl.\left.\frac{\partial\left(\mathscr{Z}^{(0)}[\varrho_{i},\,\bm{j}]\exp\{i\bar{\mathscr{S}}_{\bm{\xi}}^{(0)}[\bm{j},\,\bm{h})\}\right)}{\partial h_{n}}\frac{\delta^{2}\left(\exp\{i\bar{\mathfrak{S}}_{0\bm{\xi}}[\bm{j},\,\bm{j}^{\prime}\,\bm{h})\}\right)}{\delta j_{\nu}(t_{f})\delta j_{m}^{\prime}(s)}+\cdots\right)\bigg|_{\bm{j}=\bm{j}^{\prime}=\bm{h}=0},
\end{align}
where all other omitted terms will vanish in the end due to Theorems
\ref{thm:NKlessM} and \ref{thm:same-s}. Similar to derivations of
Eqs.~(\ref{eq:gamma-nm-kl}-\ref{eq:gamma-bar-nm-kl}), further simplifications
give 
\begin{align}
\int_{0}^{t_{f}}\!ds\;\langle\!\langle\,\partial_{n}\langle r_{\nu}(t_{f}) & r_{k}(s)q_{l}(s)q_{m}(s)\rangle_{0}\,\rangle\!\rangle\,\bigg|_{\bm{h}=0}\nonumber \\
 & \sim\int_{0}^{t_{f}}\!ds\int_{0}^{t_{f}}\!ds'\;\langle\!\langle\,\dot{D}_{2}(t_{f}-s^{\prime})\mathfrak{D}_{2}(t_{f}-s)\,\frac{\delta\mathscr{Z}^{(0)}[\varrho_{i},\,\bm{j}]\exp\Bigl\{ i\,\bar{\mathscr{S}}_{\bm{\xi}}^{(0)}[\bm{j},\,\bm{h})\Bigr\}}{\delta j_{k}(s)}\,\rangle\!\rangle\,,\label{eq:alphaFPUT-first-w2-1}\\
\int_{0}^{t_{f}}\!ds\;\langle\!\langle\,\partial_{n}\langle r_{\nu}(t_{f}) & r_{k}(s)r_{l}(s)q_{m}(s)\rangle_{0}\,\rangle\!\rangle\,\bigg|_{\bm{h}=0}\nonumber \\
 & \sim\int_{0}^{t_{f}}\!ds\int_{0}^{t_{f}}\!ds'\;\langle\!\langle\dot{D}_{2}(t_{f}-s)\mathfrak{D}(t_{f}-s^{\prime})\,\xi(s^{\prime})\,\frac{\delta^{2}\mathscr{Z}^{(0)}[\varrho_{i},\,\bm{j}]\exp\Bigl\{ i\,\bar{\mathscr{S}}_{\bm{\xi}}^{(0)}[\bm{j},\,\bm{h})\Bigr\}}{\delta j_{k}(s)j_{l}(s)}\,\rangle\!\rangle\nonumber \\
 & \qquad\qquad+\int_{0}^{t_{f}}\!ds\int_{0}^{t_{f}}\!ds'\;\langle\!\langle\mathfrak{D}(t_{f}-s)\dot{D}_{2}(t_{f}-s^{\prime})\,\xi(s^{\prime})\,\frac{\delta^{2}\mathscr{Z}^{(0)}[\varrho_{i},\,\bm{j}]\exp\Bigl\{ i\,\bar{\mathscr{S}}_{\bm{\xi}}^{(0)}[\bm{j},\,\bm{h})\Bigr\}}{\delta j_{k}(s)j_{l}(s)}\,\rangle\!\rangle\,,\label{eq:alpha-FPUT-first-w2-2}
\end{align}
where we have used $\sim$ notation to indicate that the matrix indices
are suppressed on the right hand side, and $\dot{D}_{2}$ and $\mathfrak{D}$
stand for the matrix elements of $\bm{D}_{2}$ and $\mathfrak{\bm{D}}$
respectively. According to Theorems \ref{thm:decaying-2D-2j} and
\ref{thm:decaying-2D-1j}, we immediately see that Eqs. (\ref{eq:alphaFPUT-first-w2-1},
\ref{eq:alpha-FPUT-first-w2-2}) will vanish at late times $t_{f}\to\infty$,
so we find that 
\begin{equation}
w_{\nu\to n}^{(1,\,2)}(\infty,\,\bm{\xi}]=0\,.
\end{equation}
The $\alpha$-FPUT nonlinearity in the generating functional does
not contribute to $P_{\nu\to n}^{(1,\,2)}(\infty)$ between the two
oscillators, which results solely from the quadratic intra-oscillator
coupling \eqref{eq:wn2-first-order}.

\subsection{The first order correction to the energy flux $P_{\nu\to n}^{(\eta)}(t_{f})$
for $\eta\protect\geq3$}

Now we will focus on the contributions to the energy flows between
oscillators purely from the nonlinear intra-oscillator couplings.
Since we only consider nonlinear effects up to first order of $P_{\nu\to n}^{(\eta)}(t_{f})$,
from Eq. (\ref{eq:P-nu-eta-ave}), we need to compute $w_{\nu\to n}^{(0,\,\eta)}(t_{f})$,
the zeroth order of $w_{\nu\to n}^{(\eta)}(t_{f})$, which only involves
the linear term in the generating function. The superscript $0$ reminds
the fact that we use the zeroth-order generating functional $\mathcal{Z}_{\bm{\xi}}^{(0)}[\bm{j},\,\bm{j}^{\prime},\,\bm{h})$.
The quantity $w_{\nu\to n}^{(0,\,n)}(t_{f})$ is given by 
\begin{equation}
w_{\nu\to n}^{(0,\,\eta)}(t_{f})(t_{f},\,\bm{\xi}]=-(-i)^{\eta}\frac{\partial}{\partial h_{n}}\left(\frac{\delta}{\delta j_{n}(t_{f})}-\frac{\delta}{\delta j_{\nu}(t_{f})}\right)^{\eta-1}\mathcal{Z}_{\bm{\xi}}^{(0)}[\bm{j},\,\bm{j}^{\prime},\,\bm{h})\,\bigg|_{\bm{j}=\bm{j}^{\prime}=\bm{h}=0}\,,
\end{equation}
where 
\begin{equation}
\mathcal{Z}_{\bm{\xi}}^{(0)}[\bm{j},\,\bm{j}^{\prime},\,\bm{h})\,\bigg|_{\bm{j}'=0}=\mathscr{Z}^{(0)}[\rho_{i},\,\bm{j}^ {}]\exp\Bigl\{ i\,\bar{\mathscr{S}}_{\bm{\xi}}^{(0)}[\bm{j},\,\bm{h})\Bigr\}\,.\label{eq:Z-jprime-zero}
\end{equation}
From Theorem \ref{thm:dNZscrdNj}, we know in the late time limit
$t_{f}\to\infty$, the contributions to $w_{\nu\to n}^{(0,\,\eta)}(t_{f},\,\bm{\xi}]$
come from the case when all the functional derivatives with respect
to $\bm{j}(t_{f})$ act on $\exp\Bigl\{ i\,\bar{\mathscr{S}}_{\bm{\xi}}^{(0)}[\bm{j},\,\bm{h})\Bigr\}$.
Therefore, we find 
\begin{equation}
w_{\nu\to n}^{(0,\,\eta)}(t_{f},\,\bm{\xi}]=-(-i)^{n}\frac{\partial}{\partial h_{n}}\left(\frac{\delta}{\delta j_{n}(t_{f})}-\frac{\delta}{\delta j_{\nu}(t_{f})}\right)^{\eta-1}\exp\Bigl\{ i\,\bar{\mathscr{S}}_{\bm{\xi}}^{(0)}[\bm{j},\,\bm{h})\Bigr\}\,\bigg|_{\bm{j}=\bm{h}=0}\,.\label{eq:wn-eta-nu}
\end{equation}
For $\eta=3$, after we Taylor-expand the righthand side, we obtain
\begin{align}
 & \quad w_{\nu\to n,\,}^{(0,\,3)}(t_{f},\,\bm{\xi}]\\
 & =-\frac{1}{2!}\frac{\partial}{\partial h_{n}}\left(\frac{\delta}{\delta j_{n}(t_{f})}-\frac{\delta}{\delta j_{\nu}(t_{f})}\right)^{2}\left[\bm{h}^{T}\cdot\int_{0}^{t_{f}}\!ds\;\dot{\bm{D}}_{2}(t_{f}-s)\cdot\bm{\xi}(s)\right]\left[\int_{0}^{t_{f}}\!ds\int_{0}^{t_{f}}\!ds'\;\bm{j}^{T}(s)\cdot\bm{\mathfrak{D}}(s-s^{\prime})\cdot\bm{\xi}(s^{\prime})\right]^{2}\,\bigg|_{\bm{j}=\bm{h}=0}\nonumber 
\end{align}
from which we clearly see that after performing the derivatives with
respect to $\bm{j}$ and $\bm{h}$ and setting them zero, $w_{\nu\to n,\,}^{(0,\,3)}(t_{f},\,\bm{\xi}]$
is odd in $\bm{\xi}$. Therefore, we conclude $w_{\nu\to n}^{(0,\,3)}(t_{f})\equiv\langle\!\langle\,w_{\nu\to n,\,}^{(0,\,3)}(t_{f},\,\bm{\xi}]\rangle\!\rangle=0$,
that is, no contribution to the energy flow between two oscillators
from the $P_{\nu\to n}^{(1,\,3)}(t_{f})$ component.

Now we move on to the $\eta=4$ case. All we need to do is find the
right hand side, 
\begin{equation}
w_{\nu\to n}^{(0,\,4)}(t_{f},\,\bm{\xi}]=-\sum_{klm}\kappa_{klm}\,\frac{\partial\delta^{3}\exp\Bigl\{ i\,\bar{\mathscr{S}}_{\bm{\xi}}^{(0)}[\bm{j},\,\bm{h})\Bigr\}}{\partial h_{n}\delta j_{k}(t_{f})\delta j_{l}(t_{f})\delta j_{m}(t_{f})}\,\bigg|_{\bm{j}=\bm{h}=0}.
\end{equation}
One can see that only the fourth order in the Taylor expansion of
$\exp\Bigl\{ i\,\bar{\mathscr{S}}_{\bm{\xi}}^{(0)}[\bm{j},\,\bm{h})\Bigr\}$
will give a non-zero contribution since it contains the term $h\cdot j^{3}$.
Carrying out this procedure yields,  
\begin{align}
 & \quad\frac{\partial\delta^{3}\exp\left\{ i\,\bar{\mathscr{S}}_{\bm{\xi}}^{(0)}[\bm{j},\,\bm{h})\right\} }{\partial h_{n}\delta j_{k}(t_{f})\delta j_{l}(t_{f})\delta j_{m}(t_{f})}\,\bigg|_{\bm{j}=\bm{h}=0}\nonumber \\
 & =\sum_{opqr}\idotsint_{0}^{t_{f}}\!dsds_{1}ds_{2}ds_{3}\;\Bigl[\dot{\bm{D}}_{2}(t_{f}-s)]_{no}\cdot\bm{\xi}(s)\Bigr]_{o}\Bigl[\bm{\mathfrak{D}}(t_{f}-s_{1})\Bigr]_{kp}\,\bm{\xi}_{p}(s_{1})\Bigl[\bm{\mathfrak{D}}(t_{f}-s_{2})\Bigr]_{lq}\,\bm{\xi}_{q}(s_{2})\Bigl[\bm{\mathfrak{D}}(t_{f}-s_{3})\Bigr]_{mr}\,\bm{\xi}_{r}(s_{3})\,.
\end{align}
We introduce a shorthand notation 
\begin{equation}
\Lambda_{klm}^{n}(t_{f})\equiv\langle\!\langle\,\frac{\partial\delta^{3}\exp\left\{ i\,\bar{\mathscr{S}}_{\bm{\xi}}^{(0)}[\bm{j},\,\bm{h})\right\} }{\partial h_{n}\delta j_{k}(t_{f})\delta j_{l}(t_{f})\delta j_{m}(t_{f})}\,\rangle\!\rangle\,\bigg|_{\bm{j}=\bm{h}=0}.\label{eq:Lambda-def}
\end{equation}
Making change of variables $t_{f}-s\to s$, $t_{f}-s_{i}\to s_{i}$
and employing the trick given in \eqref{eq:integration-limit}, we
obtain 
\begin{align}
\Lambda_{klm}^{n}(t_{f}) & =\idotsint_{-\infty}^{t_{f}}\!dsds_{1}ds_{2}ds_{3}\;\biggl\{\Bigl[\dot{\bm{\mathfrak{D}}}(s)\bm{G}_{H}(s_{1}-s)\bm{\mathfrak{D}}(s_{1})\Bigr]_{nk}\Bigl[\bm{\mathfrak{D}}(s_{2})\bm{G}_{H}(s_{3}-s_{2})\bm{\mathfrak{D}}(s_{3})\Bigr]_{lm}\biggr.\nonumber \\
 & \qquad\qquad\qquad\qquad\qquad\qquad+\Bigl[\dot{\bm{\mathfrak{D}}}(s)\bm{G}_{H}(s_{2}-s)\bm{\mathfrak{D}}(s_{2})\Bigr]_{nl}\Bigl[\bm{\mathfrak{D}}(s_{1})\bm{G}_{H}(s_{3}-s_{1})\bm{\mathfrak{D}}(s_{3})\Bigr]_{km}\nonumber \\
 & \qquad\qquad\qquad\qquad\qquad\qquad\qquad\qquad+\biggl.\Bigl[\dot{\bm{\mathfrak{D}}}(s)\bm{G}_{H}(s_{3}-s)\bm{\mathfrak{D}}(s_{3})\Bigr]_{nm}\Bigl[\bm{\mathfrak{D}}(s_{1})\bm{G}_{H}(s_{2}-s_{1})\bm{\mathfrak{D}}(s_{2})\Bigr]_{kl}\biggr\}\,.\label{eq:Lambda-final}
\end{align}
Since from \eqref{eq:wn-eta-nu} for $\eta=4$, we know that the coefficients
$\kappa_{klm}$ will be determined by the binomial expansion of the
form $\Bigl[\delta/\delta j_{n}(t_{f})-\delta/\delta j_{\nu}(t_{f})\Bigr]^{3}$,
we thus find 
\begin{equation}
w_{\nu\to n,\,0}^{(0,4)}(\infty)=-\sum_{klm}\kappa_{klm}\Lambda_{klm}^{n}(\infty)=-\Bigl[\Lambda_{nnn}^{n}(\infty)-3\Lambda_{nn\nu}^{n}(\infty)+3\Lambda_{n\nu\nu}^{n}(\infty)-\Lambda_{\nu\nu\nu}^{n}(\infty)\Bigr]\,,\label{eq:wbeta-zero-order}
\end{equation}
after the motion is fully relaxed. 
\begin{figure}
\begin{centering}
\includegraphics[scale=0.5]{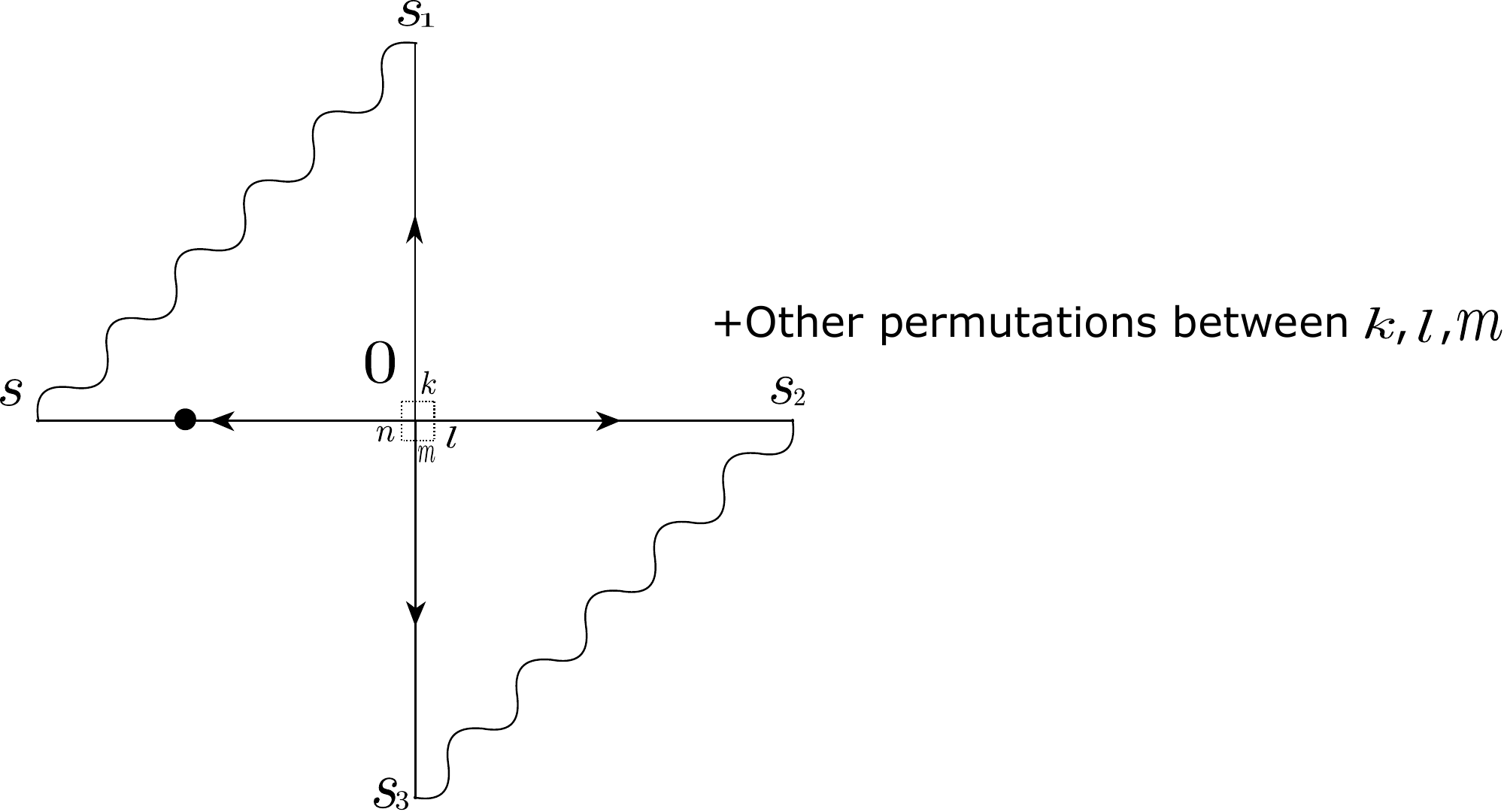} 
\par\end{centering}
\centering{}\caption{\label{fig:Lambda-time}Time-domain diagrammatic representation of
$\Lambda_{klm}^{n}(\infty)$}
\end{figure}

Now we have calculated all the first-order corrections \eqref{eq:Pxi-first-order},
\eqref{eq:Pgamma-first-order} of the rates of the energy exchange
between each oscillator and its private bath, as well as the correction
\eqref{eq:wn2-first-order}, \eqref{eq:wbeta-zero-order} of the energy
flow between the neighboring oscillators. We have noted that 1) the
$\alpha$-FPUT nonlinearity does not contribute to the quantities
of our interest, and 2) these first-order corrections all become time-independent
on a time scale much greater than the relaxation time. Thus it seems
to imply a steady state for the current configuration. We will address
this feature in Sec. \ref{sec:NESS}

We conclude this section by highlighting that the first-order nonlinear
corrections to the rate of the energy exchange between the oscillators
consist of two distinct components: (a) the first-order correction
of the current $P_{\nu\to n}^{(2)}$ corresponding to the intra-oscillator
quadratic coupling, which can be induced by either cubic $\alpha$-FPUT
or $\beta$-FPUT and KG quartic interaction in the chain Hamiltonian.
(b) the the first-order correction of the current $P_{\nu\to n}^{(\eta)}(\eta\ge3)$
corresponding to the intra-oscillator cubic $\alpha$-FPUT or quartic
$\beta$-FPUT coupling. Obviously, implied by (a) and (b), the KG
nonlinearity is self-interaction at each site and does not couple
two different oscillators. So it only induces a correction to $P_{\nu\to n}^{(2)}$
and does not generate the current $P_{\nu\to n}^{(4)}$.

\section{\label{sec:Diagrammatic-representations}Diagrammatic representations}

From now on, we will discuss the steady-state energy and will implicitly
assume $t_{f}=\infty$. Therefore we suppress all the time dependence
in all the relevant energy currents, which are given by Eqs. (\ref{eq:Pxi0},
\ref{eq:P-gamma-zeroth}, \ref{eq:W-n2-zeroth},\ref{eq:Gamma-nr-klm-final},
\ref{eq:Gamma-tilde-nr-klm-final}, \ref{eq:Ups-nr-klm-final}, \ref{eq:Ups-tilde-nr-klm-final},
\ref{eq:Lambda-final}) respectively. Although they may look formidable,
all these energy currents admit time-domain diagrammatic representations,
which can be seen directly from their respective analytic expressions
through definite rules. Furthermore, one can easily convert the time-domain
diagrams to the frequency domain diagrams, which allows one to obtain
the Fourier transform of these energy currents in an economical way.
Finally, these diagrams provide an intuitive understanding of the
energy exchange between the baths and the oscillators and among the
oscillators.

\subsection{\label{subsec:Properties-of-causal-propagator}Properties of causal
propagator $\tilde{\bm{\mathfrak{D}}}(\omega)$}

Before moving to the diagrammatic representation, let us discuss some
useful properties of the Fourier transform the the causal propagator
$\bm{\mathfrak{D}}(s)$, which will be used later. According to Eq.
\eqref{eq:D2-s-def}, the Laplace transform of $\bm{D}_{2}(s)$, i.e.,
$\bm{D}_{2}(\mathfrak{s})$ is symmetric and so is $\bm{D}_{2}(s)$.
Here one should not be confused with the variable in the Laplace domain
denoted as $\mathfrak{s}$ and the variable in the time domain denoted
as $s$. According to the relation between Fourier and Laplace transform,
one can readily find

\begin{equation}
\tilde{\bm{\mathfrak{D}}}(\omega)=\bm{D}_{2}(\mathfrak{s}=i\omega)=[-\omega^{2}+2i\gamma\omega+\bm{\Omega}_{R}^{2}]^{-1},\label{eq:Dfrak-omega}
\end{equation}
where $\bm{\Omega}_{R}^{2}$ defined in Eq. \eqref{eq:Omega-Renormalized}
is the renormalized version of Eq. \eqref{eq:Omega-bare}. From Eq.
\eqref{eq:Dfrak-omega}, we know 
\begin{eqnarray}
[\tilde{\bm{\mathfrak{D}}}(\omega)]_{12} & = & [\tilde{\bm{\mathfrak{D}}}(\omega)]_{21},\label{eq:off-diag-ME}\\{}
[\tilde{\bm{\mathfrak{D}}}(\omega)]_{11} & = & [\tilde{\bm{\mathfrak{D}}}(\omega)]_{22},\label{eq:diag-ME}
\end{eqnarray}
\begin{eqnarray}
(-\omega^{2}+2i\gamma\omega+\omega_{R0}^{2}+\lambda_{2})[\tilde{\bm{\mathfrak{D}}}(\omega)]_{nn}-\lambda_{2}[\tilde{\bm{\mathfrak{D}}}(\omega)]_{n\nu} & = & 1,\label{eq:Inverse1}\\
(-\omega^{2}+2i\gamma\omega+\omega_{R0}^{2}+\lambda_{2})[\tilde{\bm{\mathfrak{D}}}(\omega)]_{n\nu}-\lambda_{2}[\tilde{\bm{\mathfrak{D}}}(\omega)]_{nn} & = & 0.\label{eq:Inverse2}
\end{eqnarray}
where $\nu\neq n$. Furthermore, since $\bm{\mathfrak{D}}(s)$ is
real, one can easily conclude 
\begin{equation}
\tilde{\bm{\mathfrak{D}}}^{*}(\omega)=\tilde{\bm{\mathfrak{D}}}(-\omega)\label{eq:D-omega-symmetry}
\end{equation}
or equivalently $\text{Re}\tilde{\bm{\mathfrak{D}}}(\omega)$ is even
in $\omega$ while $\text{Im}\tilde{\bm{\mathfrak{D}}}(\omega)$ is
odd in $\omega$.

\subsection{Feynman diagrams in the time domain}

The following rules can be used in order convert the analytic expressions
typically in Eqs. (\ref{eq:Pxi0}, \ref{eq:pn2-zeroth}, \ref{eq:W-n2-zeroth},\ref{eq:Gamma-nr-klm-final},
\ref{eq:Gamma-tilde-nr-klm-final}, \ref{eq:Ups-nr-klm-final}, \ref{eq:Ups-tilde-nr-klm-final},
\ref{eq:Lambda-final} ) into diagrams shown in Figs. \ref{fig:zeroth-order-time}-\ref{fig:Lambda-time}. 
\begin{enumerate}
\item $s_{1}\Dcausal s_{2}$ denotes the causal Green's function of the
chain $\bm{\mathfrak{D}}(s_{2}-s_{1})$. 
\item The bullet denotes the time derivative with respect to the argument
of the chain's causal Green's function. Therefore $s_{1}\DcausalDotA s_{2}$
or $s_{1}\DcausalDotB s_{2}$ is the time derivative the causal Green's
function of the chain, i.e., $\dot{\bm{\mathfrak{D}}}(s_{2}-s_{1})$,
where the overhead dot denotes the derivative with respect to the
argument of $\dot{\bm{\mathfrak{D}}}(s_{2}-s_{1})$. The order that
the arrow and the bullet appears does not matter. 
\item $s_{1}\GH s_{2}$ denotes the Hadamard Green's function of the field
$\bm{G}_{H}(s_{2}-s_{1})$ or $\bm{G}_{H}(s_{1}-s_{2})$. There is
no arrows placed on the wavy line since $\bm{G}_{H}$ is symmetric
in $s_{1}$ and $s_{2}$. 
\item All the vortices such as $s$, $s^{\prime}$, $s_{i}$, except the
origin, must be integrated out, with integration limit goes from $-\infty$
to $+\infty$ (since we implicitly assume $t_{f}=\infty$). 
\item The indices of matrix elements associated with the propagators are
indicated by the letters beside the short line segments. 
\end{enumerate}
Conversely, with above rules, one can recover Eqs. (\ref{eq:Pxi0},
\ref{eq:pn2-zeroth}, \ref{eq:W-n2-zeroth},\ref{eq:Gamma-nr-klm-final},
\ref{eq:Gamma-tilde-nr-klm-final}, \ref{eq:Ups-nr-klm-final}, \ref{eq:Ups-tilde-nr-klm-final},
\ref{eq:Lambda-final} ) from Figs. \ref{fig:zeroth-order-time}-\ref{fig:Lambda-time}.
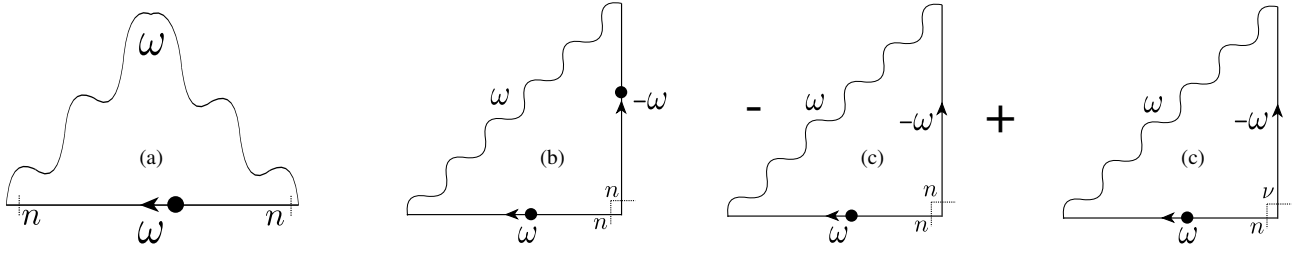
\begin{figure}
\begin{picture}(450,100) 
\begin{centering}
\put(0,0){\includegraphics[scale=0.75]{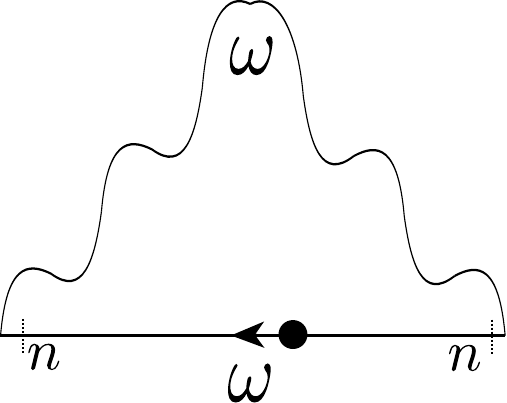}}\put(150,0){\includegraphics[scale=0.55]{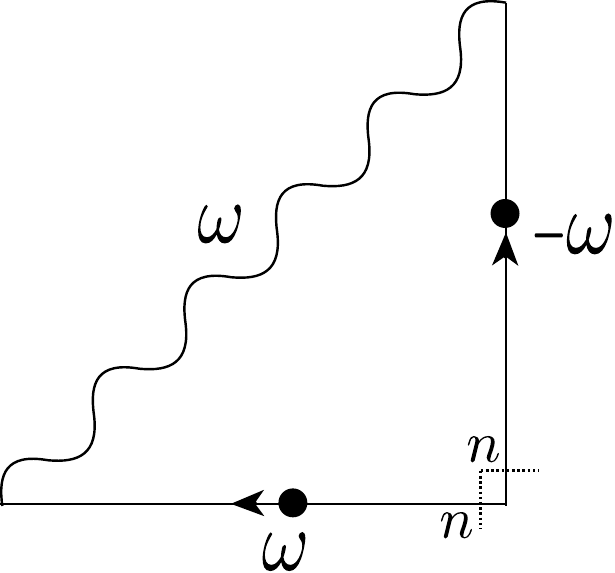}}\put(270,0){\includegraphics[scale=0.55]{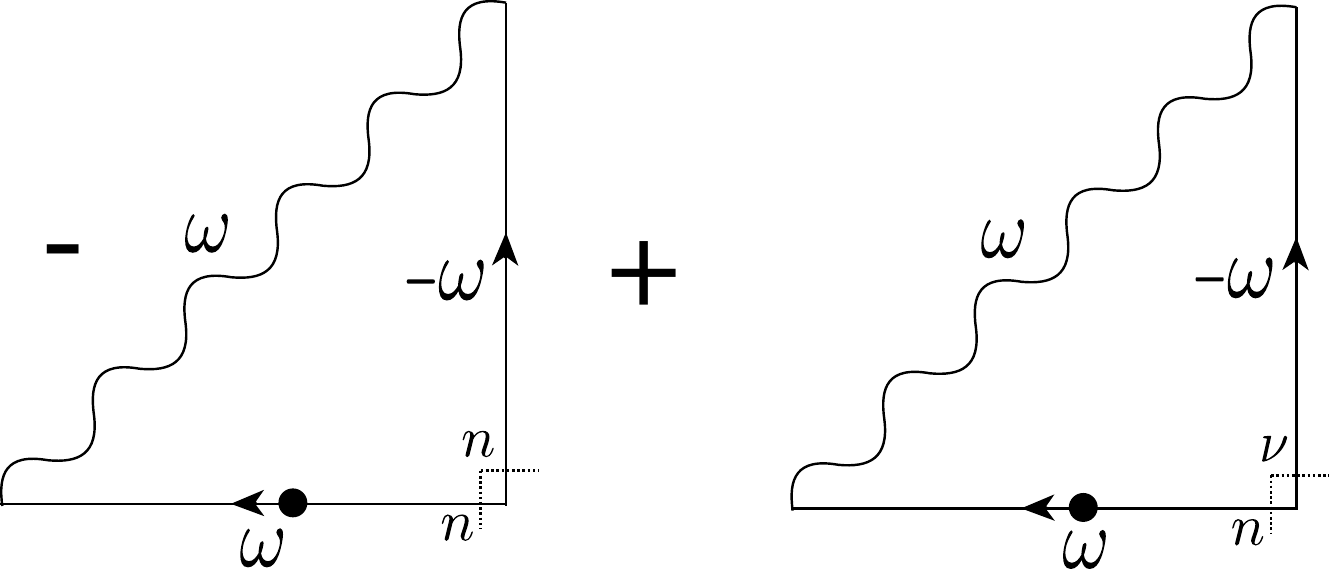}} 
\par\end{centering}
\put(50,30){\footnotesize{}(a)} \put(200,30){\footnotesize{}(b)}
\put(320,30){\footnotesize{}(c)} \put(440,30){\footnotesize{}(c)}
\end{picture}

\caption{\label{fig:zeroth-order-frequency}Frequency-domain diagrammatic representation
of the zeroth-order energy currents including Figs. (a) $P_{\xi_{n}}(\infty)$
(b) $p_{n}^{2(0)}(\infty)$ (c) $w_{n,\,2}^{\nu}(\infty)$. As we
show Sec. \ref{subsec:Properties-of-causal-propagator}, the first
digram in (c) actually vanishes.}
\end{figure}

\subsection{Feynman diagrams in the frequency domain\label{subsec:Feynman-diagrams-frequency}}

In standard quantum field theory, the momentum-space Feynman-diagrams
are far more important and easier to work with than their position-space
counter parts. Similarly, here the diagrams in the frequency domain
are more useful in proving the NESS. One can convert the time-domain
diagrams shown in Figs. \ref{fig:zeroth-order-time}-\ref{fig:Lambda-time}
according to the following rules: 
\begin{enumerate}
\item Removing all the time arguments but keeping the indices of the matrix
element from the time-domain diagrams 
\item Associated frequencies (energies) with for both the causal and Hadamard
propagators such that the energy is conserved at all the vortices. 
\item When imposing energy conservation at each vortex, the arrows in the
causal propagator denotes the direction of energy flow. The direction
of the energy flow in the Hadamard propagator is omitted since it
can be easily inferred by the frequency associated with it and the
energy conservation. 
\end{enumerate}
After implementing above procedures, all the corresponding frequency-domain
diagrams are shown in Figs. \eqref{fig:zeroth-order-frequency}-\eqref{fig:Lambda-Frequency}.
One can directly read off the Fourier transforms of Eqs. (\ref{eq:Pxi0},
\ref{eq:pn2-zeroth}, \ref{eq:W-n2-zeroth},\ref{eq:Gamma-nr-klm-final},
\ref{eq:Gamma-tilde-nr-klm-final}, \ref{eq:Ups-nr-klm-final}, \ref{eq:Ups-tilde-nr-klm-final},
\ref{eq:Lambda-final} ) from Figs. \eqref{fig:zeroth-order-frequency}-\eqref{fig:Lambda-Frequency},
according to the following rules: 
\begin{figure}
\begin{picture}(450,100)
\begin{centering}
\put(0,0){\includegraphics[scale=0.4]{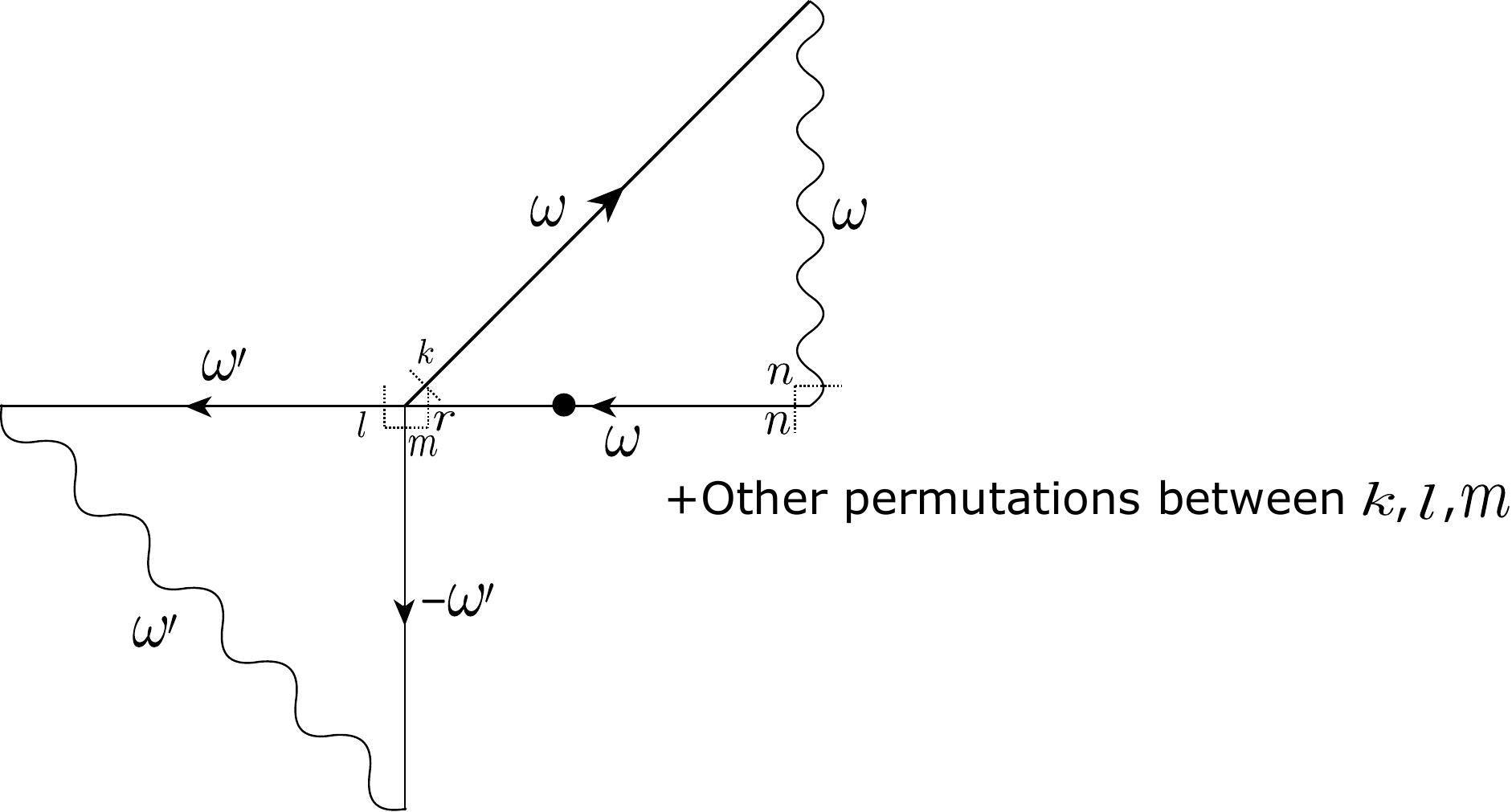}}\put(230,0){\includegraphics[scale=0.4]{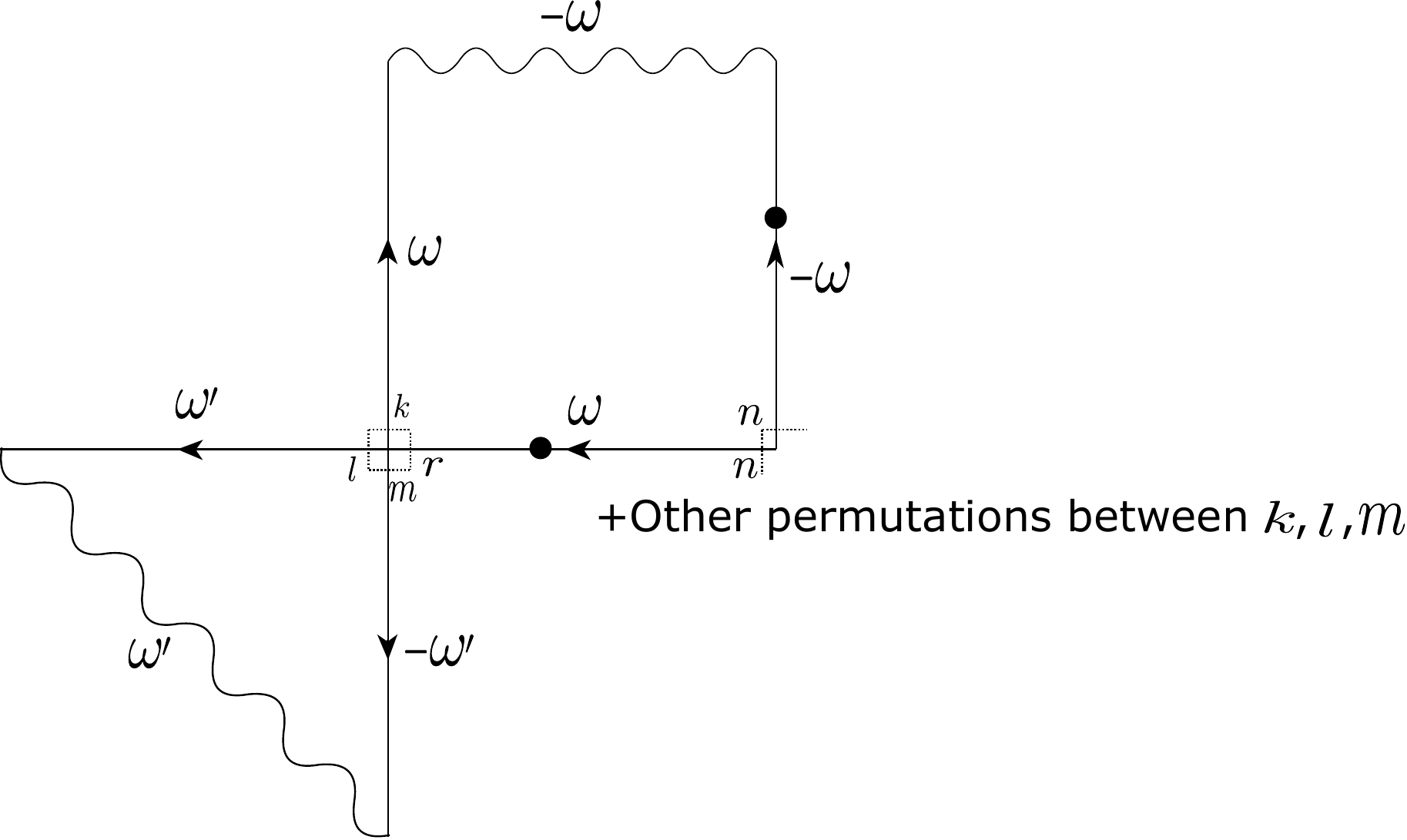}}
\par\end{centering}
\put(50,80){\footnotesize{}(a)} \put(250,80){\footnotesize{}(b)}
\end{picture} \centering{}\caption{\label{fig:Gamma-GammaTilde-Frequency}Frequency-domain diagrammatic
representations of $\Gamma_{klm}^{nr}(\infty)$ and $\tilde{\Gamma}_{klm}^{nr}(\infty)$
are shown in (a) and (b) respectively. Note that the factor of $2$
in Eq. \eqref{eq:Gamma-tilde-nr-klm-final} is not shown in (b).}
\end{figure}

\begin{enumerate}
\item $\Dcausal$ denotes $\tilde{\bm{\mathfrak{D}}}(\pm\omega)$, the Fourier
transform of the causal propagator $\bm{\mathfrak{D}}(s)$, where
the sign depends on the frequency associated with the propagator. 
\item $\DcausalDotA$ or $\DcausalDotB$ denotes $(\pm i\omega)\tilde{\bm{\mathfrak{D}}}(\pm\omega)$.
Again the sign depends on the sign of the frequency associated with
the propagator. 
\item $\GH$ denotes $\tilde{\bm{G}}_{H}(\pm\omega)$, the Fourier transform
of $\bm{G}_{H}(s)$, where the sign depends on the assigned frequency. 
\item In the end, the independent frequencies must be integrated out with
the measure $\int d\omega/(2\pi)$. Note that the number of independent
frequencies are equal to the number of loops in the diagrams. 
\end{enumerate}
According to above rules, one can easily read off the the expression
of all the zeroth order correction from Fig. \ref{fig:zeroth-order-frequency}

\begin{equation}
P_{\xi_{n}}^{(0)}=\frac{1}{2\pi}\int_{-\infty}^{\infty}(i\omega)[\tilde{\bm{\mathfrak{D}}}(\omega)\tilde{\bm{G}}_{H}(\omega)]_{nn}d\omega,\label{eq:Pxi0-frequency}
\end{equation}
\begin{align}
p_{n}^{2(0)} & =\frac{1}{2\pi}\int_{-\infty}^{\infty}\omega^{2}[\tilde{\bm{\mathfrak{D}}}(\omega)\tilde{\bm{G}}_{H}(\omega)\tilde{\bm{\mathfrak{D}}}(-\omega)]_{nn}d\omega,\label{eq:pn20-frequency}
\end{align}
\begin{equation}
w_{\nu\to n}^{(0,\,2)}=-\frac{1}{2\pi}\int_{-\infty}^{\infty}i\omega\left([\tilde{\bm{\mathfrak{D}}}(\omega)\tilde{\bm{G}}_{H}(\omega)\tilde{\bm{\mathfrak{D}}}(-\omega)]_{nn}-\tilde{\bm{\mathfrak{D}}}(\omega)\tilde{\bm{G}}_{H}(\omega)\tilde{\bm{\mathfrak{D}}}(-\omega)]_{n\nu}\right)d\omega.\label{eq:wnnu0-frequency}
\end{equation}
With Eq. \eqref{eq:D-omega-symmetry}, one can easily show that 
\begin{equation}
\int_{-\infty}^{\infty}i\omega[\tilde{\bm{\mathfrak{D}}}(\omega)\tilde{\bm{G}}_{H}(\omega)\tilde{\bm{\mathfrak{D}}}(-\omega)]_{nn}=\sum_{k=n,\nu}\int_{-\infty}^{\infty}i\omega\big|[\tilde{\bm{\mathfrak{D}}}(\omega)]_{nk}\big|^{2}[\tilde{\bm{G}}_{H}(\omega)]_{kk}=0,\label{eq:odd-integral}
\end{equation}
since the integrand is an odd function of $\omega$. Eq. \eqref{eq:wnnu0-frequency}
becomes 
\begin{equation}
w_{\nu\to n}^{(0,\,2)}=\frac{1}{2\pi}\int_{-\infty}^{\infty}i\omega[\tilde{\bm{\mathfrak{D}}}(\omega)\tilde{\bm{G}}_{H}(\omega)\tilde{\bm{\mathfrak{D}}}(-\omega)]_{n\nu}d\omega.\label{eq:wn20-frequency-simpified}
\end{equation}
Equivalently, Eq. \eqref{eq:odd-integral} justifies that the first
diagram in Fig. \ref{fig:zeroth-order-time}(c) or Fig.\ref{fig:zeroth-order-frequency}(c)
vanishes. Similarly, from Fig. \ref{fig:Gamma-GammaTilde-Frequency},
one can easily obtain the frequency-domain representation for $\Gamma_{klm}^{nr}$
and $\tilde{\Gamma}_{klm}^{nr}$

\begin{align}
\Gamma_{klm}^{nr} & =\frac{1}{(2\pi)^{2}}\iint d\omega d\omega^{\prime}(i\omega)[\tilde{\bm{\mathfrak{D}}}(\omega)]_{nr}\left\{ [\tilde{\bm{G}}_{H}(\omega)\tilde{\bm{\mathfrak{D}}}(\omega)]_{nk}[\tilde{\bm{\mathfrak{D}}}(\omega^{\prime})\tilde{\bm{G}}_{H}(\omega^{\prime})\tilde{\bm{\mathfrak{D}}}(-\omega^{\prime})]_{lm}\right.\nonumber \\
 & +[\tilde{\bm{G}}_{H}(\omega)\tilde{\bm{\mathfrak{D}}}(\omega)]_{nl}[\tilde{\bm{\mathfrak{D}}}(\omega^{\prime})\tilde{\bm{G}}_{H}(\omega^{\prime})\tilde{\bm{\mathfrak{D}}}(-\omega^{\prime})]_{km}\nonumber \\
 & +\left.[\tilde{\bm{G}}_{H}(\omega)\tilde{\bm{\mathfrak{D}}}(\omega)]_{nm}[\tilde{\bm{\mathfrak{D}}}(\omega^{\prime})\tilde{\bm{G}}_{H}(\omega^{\prime})\tilde{\bm{\mathfrak{D}}}(-\omega^{\prime})]_{kl}\right\} ,\label{eq:Gamma-nr-klm-Frequency}
\end{align}
\begin{align}
\tilde{\Gamma}_{klm}^{nr} & =\frac{2}{(2\pi)^{2}}\iint d\omega d\omega^{\prime}\omega^{2}[\tilde{\bm{\mathfrak{D}}}(\omega)]_{nr}\left\{ [\tilde{\bm{\mathfrak{D}}}(-\omega)\tilde{\bm{G}}_{H}(-\omega)\tilde{\bm{\mathfrak{D}}}(\omega)]_{nk}[\tilde{\bm{\mathfrak{D}}}(\omega^{\prime})\tilde{\bm{G}}_{H}(\omega^{\prime})\tilde{\bm{\mathfrak{D}}}(-\omega^{\prime})]_{lm}\right.\nonumber \\
 & +[\tilde{\bm{\mathfrak{D}}}(-\omega)\tilde{\bm{G}}_{H}(-\omega)\tilde{\bm{\mathfrak{D}}}(\omega)]_{nl}[\tilde{\bm{\mathfrak{D}}}(\omega^{\prime})\tilde{\bm{G}}_{H}(\omega^{\prime})\tilde{\bm{\mathfrak{D}}}(-\omega^{\prime})]_{km}\nonumber \\
 & +\left.[\tilde{\bm{\mathfrak{D}}}(-\omega)\tilde{\bm{G}}_{H}(-\omega)\tilde{\bm{\mathfrak{D}}}(\omega)]_{nm}[\tilde{\bm{\mathfrak{D}}}(\omega^{\prime})\tilde{\bm{G}}_{H}(\omega^{\prime})\tilde{\bm{\mathfrak{D}}}(-\omega^{\prime})]_{kl}\right\} .\label{eq:GammaTilde-nr-klm-Frequency}
\end{align}
Finally, the frequency-domain diagrammatic representations for Eqs.
(\ref{eq:Ups-nr-klm-final}, \ref{eq:Ups-tilde-nr-klm-final}, \ref{eq:Lambda-final})
are shown in Figs. \ref{fig:Ups-UpsTilde-Frequency} and \ref{fig:Lambda-Frequency}.
From these diagrams, one can easily find 
\begin{align}
\Upsilon_{\nu klm}^{nr} & =\frac{1}{(2\pi)^{2}}\iint d\omega d\omega^{\prime}(i\omega)[\tilde{\bm{\mathfrak{D}}}(\omega)]_{nr}\left\{ [\tilde{\bm{\mathfrak{D}}}(-\omega)\tilde{\bm{G}}_{H}(-\omega)\tilde{\bm{\mathfrak{D}}}(\omega)]_{\nu k}[\tilde{\bm{\mathfrak{D}}}(\omega^{\prime})\tilde{\bm{G}}_{H}(\omega^{\prime})\tilde{\bm{\mathfrak{D}}}(-\omega^{\prime})]_{lm}\right.\nonumber \\
 & +[\tilde{\bm{\mathfrak{D}}}(-\omega)\tilde{\bm{G}}_{H}(-\omega)\tilde{\bm{\mathfrak{D}}}(\omega)]_{\nu l}[\tilde{\bm{\mathfrak{D}}}(\omega^{\prime})\tilde{\bm{G}}_{H}(\omega^{\prime})\tilde{\bm{\mathfrak{D}}}(-\omega^{\prime})]_{km}\nonumber \\
 & +\left.[\tilde{\bm{\mathfrak{D}}}(-\omega)\tilde{\bm{G}}_{H}(-\omega)\tilde{\bm{\mathfrak{D}}}(\omega)]_{\nu m}[\tilde{\bm{\mathfrak{D}}}(\omega^{\prime})\tilde{\bm{G}}_{H}(\omega^{\prime})\tilde{\bm{\mathfrak{D}}}(-\omega^{\prime})]_{kl}\right\} ,\label{eq:Ups-frequency}
\end{align}
\begin{align}
\tilde{\Upsilon}_{\nu klm}^{nr} & =\frac{1}{(2\pi)^{2}}\iint d\omega d\omega^{\prime}(-i\omega)[\tilde{\bm{\mathfrak{D}}}(\omega)]_{\nu r}\left\{ [\tilde{\bm{\mathfrak{D}}}(-\omega)\tilde{\bm{G}}_{H}(-\omega)\tilde{\bm{\mathfrak{D}}}(\omega)]_{nk}[\tilde{\bm{\mathfrak{D}}}(\omega^{\prime})\tilde{\bm{G}}_{H}(\omega^{\prime})\tilde{\bm{\mathfrak{D}}}(-\omega^{\prime})]_{lm}\right.\nonumber \\
 & +[\tilde{\bm{\mathfrak{D}}}(-\omega)\tilde{\bm{G}}_{H}(-\omega)\tilde{\bm{\mathfrak{D}}}(\omega)]_{nl}[\tilde{\bm{\mathfrak{D}}}(\omega^{\prime})\tilde{\bm{G}}_{H}(\omega^{\prime})\tilde{\bm{\mathfrak{D}}}(-\omega^{\prime})]_{km}\nonumber \\
 & +\left.[\tilde{\bm{\mathfrak{D}}}(-\omega)\tilde{\bm{G}}_{H}(-\omega)\tilde{\bm{\mathfrak{D}}}(\omega)]_{nm}[\tilde{\bm{\mathfrak{D}}}(\omega^{\prime})\tilde{\bm{G}}_{H}(\omega^{\prime})\tilde{\bm{\mathfrak{D}}}(-\omega^{\prime})]_{kl}\right\} ,\label{eq:Ups-tilde-frequency}
\end{align}
\begin{align}
\text{\ensuremath{\Lambda_{klm}^{n}}}= & \frac{1}{(2\pi)^{2}}\iint d\omega d\omega^{\prime}(i\omega)\left\{ [\tilde{\bm{\mathfrak{D}}}(\omega)\tilde{\bm{G}}_{H}(\omega)\tilde{\bm{\mathfrak{D}}}(-\omega)]_{nk}[\tilde{\bm{\mathfrak{D}}}(\omega^{\prime})\tilde{\bm{G}}_{H}(\omega^{\prime})\tilde{\bm{\mathfrak{D}}}(-\omega^{\prime})]_{lm}\right.\nonumber \\
+ & [\tilde{\bm{\mathfrak{D}}}(\omega)\tilde{\bm{G}}_{H}(\omega)\tilde{\bm{\mathfrak{D}}}(-\omega)]_{nl}[\tilde{\bm{\mathfrak{D}}}(\omega^{\prime})\tilde{\bm{G}}_{H}(\omega^{\prime})\tilde{\bm{\mathfrak{D}}}(-\omega^{\prime})]_{km}\nonumber \\
+ & \left.[\tilde{\bm{\mathfrak{D}}}(\omega)\tilde{\bm{G}}_{H}(\omega)\tilde{\bm{\mathfrak{D}}}(-\omega)]_{nm}[\tilde{\bm{\mathfrak{D}}}(\omega^{\prime})\tilde{\bm{G}}_{H}(\omega^{\prime})\tilde{\bm{\mathfrak{D}}}(-\omega^{\prime})]_{kl}\right\} .\label{eq:Lambda-frequency}
\end{align}

One can check that Eqs. (\ref{eq:Pxi0-frequency}-\ref{eq:Lambda-frequency})
are indeed Fourier transform of Eqs. (\ref{eq:Pxi0}, \ref{eq:pn2-zeroth},
\ref{eq:W-n2-zeroth},\ref{eq:Gamma-nr-klm-final}, \ref{eq:Gamma-tilde-nr-klm-final},
\ref{eq:Ups-nr-klm-final}, \ref{eq:Ups-tilde-nr-klm-final}, \ref{eq:Lambda-final}
) respectively, with the Fourier transform defined as $f(s)=1/(2\pi)\int_{-\infty}^{\infty}\tilde{f}(\omega)e^{i\omega s}d\omega$.
From Fig. \ref{fig:Ups-UpsTilde-Frequency} or explicit expressions
(\ref{eq:Ups-frequency}, \ref{eq:Ups-tilde-frequency}), one can
easily observe, 
\begin{equation}
\Upsilon_{nklm}^{nr}=-\tilde{\Upsilon}_{nklm}^{nr},
\end{equation}
which simplifies Eq. \eqref{eq:wn2-first-order} as 
\begin{figure}
\begin{picture}(450,100)
\begin{centering}
\put(0,0){\includegraphics[scale=0.4]{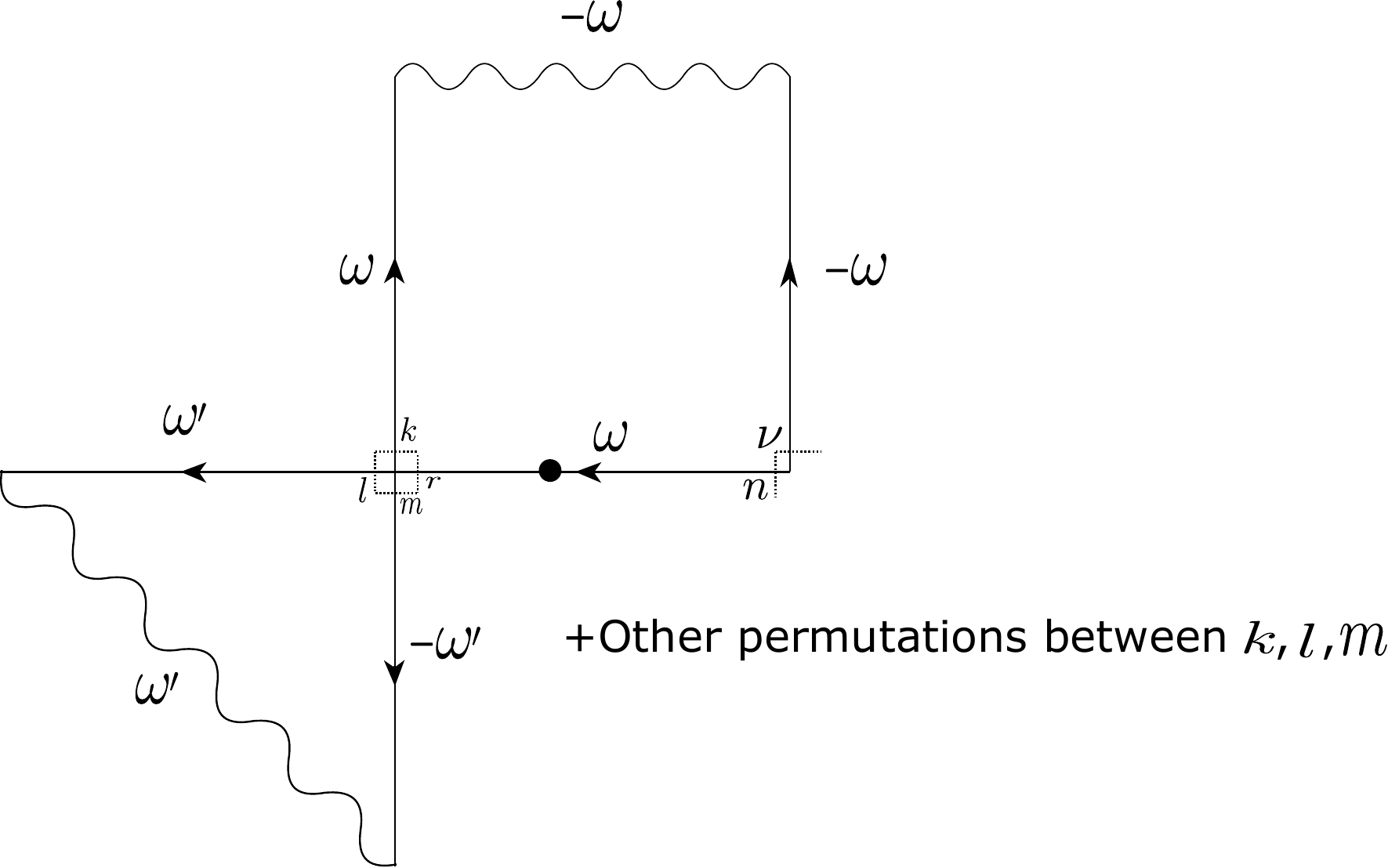}}\put(230,0){\includegraphics[scale=0.4]{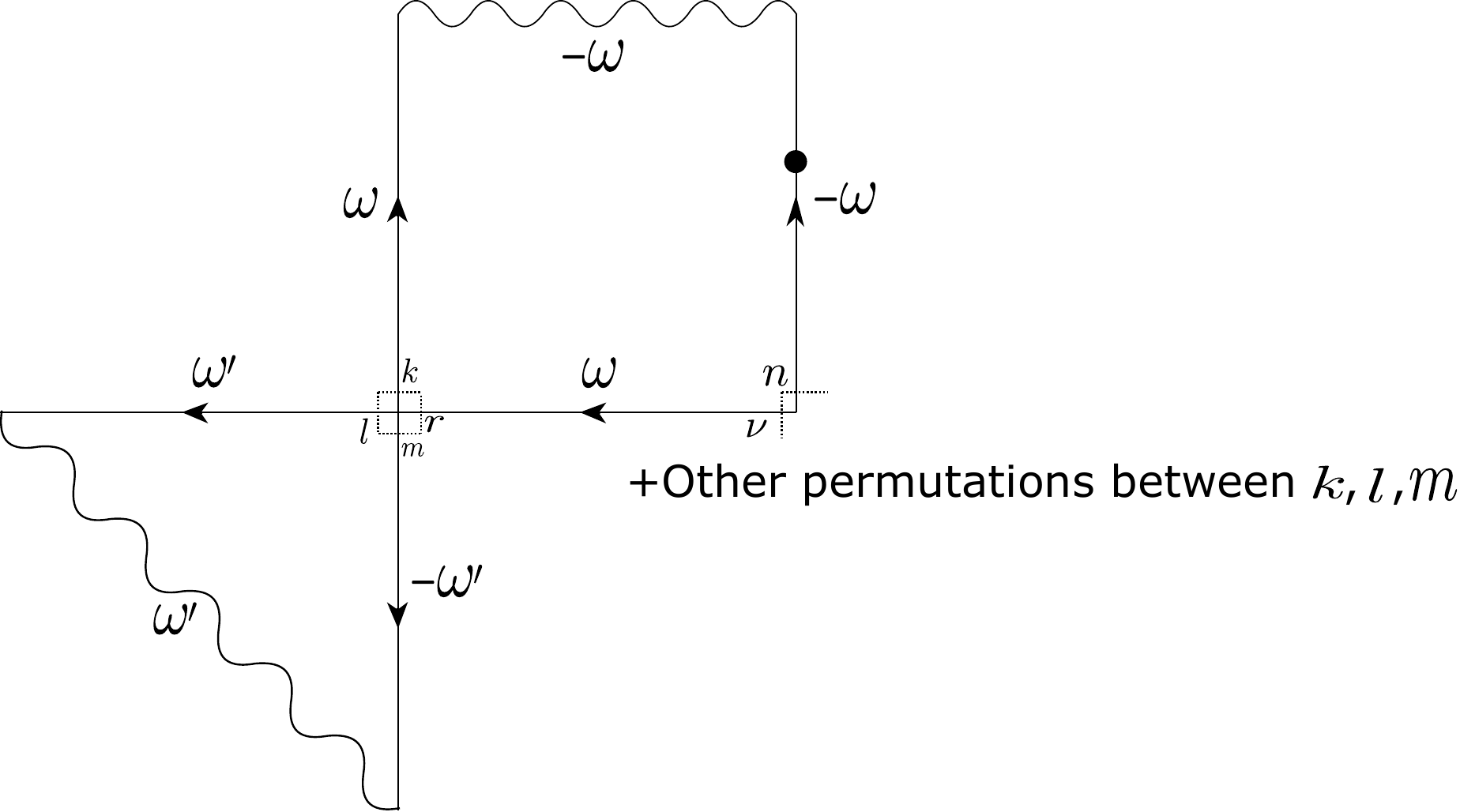}}
\par\end{centering}
\put(25,80){\footnotesize{}(a)} \put(250,80){\footnotesize{}(b)}
\end{picture} \centering{}\caption{\label{fig:Ups-UpsTilde-Frequency}Frequency-domain diagrammatic representation
of (a) $\Upsilon_{klm}^{nr}(\infty)$ and (b) $\tilde{\Upsilon}_{\nu klm}^{nr}(\infty)$.}
\end{figure}
\begin{equation}
w_{\nu\to n}^{(1,\,2)}=\sum_{klmr}\mu_{klmr}[\Upsilon_{\nu klm}^{nr}+\tilde{\Upsilon}_{\nu klm}^{nr}].\label{eq:wn2-first-simplified}
\end{equation}

\section{\label{sec:NESS}The non-equilibrium steady state (NESS)}

For the configuration we are interested, if the NESS exists, then
we expect when the NESS is reached we will have a steady, time-independent
energy current along the chain. Thus, in principle, in order to demonstrate
the existence of the NESS, we would like to show for the configuration
in consideration that in the late time limit $t_{f}\to\infty$, the
time rate of the internal energy of each oscillator vanishes. In other
words, we will show that the ensemble average of Eq. (\ref{eq:dEdt})
vanishes in the late time limit $t_{f}\to\infty$, such that 
\begin{equation}
P_{\xi_{n}}+P_{\gamma_{n}}+P_{n-1\to n}+P_{n+1\to n}=0\,,\label{eq:NESS}
\end{equation}
for each $n$.

However, we will discuss a particular case, for the proof of concept,
that involves only two oscillators in contact with their own private
baths. We will show Eq.~(\ref{eq:NESS}) holds in the late time limit,
up to the first order. This is our main result. We start with the
zeroth-order case, which has been shown to hold for a chain made of
any number of linear oscillators, strung together by bilinear coupling
\citep{HHAoP}, that is, the quadratic coupling discussed in this
paper. Nevertheless, For sake of completeness, we will still demonstrate
\eqref{eq:NESS} for the zeroth-order contributions based the results
derived by our approach in Sec.~\ref{sec:zeroth-order}. Then we
further proceed to the first-order corrections. 
\begin{figure}
\begin{centering}
\includegraphics[scale=0.5]{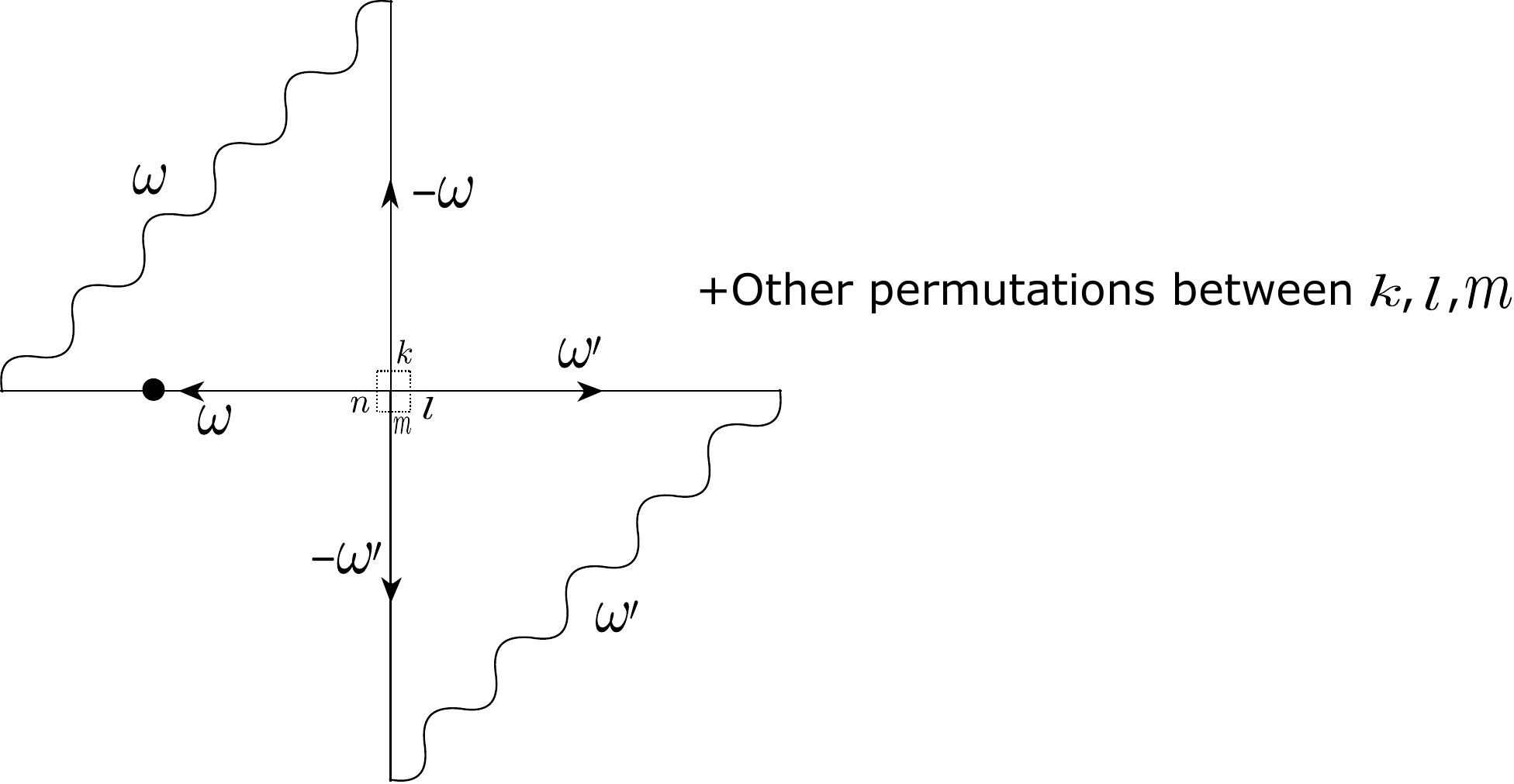}
\par\end{centering}
\centering{}\caption{\label{fig:Lambda-Frequency}Frequency-domain diagrammatic representation
of $\Lambda_{klm}^{n}(\infty)$. According to Eqs.~(\ref{eq:odd-integral},
\ref{eq:Lambda-frequency}), $\Lambda_{nnn}^{n}(\infty)=0$.}
\end{figure}

\subsection{NESS at the zeroth order}

From \eqref{eq:Pxi0-frequency}, \eqref{eq:pn20-frequency}, and \eqref{eq:wn20-frequency-simpified},
we use the property that 1) $\tilde{\bm{G}}_{H}(\omega)$ is diagonal,
and 2) $\operatorname{Im}\tilde{\bm{\mathfrak{D}}}(\omega)$ is odd
in $\omega$, we can write them as 
\begin{equation}
P_{\xi_{n}}^{(0)}=-\int_{-\infty}^{\infty}\!\frac{d\omega}{2\pi}\;\omega\,\operatorname{Im}\Bigl[\tilde{\bm{\mathfrak{D}}}(\omega)\Bigr]_{nn}\Bigl[\tilde{\bm{G}}_{H}(\omega)\Bigr]_{nn},\label{eq:P-xi-n-zero}
\end{equation}

\begin{equation}
P_{\gamma_{n}}^{(0)}=-2\gamma\sum_{k}\int_{-\infty}^{\infty}\!\frac{d\omega}{2\pi}\;\omega^{2}\,\lvert\Bigl[\tilde{\bm{\mathfrak{D}}}(\omega)\Bigr]_{nk}\rvert^{2}\Bigl[\tilde{\bm{G}}_{H}(\omega)\Bigr]_{kk},\label{eq:P-gamma-n-zero}
\end{equation}

\begin{equation}
P_{\nu\to n}^{(0,\,2)}=-\lambda_{2}\sum_{k}\int_{-\infty}^{\infty}\!\frac{d\omega}{2\pi}\;\omega\,\operatorname{Im}\Bigl\{\Bigl[\tilde{\bm{\mathfrak{D}}}(\omega)\Bigr]_{nk}\Bigl[\tilde{\bm{\mathfrak{D}}}^{*}(\omega)\Bigr]_{\nu k}\Bigr\}\Bigl[\tilde{\bm{G}}_{H}(\omega)\Bigr]_{kk}.\label{eq:P-nu-n-2-zero}
\end{equation}
So we obtain 
\begin{align}
 & P_{\xi_{n}}^{(0)}+P_{\gamma_{n}}^{(0)}+P_{\nu\to n}^{(0,\,2)}\nonumber \\
= & -\frac{1}{2\pi}\int_{-\infty}^{\infty}\left\{ \omega\text{Im}[\tilde{\bm{\mathfrak{D}}}(\omega)]_{nn}+2\gamma\omega^{2}\big|[\tilde{\bm{\mathfrak{D}}}(\omega)]_{nn}\big|^{2}+\lambda_{2}\omega\text{Im}\left([\tilde{\bm{\mathfrak{D}}}(\omega)]_{nn}[\tilde{\bm{\mathfrak{D}}}(\omega)]_{\nu n}^{*}\right)\right\} \Bigl[\tilde{\bm{G}}_{H}(\omega)\Bigr]_{nn}\nonumber \\
- & \frac{1}{2\pi}\int_{-\infty}^{\infty}\left\{ 2\gamma\omega^{2}\big|[\tilde{\bm{\mathfrak{D}}}(\omega)]_{n\nu}\big|^{2}+\lambda_{2}\omega\text{Im}\left([\tilde{\bm{\mathfrak{D}}}(\omega)]_{n\nu}[\tilde{\bm{\mathfrak{D}}}(\omega)]_{\nu\nu}^{*}\right)\right\} \Bigl[\tilde{\bm{G}}_{H}(\omega)\Bigr]_{\nu\nu}.
\end{align}
Now by replacing eliminating $[\tilde{\bm{\mathfrak{D}}}(\omega)]_{\nu n}^{*}$
through Eqs. (\ref{eq:off-diag-ME}, \ref{eq:Inverse1}), one can
easily show that 
\begin{equation}
\omega\text{Im}[\tilde{\bm{\mathfrak{D}}}(\omega)]_{nn}+2\gamma\omega^{2}\big|[\tilde{\bm{\mathfrak{D}}}(\omega)]_{nn}\big|^{2}+\lambda_{2}\omega\text{Im}\left([\tilde{\bm{\mathfrak{D}}}(\omega)]_{nn}[\tilde{\bm{\mathfrak{D}}}(\omega)]_{\nu n}^{*}\right)=0.\label{E:dgbrhtrdf1}
\end{equation}
Similarly, by eliminating $[\tilde{\bm{\mathfrak{D}}}(\omega)]_{\nu\nu}^{*}$
through Eqs. (\ref{eq:diag-ME}, \ref{eq:Inverse2}), one can easily
prove 
\begin{equation}
2\gamma\omega^{2}\big|[\tilde{\bm{\mathfrak{D}}}(\omega)]_{n\nu}\big|^{2}+\lambda_{2}\omega\text{Im}\left([\tilde{\bm{\mathfrak{D}}}(\omega)]_{n\nu}[\tilde{\bm{\mathfrak{D}}}(\omega)]_{\nu\nu}^{*}\right)=0.\label{E:dgbrhtrdf2}
\end{equation}
Eqs.~\eqref{E:dgbrhtrdf1} and \eqref{E:dgbrhtrdf2} together prove
that for both $n=1,\,2$, 
\begin{equation}
P_{\xi_{n}}^{(0)}+P_{\gamma_{n}}^{(0)}+P_{\nu\to n}^{(0,\,2)}=0\,.
\end{equation}

\subsection{NESS at the first order\label{S:mgbfdmg}}

Here we address existence of the NESS at the first order of the nonlinear
coupling constant. Since at this order, the contributions from the
nonlinear couplings are additive, and the $\alpha$-FPUT coupling
does not contribute, we consider only the KG- and the $\beta$-FPUT
nonlinearity, and will show $P_{\xi_{n}}^{(1)}+P_{\gamma_{n}}^{(1)}+P_{\nu\to n}^{(1)}=0$
for $n=1$ and $n=2$
\begin{figure}
\begin{picture}(400,400)

\put(-30,200){\includegraphics[scale=0.32]{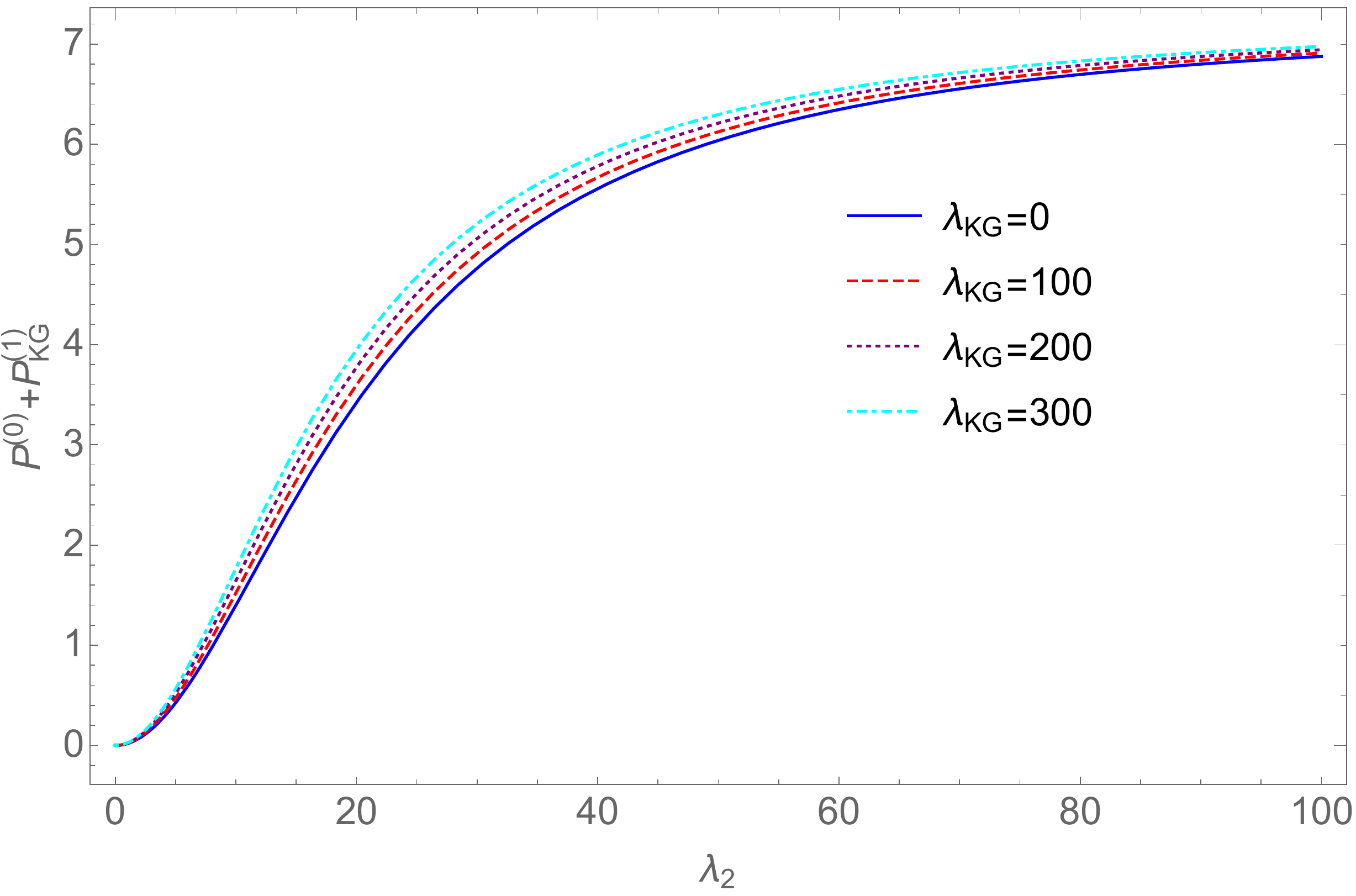}}\put(220,200){\includegraphics[scale=0.42]{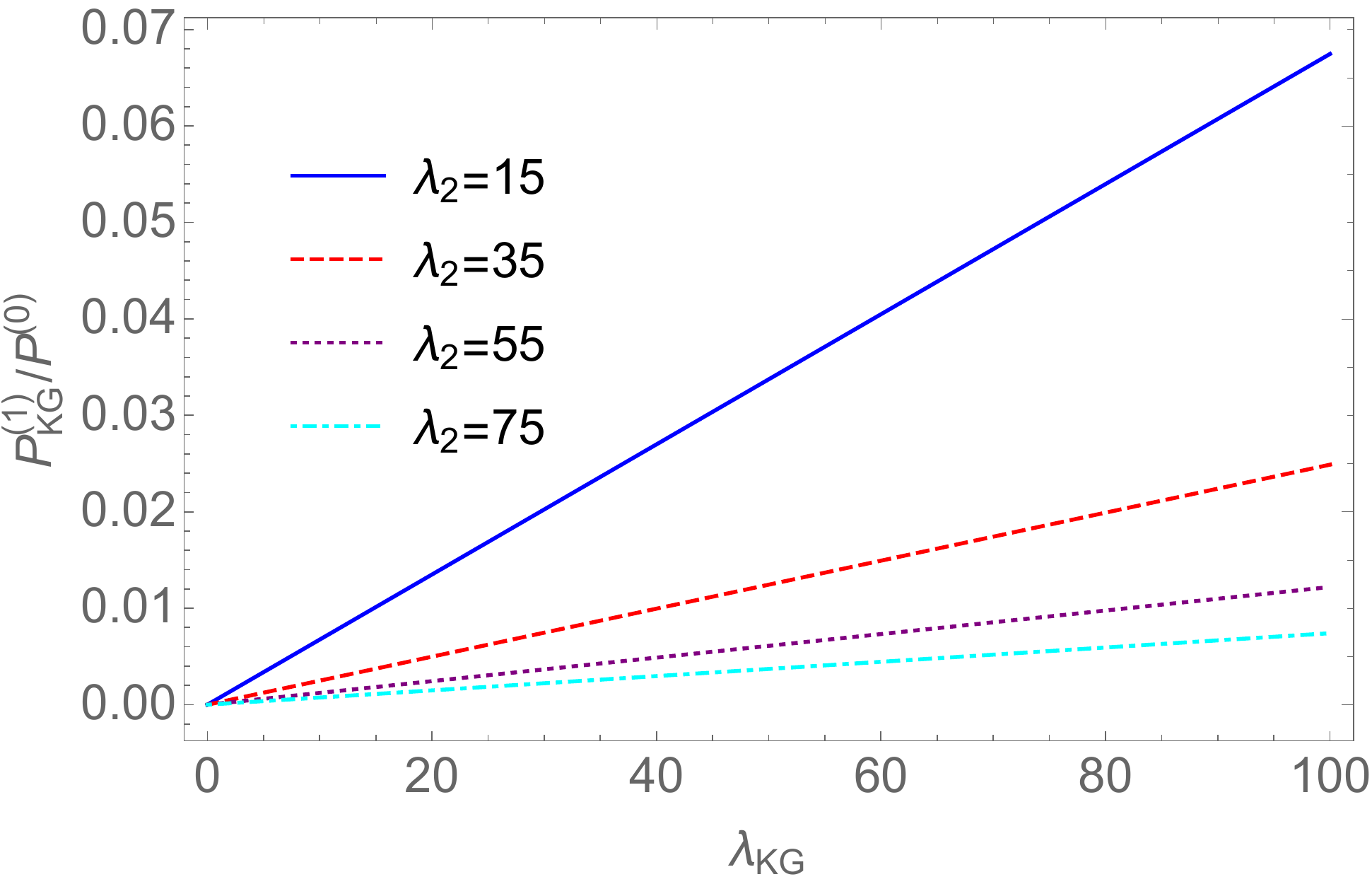}}

\put(0,0){\includegraphics[scale=0.32]{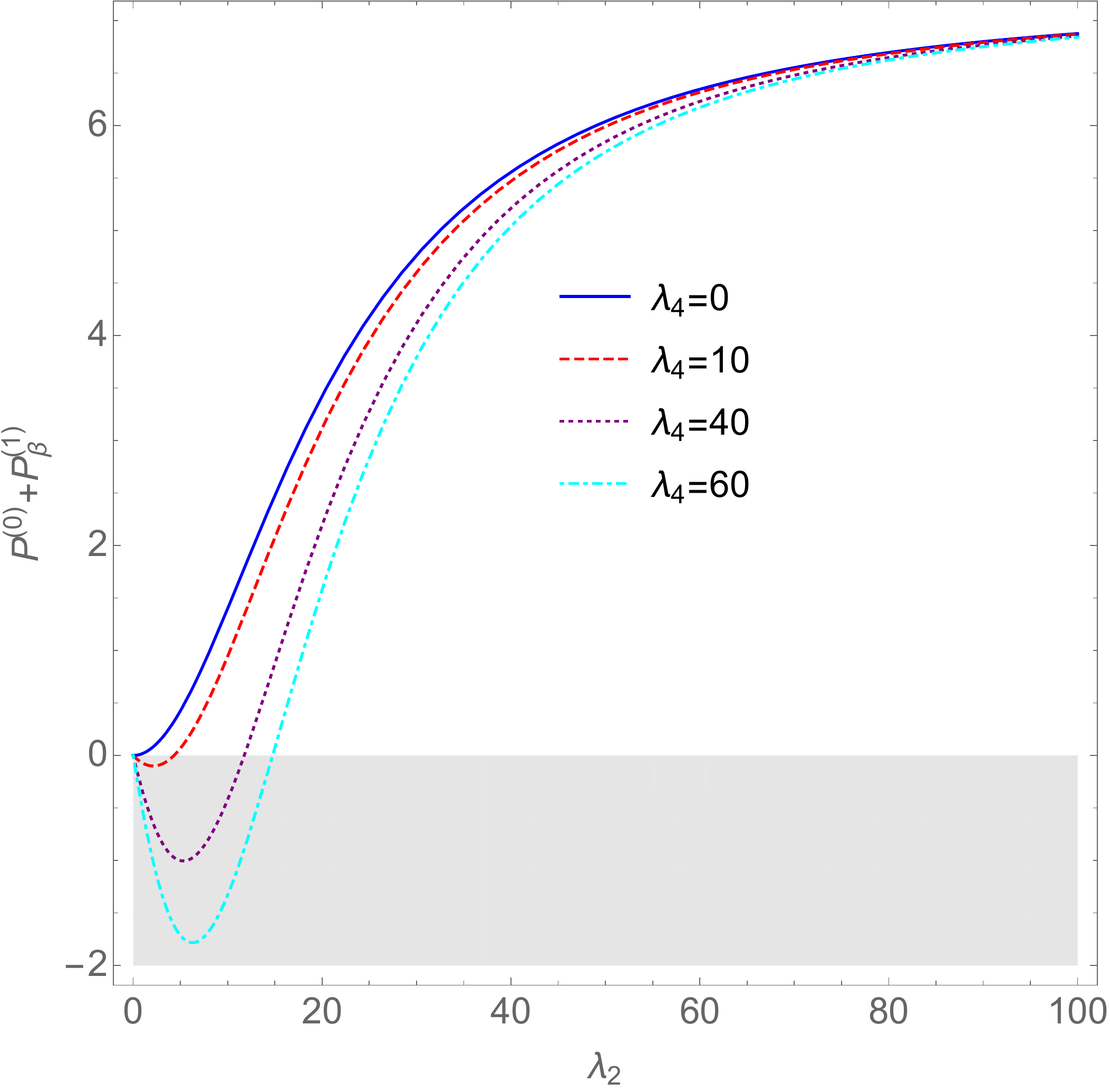}}\put(220,0){\includegraphics[scale=0.36]{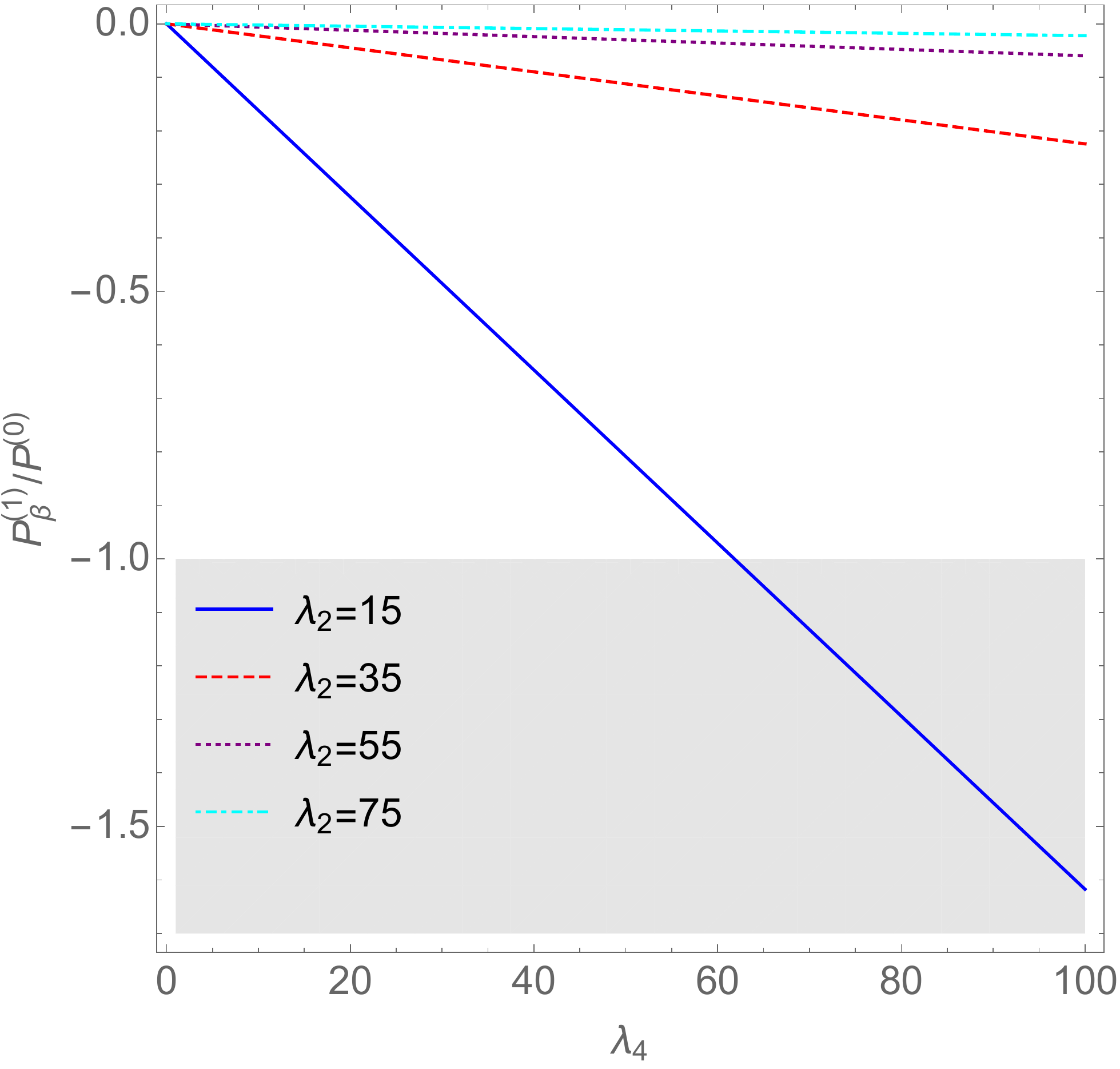}}

\put(40,330){\footnotesize{}(a)} \put(260,330){\footnotesize{}(b)}
\put(50,180){\footnotesize{}(c)} \put(260,180){\footnotesize{}(d)}
\end{picture}

\caption{\label{fig:Energy-current}Energy current up to the first order across
the two oscillators for (a) KG and (c) $\beta$-FPUT nonlinearities.
The ratio between the first-order energy current and the zeroth for
(b) KG and (d) $\beta$-FPUT nonlinearities. Value of parameters for
all the figures: $\omega_{0}=10$; $\gamma=1$; the temperature of
the hot bath is $T_{H}=\beta_{1}^{-1}=100$ and the temperature of
the cold bath is $T_{C}=\beta_{2}^{-1}=0.002$. Note the first-order
correction is positive for KG nonlinearity and negative $\beta$-FPUT
nonlinearity. Therefore for $\beta$-FPUT when the ratio $|P_{\beta}^{(1)}/P^{(0)}|\ge1$,
as shown in the gray shaded region in Figs. (c, d), the energy current
up to the first-order may become negative, meaning heat flow from
the cold bath to hot bath. However, we note that this effect is spurious
as whenever, $|P_{\beta}^{(1)}/P^{(0)}|\sim1$, our perturbative calculation
breaks down.}
\end{figure}
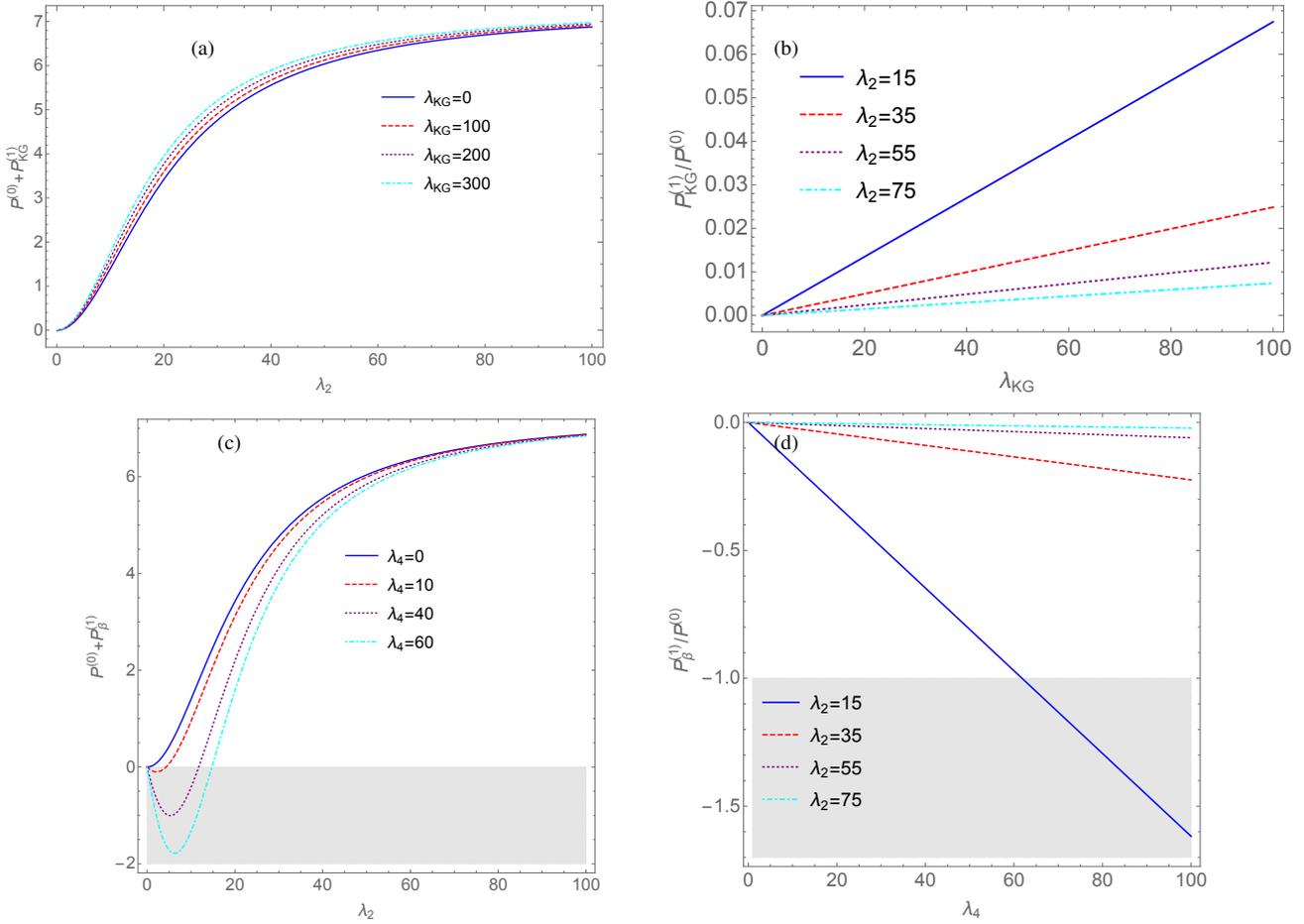

\subsubsection{KG nonlinearity}

When only KG nonlinearity is presented, we in fact have 
\begin{equation}
\sum_{klmr}\sigma_{klmr}r_{k}(s)q_{l}(s)q_{m}(s)q_{r}(s)=-\frac{\lambda_{\text{KG}}}{4}\sum_{k=1}^{2}r_{k}q_{k}^{3},\label{eq:KG-sigma-sum}
\end{equation}
\begin{equation}
\sum_{klmr}\mu_{klmr}r_{k}(s)r_{l}(s)r_{m}(s)q_{r}(s)=-\lambda_{\text{KG}}\sum_{k=1}^{2}r_{k}^{3}q_{k}.\label{eq:KG-mu-sum}
\end{equation}
Substitution of Eq. (\ref{eq:KG-mu-sum}) into Eqs.~\eqref{eq:Pxi-first-order},
\eqref{eq:Pgamma-first-order}, yields 
\begin{align}
P_{\xi_{n}}^{(1)} & =-\lambda_{\text{KG}}\sum_{k=1}^{2}\Gamma_{kkk}^{nk}\,, & P_{\gamma_{n}}^{(1)} & =2\gamma\lambda_{\text{KG}}\sum_{k=1}^{2}\tilde{\Gamma}_{kkk}^{nk}\,,
\end{align}
and from \eqref{eq:P-nu-eta-ave} and \eqref{eq:wn2-first-order},
we also obtain 
\begin{equation}
P_{\nu\to n}^{(1,\,2)}=-\lambda_{\text{KG}}\lambda_{2}\sum_{k=1}^{2}(\Upsilon_{\nu kkk}^{nk}+\tilde{\Upsilon}_{\nu kkk}^{nk}).\label{E:fgbkr}
\end{equation}
As a reminder, $P_{\nu\to n}^{(1,\,2)}$ is the first-order correction
to the energy flow $P_{\nu\to n}^{(2)}$ due to the quadratic intra-oscillator
coupling. Note that although both KG and $\beta-$FPUT are quartic
interactions, the KG nonlinearity only involves self-interaction at
each site and therefore, unlike the $\beta-$FPUT interactions, does
not induce $P_{\nu\to n}^{(4)}$. In  \ref{sec:NESS-Identities},
we have shown the following the identities 
\begin{eqnarray}
-\Gamma_{nnn}^{nn}+2\gamma\tilde{\Gamma}_{nnn}^{nn}-\lambda_{2}(\Upsilon_{\nu nnn}^{nn}+\tilde{\Upsilon}_{\nu nnn}^{nn}) & = & 0,\label{eq:NESS-Id0}\\
-\Gamma_{\nu\nu\nu}^{n\nu}+2\gamma\tilde{\Gamma}_{\nu\nu\nu}^{n\nu}-\lambda_{2}(\Upsilon_{\nu\nu\nu\nu}^{n\nu}+\tilde{\Upsilon}_{\nu\nu\nu\nu}^{n\nu}) & = & 0.\label{eq:NESS-Id1}
\end{eqnarray}
With these identity, we see that 
\begin{equation}
P_{\xi_{n}}^{(1)}+P_{\gamma_{n}}^{(1)}+P_{\nu\to n}^{(1,\,2)}=0,\label{eq:1order-balance-KG}
\end{equation}
for both $n=1$ and $n=2$ in the presence of the KG-type nonlinearity.
\begin{figure}
\begin{picture}(400,400)

\put(-30,210){\includegraphics[scale=0.39]{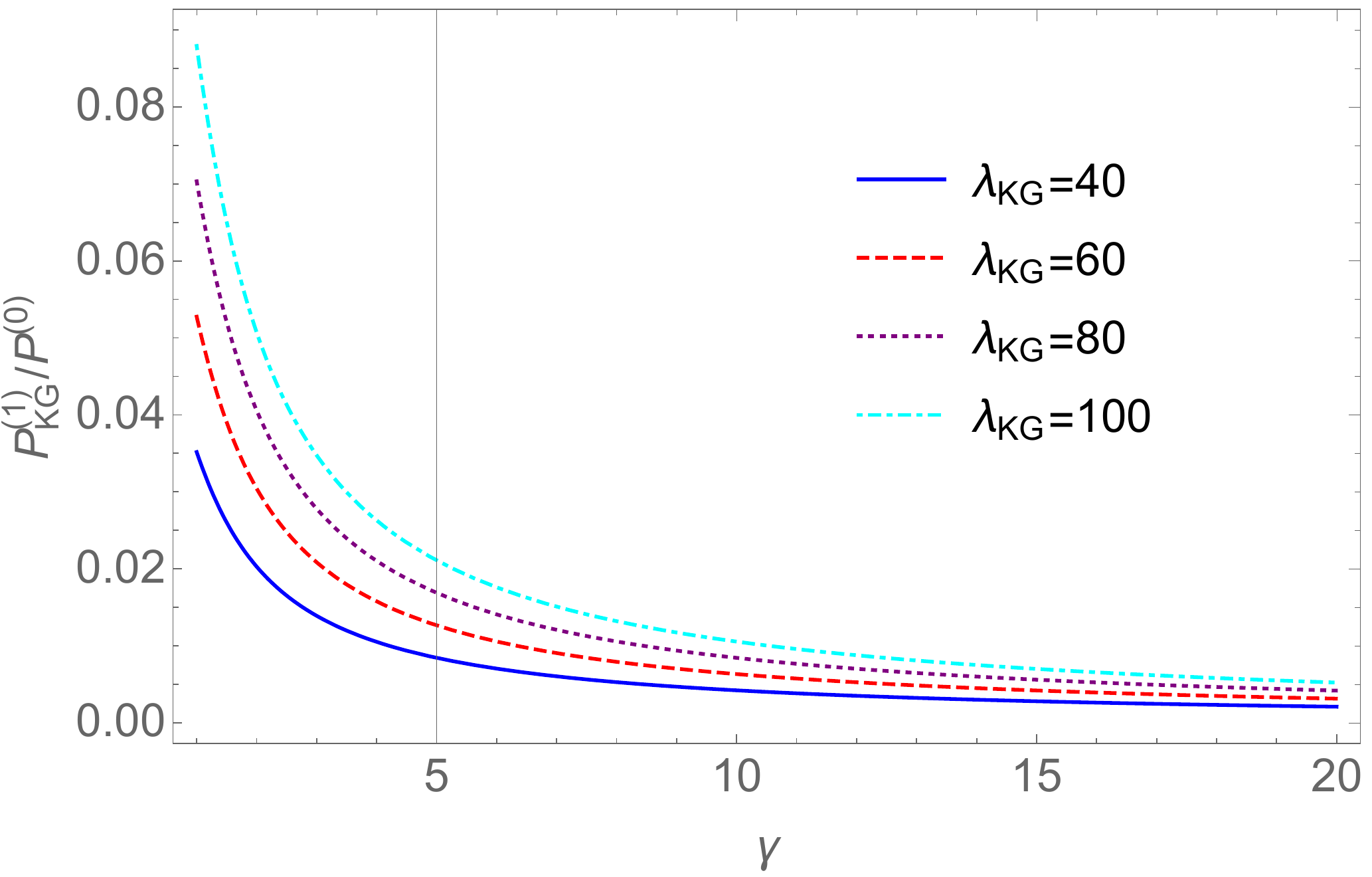}}\put(220,210){\includegraphics[scale=0.38]{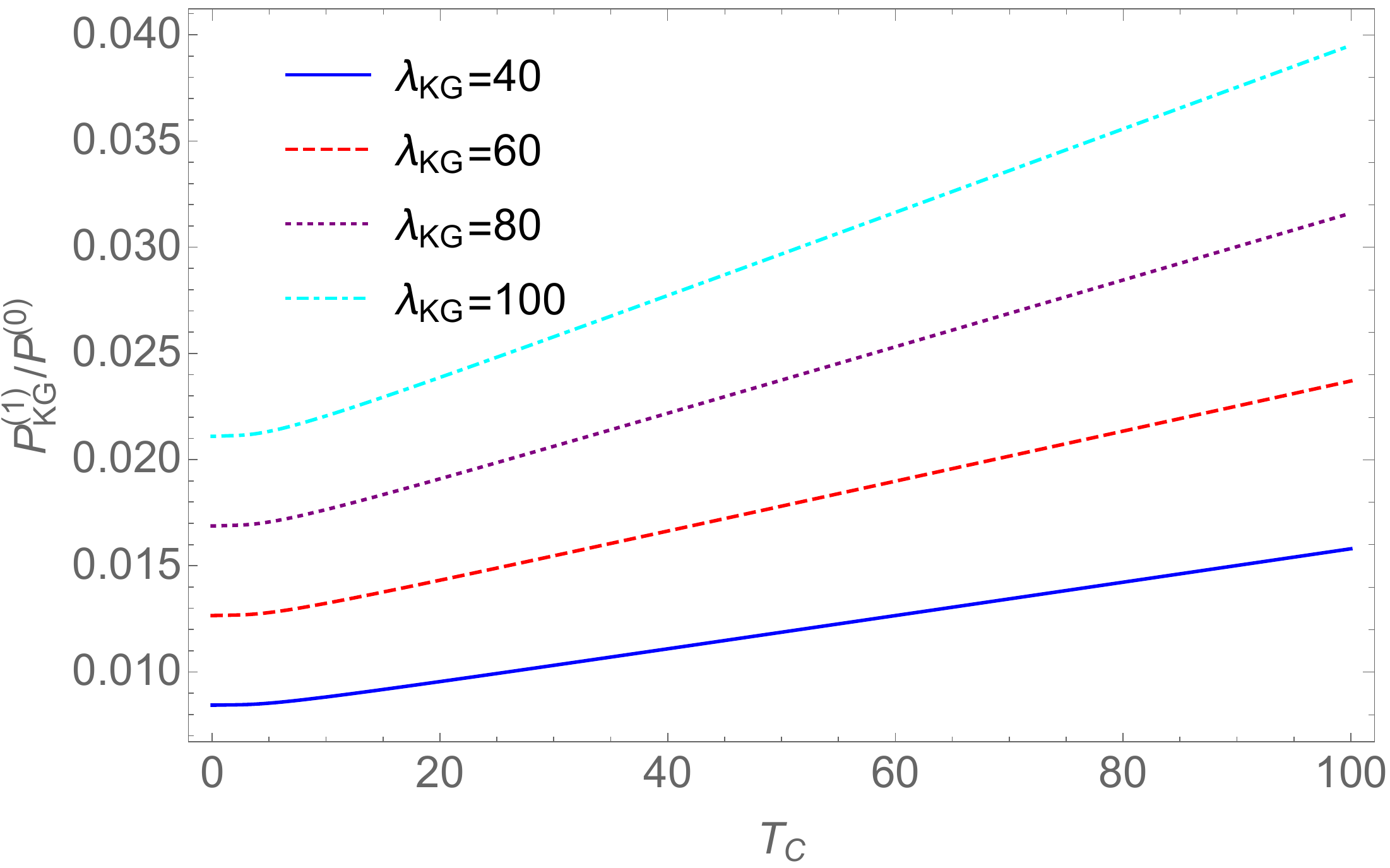}}

\put(0,0){\includegraphics[scale=0.42]{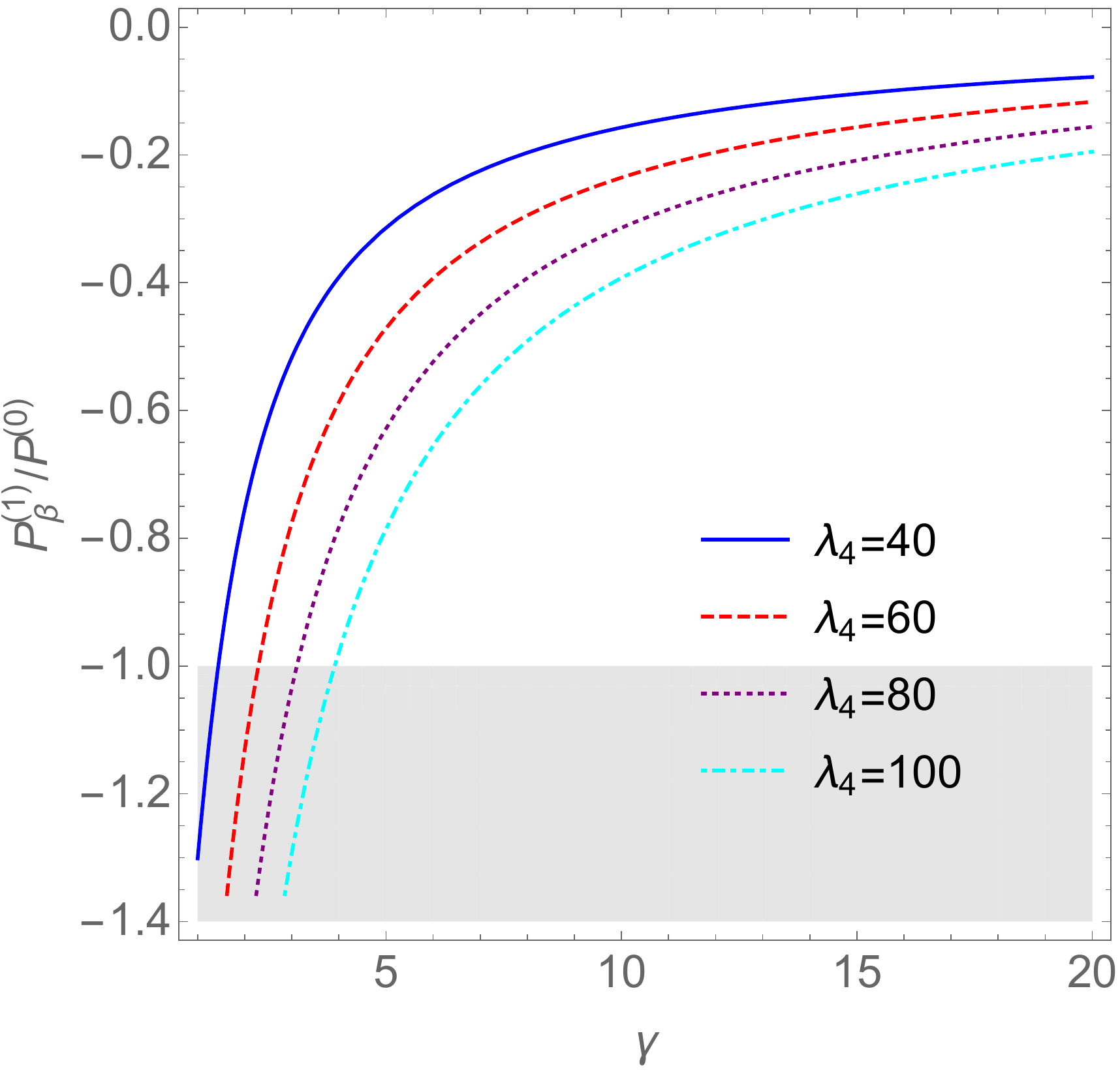}}\put(220,0){\includegraphics[scale=0.38]{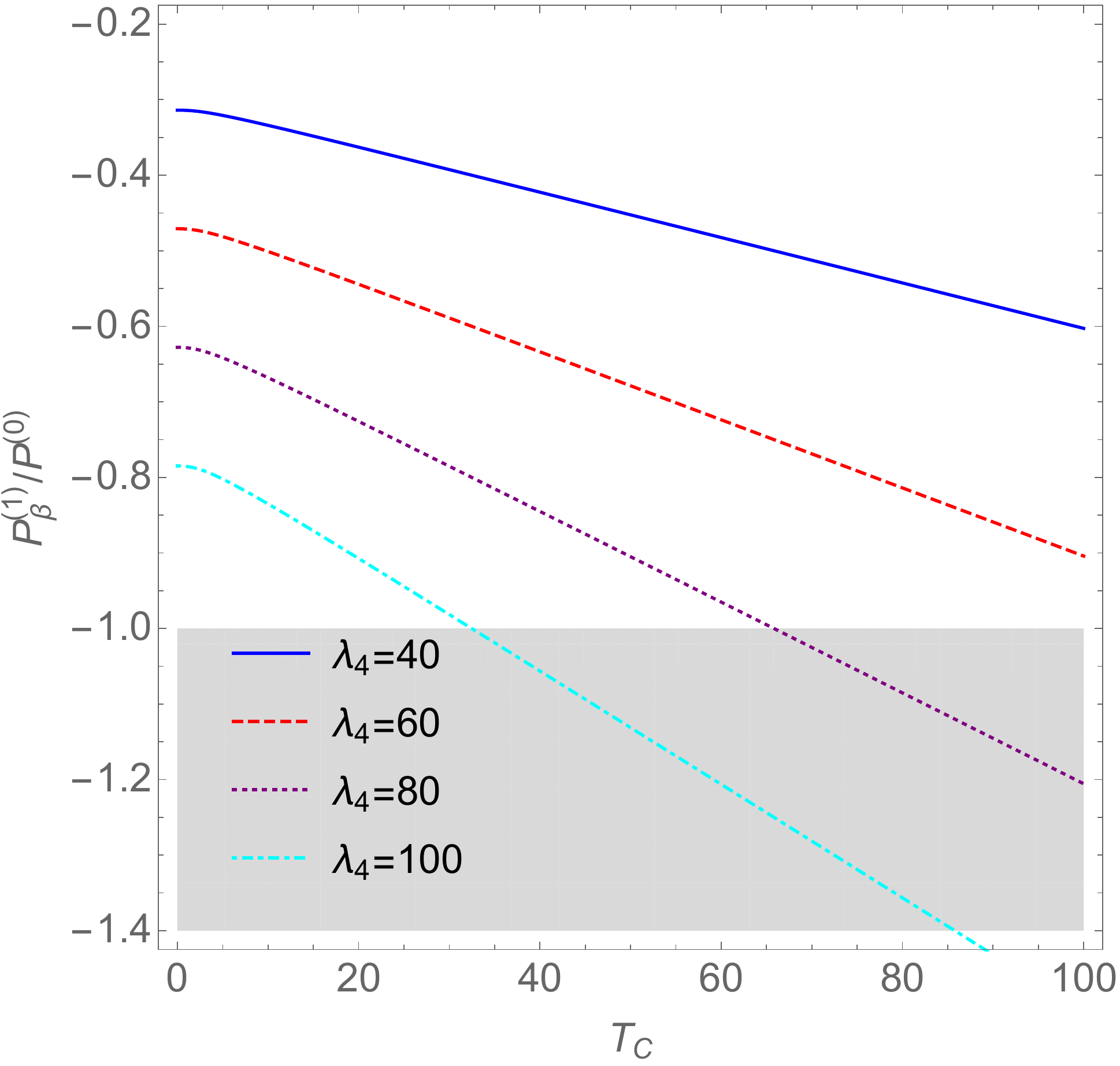}}

\put(40,330){\footnotesize{}(a)} \put(260,330){\footnotesize{}(b)}
\put(50,180){\footnotesize{}(c)} \put(260,180){\footnotesize{}(d)}
\end{picture}

\caption{\label{fig:ratio-tuning}The role of coupling (damping) and the temperature
bias in mediating the strength of (a, b) KG and (c, d) $\beta$-FPUT
nonlinearities. Common values of the parameters for all figures: $\omega_{0}=\lambda_{2}=10$,
the temperature of the hot bath is $T_{H}=\beta_{1}^{-1}=100$. For
(a) and (c): the temperature of the cold bath is $T_{C}=\beta_{2}^{-1}=0.002$.
For (b) and (d), $\gamma=5$. The gray shaded region shown in Figs.
(c, d) corresponding to $|P_{\beta}^{(1)}/P^{(0)}|\ge1$, where our
perturbative calculation breaks down. Therefore, the seemingly heat
flow from the cold bath to the hot bath is artificial.}
\end{figure}
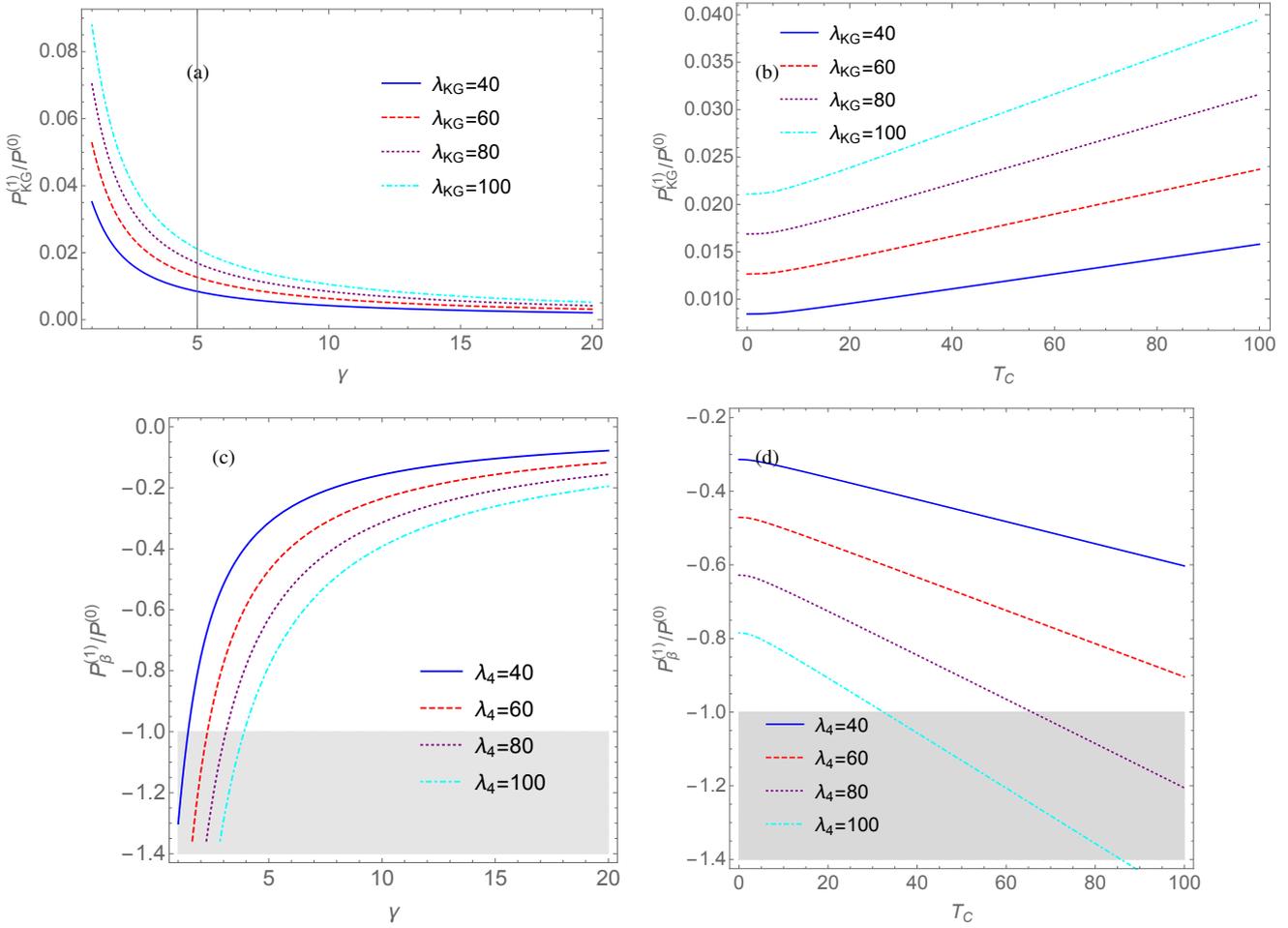

\subsubsection{$\beta$-FPUT nonlinearity}

Next we move on to the $\beta$-FPUT nonlinearity. When only the $\beta$-FPUT
nonlinearity is presented, we can show that terms on the righthand
side of Eq.~(\ref{eq:calJ-xi}) becomes 
\begin{equation}
\sum_{klmr}\sigma_{klmr}r_{k}(s)q_{l}(s)q_{m}(s)q_{r}(s)=-\frac{\lambda_{4}}{4}(q_{n}-q_{\nu})(r_{n}-r_{\nu})^{3},\label{eq:beta-sigma-sum}
\end{equation}
\begin{equation}
\sum_{klmr}\mu_{klmr}r_{k}(s)r_{l}(s)r_{m}(s)q_{r}(s)=-\lambda_{4}(q_{n}-q_{\nu})(r_{n}-r_{\nu})^{3},\label{eq:beta-mu-sum}
\end{equation}
where the pairs $(n,\,\nu)$ could be either $(1,\,2)$ or $(2,\,1)$.
Substituting Eq.~(\ref{eq:beta-mu-sum}) into Eqs. (\ref{eq:Pxi-first-order},
\ref{eq:Pgamma-first-order}, \ref{eq:wn2-first-order}), we obtain
the first-order correction for the following energy fluxes 
\begin{align}
P_{\xi_{n}}^{(1)} & =-\lambda_{4}\,c_{n}(\Gamma)\,, & P_{\gamma_{n}}^{(1)} & =2\gamma\lambda_{4}\,c_{n}(\tilde{\Gamma})\,, & P_{\nu\to n}^{(1,\,2)} & =-\lambda_{2}\lambda_{4}\,\Bigl[f_{\nu n}(\Upsilon)+f_{\nu n}(\tilde{\Upsilon})\Bigr]\,,
\end{align}
where 
\begin{align}
c_{n}(\Gamma) & \equiv\Gamma_{nnn}^{nn}-3\Gamma_{nn\nu}^{nn}+3\Gamma_{n\nu\nu}^{nn}-\Gamma_{\nu\nu\nu}^{nn}-\Gamma_{nnn}^{n\nu}+3\Gamma_{nn\nu}^{n\nu}-3\Gamma_{n\nu\nu}^{n\nu}+\Gamma_{\nu\nu\nu}^{n\nu},\label{eq:c}\\
f_{\nu n}(\Upsilon) & \equiv\Upsilon_{\nu nnn}^{nn}-3\Upsilon_{\nu nn\nu}^{nn}+3\Upsilon_{\nu n\nu\nu}^{nn}-\Upsilon_{\nu\nu\nu\nu}^{nn}-\Upsilon_{\nu nnn}^{n\nu}+3\Upsilon_{\nu nn\nu}^{n\nu}-3\Upsilon_{\nu n\nu\nu}^{n\nu}+\Upsilon_{\nu\nu\nu\nu}^{n\nu}.\label{eq:f}
\end{align}
Eqs.~(\ref{eq:odd-integral}, \ref{eq:Lambda-frequency}) indicates
that $\Lambda_{nnn}^{n}=0$. Thus, from Eqs. (\ref{eq:P-nu-eta-ave},
\ref{eq:wbeta-zero-order}), we find 
\begin{equation}
P_{\nu\to n}^{(1,\,4)}=\lambda_{4}[3\Lambda_{nn\nu}^{n}-3\Lambda_{n\nu\nu}^{n}+\Lambda_{\nu\nu\nu}^{n}].
\end{equation}
Here $P_{\nu\to n}^{(1,\,4)}$ is the first-order correction to the
energy flow due to the quartic intra-oscillator coupling. In  \ref{sec:NESS-Identities},
we show that 
\begin{eqnarray}
-\Gamma_{klm}^{n\nu}+2\gamma\tilde{\Gamma}_{klm}^{n\nu}-\lambda_{2}(\Upsilon_{\nu klm}^{n\nu}+\tilde{\Upsilon}_{\nu klm}^{n\nu}) & = & 0,\label{eq:NESS-Id2}\\
-\Gamma_{nk\nu}^{nn}+2\gamma\tilde{\Gamma}_{nk\nu}^{nn}-\lambda_{2}(\Upsilon_{\nu nk\nu}^{nn}+\tilde{\Upsilon}_{\nu nk\nu}^{nn}) & = & \Lambda_{nk\nu}^{n},\label{eq:NESS-Id3}\\
-\Gamma_{\nu\nu\nu}^{nn}+2\gamma\tilde{\Gamma}_{\nu\nu\nu}^{nn}-\lambda_{2}(\Upsilon_{\nu\nu\nu\nu}^{nn}+\tilde{\Upsilon}_{\nu\nu\nu\nu}^{nn}) & = & \Lambda_{\nu\nu\nu}^{n},\label{eq:NESS-Id4}
\end{eqnarray}
where $k,\,l,\,m=n,\,\nu$. With Eqs. (\ref{eq:NESS-Id0}, \ref{eq:NESS-Id2}-\ref{eq:NESS-Id4}),
we readily obtain 
\begin{equation}
P_{\xi_{n}}^{(1)}+P_{\gamma_{n}}^{(1)}+P_{\nu\to n}^{(1,\,2)}+P_{\nu\to n}^{(1,\,4)}=0\,,\label{E:kbgksrt}
\end{equation}
for $n=1,\,2$.

Together with $P_{\nu\to n}^{(1,\,4)}$, the net energy flow from
site $\nu$ to site $n$ to the first order in the $\beta$-FPUT intra-oscillator
coupling, we can define 
\begin{equation}
P_{\nu\to n}^{(\le1)}=P_{\nu\to n}^{(0,\,2)}+P_{\nu\to n}^{(1,\,2)}+P_{\nu\to n}^{(1,\,4)}\,,
\end{equation}
representing the energy flow from site $\nu$ to site $n$, up to
the first order. Note that the $P_{\nu\to n}^{(0,\,4)}=0$ according
to the definition \eqref{eq:P-nu-eta-ave}. We define 
\begin{equation}
P_{\text{B}_{n}}^{(\le1)}=P_{\xi_{n}}^{(0)}+P_{\gamma_{n}}^{(0)}+P_{\xi_{n}}^{(1)}+P_{\gamma_{n}}^{(1)}\,,
\end{equation}
the net energy flow into oscillator $i$ from its private bath, up
to the first order. Thus \eqref{E:kbgksrt} becomes 
\begin{equation}
P_{\textsc{b}_{n}}^{(\le1)}=-P_{\nu\to n}^{(\le1)}
\end{equation}
at late times. It means we have a steady energy flow through each
oscillator. The direction of the flow will be determined by the temperatures
of the private baths. As seen from~\citep{HHAoP,Paz}, at the zeroth-order,
energy always flows from the higher temperature bath to the lower
one. Here, even we consider its first-order correction due to nonlinear
couplings, within the validity of the perturbative treatment, the
corrections are required much smaller than the zeroth-order contribution,
and thus will not change the direction of the flow.

We conclude this section by noting that the zeroth-order heat fluxes,
$P_{\xi_{n}}^{(0)}$ and $P_{\gamma_{n}}^{(0)}$ given in Eqs. (\ref{eq:P-xi-n-zero},
\ref{eq:P-gamma-n-zero}) diverge, which can be seen by counting the
powers of $\tilde{\bm{G}}_{H}(\omega)$ and $\tilde{\bm{\mathfrak{D}}}(\omega)$
at large $\omega$. Therefore, we need to perform regularization in
order to evaluate them. We see that their regularized sum, according
to the NESS condition at the zeroth order, equals to the finite $-P_{\nu\to n}^{(0,\,2)}$,
which again can be seen by counting the powers of $\tilde{\bm{G}}_{H}(\omega)$
and $\tilde{\bm{\mathfrak{D}}}(\omega)$ at large $\omega$. By similar
analysis, one can readily find all the first-order corrections of
the heat fluxes in the KG and $\beta$-FPUT model are all finite.

\section{\label{sec:Conclusion}Discussion and Conclusion}


Having established the NESS in the late time limit in the last section,
we now show how our formal perturbative results presented in Sec.~\ref{sec:zeroth-order}
and \ref{sec:The-first-order-corrections} naturally give a measure
of the strength of nonlinearity for nonlinear open quantum systems
with KG and $\beta$-FPUT type nonlinearity. This should provide insights
for understanding and controlling the mesoscopic quantum heat transport.
We denote the steady-state energy flows across the two oscillators
at the zeroth order, the first-order for KG and $\beta$-FPUT types
of nonlinearity as $P^{(0)}$, $P_{\text{KG}}^{(1)}$ and $P_{\beta}^{(1)}$
respectively. Then the ratios $P_{\text{KG}}^{(1)}/P^{(0)}$ and $P_{\text{\ensuremath{\beta}}}^{(1)}/P^{(0)}$,
which characterize the strength of nonlinearity, can provide a measurement
of how valid a perturbative calculations is, namely, as long as these
ratios are much smaller than one.

As shown in Fig. \ref{fig:Energy-current}, the ratios $P_{\text{KG}}^{(1)}/P^{(0)}$
and $P_{\text{\ensuremath{\beta}}}^{(1)}/P^{(0)}$ become smaller
with the increase of the quadratic intra-oscillator coupling constant
or the decrease of the quartic coupling constants $\lambda_{\text{KG}}$
and $\lambda_{4}$. However, comparing Fig. \ref{fig:Energy-current}
(a) with (c) or (b) with (d), we find two distinct features between
the KG and $\beta$-FPUT nonlinearities: First, the KG nonlinearity
is weaker than the $\beta$-FPUT nonlinearities in that under the
same value of parameters $P_{\text{KG}}^{(1)}/P^{(0)}$ smaller than
$P_{\text{\ensuremath{\beta}}}^{(1)}/P^{(0)}$. Secondly, the KG nonlinearity
tends to strengthen the zeroth order heat flow since it induces a
positive first-order correction while for the $\beta$-FPUT tends
to weaken the zeroth order heat flow since it induces a negative first-order
correction. As shown in Figs. \ref{fig:Energy-current}(c, d), for
the gray shaded region, where $|P_{\beta}^{(1)}/P^{(0)}|\ge1$, the
net energy current up to the first-order becomes spuriously negative,
implying heat flow from the cold bath to hot bath spontaneously. However,
we should emphasize that this regime actually belongs to the one of
strong nonlinearity, which lies beyond the reach of our perturbative
calculation here. Nonetheless we see the general tendency in how nonlinear
coupling may affect the energy transport along the chain. Hence our
perturbative result, although not applicable in the regime of strong
nonlinearity, may still provide some useful physical insights into
this regime.

Fig. \ref{fig:ratio-tuning} illustrates how the nonlinear ratio $P^{(1)}/P^{(0)}$
changes when the coupling between the two oscillators and their respective
private baths and the temperature bias are tuned. Fig. \ref{fig:ratio-tuning}(a)
and (c) indicate that for both types of nonlinearities, the stronger
the coupling is, the less the transport deviates from the linear one.
Similarly, by increasing the temperature bias between the two baths,
one may also be able to bring the transport from strong nonlinear
regime to weak nonlinear regime, as shown in Fig. \ref{fig:ratio-tuning}(b)
and (d).

On the contrary, as discussed in \citep{Paz}, the absorption refrigerator
consisting of a network of linear oscillators and three-terminal baths
does not produce any heat flow from the cold to the hot bath through
the working bath when reaching NESS. Therefore, in order to reach
efficient cooling, one must sufficiently tune the device away from
the linear transport regime. The formal expression in Sec. \ref{sec:zeroth-order}
and \ref{sec:The-first-order-corrections} still applies, as long
as the oscillator network is weakly nonlinear. One only needs to replace
the causal propagator corresponding to two-oscillator two-bath configuration
discussed here with the one corresponding to the three-oscillator
three-bath configuration. Then our perturbative calculations here,
although it cannot provide quantitatitive predictions about the transport
in the strongly nonlinear regime, can still indicate the trend in
how an interacting system enter the strongly nonlinear regime in order
to generate significant cooling effects.

In summary, we have provided a functional perturbative approach based
on the stochastic generating functional which allows one to compute
perturbatively the energy flux in the late time limit for a quantum
anharmonic chain in the presence of \textit{weak} nonlinearity. We
considered the $\alpha$-FPUT, KG, and $\beta$-FPUT types of nonlinearities
and gave the first order corrections to all the relevant energy fluxes
to the first-order of nonlinear coupling constants. Up to this order,
the contributions to the energy fluxes due to the initial state of
the chain vanish in the late time limit. We found the first order
corrections due to $\alpha$-FPUT nonlinearity vanish. For two coupled
anharmonic oscillators that are coupled to their own private harmonic
bath, the KG and $\beta$-FPUT nonlinearities lead to an NESS in the
late time limit, up to the first-order of the nonlinear coupling constant.
We conjectured that NESS exists in a chain with any of the $\alpha$-FPUT,
KG, and $\beta$-FPUT types of nonlinearities, to arbitrary orders
of the respective nonlinear the coupling constant. 
From the fractional change of the energy current due to nonlinearity,
its dependence on the oscillator-bath coupling and bath temperature
bias can provide information about the range of validity of the perturbative
analysis. This also reveals the tendency by which the energy currents
can be modified by the nonlinearities. The expressions for the energy
currents we found can be straightforwardly extended to an anharmonic
chain consisting of arbitrary number of oscillators in contact with
their own private baths. Therefore, our perturbative calculations
can help understand the anomalous heat transport in low dimensions,
at least provide some benchmark for the numerical simulations and
other analytical approximation schemes in the weakly nonlinear regime.
Our predictions about the $\alpha$-FPUT, KG, and $\beta$-FPUT types
of nonlinearities can be directly implemented and verified by e.g.
engineering an array of Josephson junctions in the limit of large
Josephson energy.

\section{Acknowledgement}

Work by JY and ANJ was supported by the U.S. Department of Energy
(DOE), Office of Science, Basic Energy Sciences (BES) under Award
No. DE-SC-0017890. JY would like to thank Professors A. Das and S.
G. Rajeev for helpful discussions. JTH and BLH are thankful to Prof.
Hong Zhao for his insight in nonlinear transport problems based on
his extensive numerical work.

\appendix

\renewcommand*{\appendixname}{Appendix }

\renewcommand*\cftsecnumwidth{4.5em}

\addtocontents{toc}{\protect\addtolength{\protect\cftsecnumwidth}{4.5em}}

\section{Feynman-Vernon influence functional formalism\label{sec:FV-formalism}}

With the help of the Feynman-Vernon influence functional formalism,
the propagating function $\mathcal{J}(\bm{\chi}_{f},\,\bm{\chi}_{f}^{\prime},\,t_{f};\bm{\chi}_{i},\,\bm{\chi}_{i}^{\prime},\,t_{i})$
can be formally expressed as 
\begin{align}
J(\bm{\chi}_{f},\,\bm{\chi}_{f}^{\prime},\,t_{f};\bm{\chi}_{i},\,\bm{\chi}_{i}^{\prime},\,t_{i})=\int_{\bm{\chi}(t_{i})=\bm{\chi}_{i}}^{\bm{\chi}(t_{f})=\bm{\chi}_{f}}\!\mathcal{D}\bm{\chi}\int_{\bm{\chi}^{\prime}(t_{i})=\bm{\chi}_{i}^{\prime}}^{\bm{\chi}^{\prime}(t_{f})=\bm{\chi}_{f}^{\prime}}\!\mathcal{D}\bm{\chi}^{\prime}\;\exp\Bigl\{ i\,S_{\chi}[\bm{\chi}]-i\,S_{\chi}[\bm{\chi}^{\prime}]\Bigr\} F[\bm{\chi},\,\bm{\chi}^{\prime}]\,,\label{eq:J}
\end{align}
where the influence functional $F[\bm{\chi},\,\bm{\chi}^{\prime}]$
in Eq. \eqref{eq:J} is 
\begin{align}
F[\bm{\chi},\bm{\chi}^{\prime}] & =\mathcal{F}[\bm{q},\bm{r}]=\exp\left\{ \int_{t_{i}}^{t_{f}}\int_{t_{i}}^{t_{f}}dsds^{\prime}\left[i\,\bm{q}^{T}(s)\bm{G}_{R}(s,\,s^{\prime})\bm{r}(s^{\prime})-\frac{1}{2}\bm{q}^{T}(s)\bm{G}_{H}(s,\,s^{\prime})\bm{q}(s^{\prime})\right]\right\} .\label{eq:F-gen-qr}
\end{align}
where we have introduced the center-of-mass coordinate $\bm{r}$ and
relative coordinate $\bm{q}$ 
\begin{align}
\bm{q}(s) & =\bm{\chi}(s)-\bm{\chi}^{\prime}(s)\,, & \bm{r}(s) & =\frac{1}{2}[\bm{\chi}(s)+\bm{\chi}^{\prime}(s)]\,,\label{eq:r}
\end{align}
and the $N\times N$ diagonal Green's function matrices 
\begin{align}
\bm{G}_{R}(s,\,s^{\prime}) & =\operatorname{diag}\bigl[G_{R}^{(n)}(s,\,s^{\prime})\bigr]\,, & \bm{G}_{H}(s,\,s^{\prime}) & =\operatorname{diag}\bigl[G_{H}^{(n)}(s,\,s^{\prime})\bigr]\,.\label{eq:G-matrix}
\end{align}
The Green's functions $G_{R}^{(n)}(s,\,s^{\prime})$ and $G_{H}^{(n)}(s,\,s^{\prime})$
are respectively the dissipation kernel and the noise kernel of the
private thermal bath at the temperature $\beta_{n}^{-1}$ associated
with the $n^{\text{th}}$ oscillator 
\begin{align}
G_{R}^{(n)}(s,\,s^{\prime}) & =-\frac{e_{n}^{2}}{2\pi}\theta(s-s^{\prime})\partial_{s}\delta(s-s^{\prime})\,, & G_{H}^{(n)}(s,\,s^{\prime}) & =e_{n}^{2}\int_{-\infty}^{\infty}\!d\omega\;\tilde{G}_{H}^{(n)}(\omega)\,e^{i\omega(s-s^{\prime})}\,,\label{eq:GHn-omega-thermal}\\
\tilde{G}_{H}^{(n)}(\omega) & =\omega\coth\left(\frac{\beta_{n}\omega}{2}\right)\,,\label{E:kbgsdbf}
\end{align}
and they are defined by 
\begin{align}
G_{R}^{(n)}(s,\,s^{\prime}) & \equiv i\,\theta(s-s^{\prime})\,\langle[\hat{\phi}^{(i)}(\bm{z}_{i}(s),\,s),\hat{\phi}^{(i)}(\bm{z}_{i}(s^{\prime}),\,s^{\prime})]\,,\rangle\label{eq:GRn-def}\\
G_{H}^{(n)}(s,\,s^{\prime}) & \equiv\langle\{\hat{\phi}^{(i)}(\bm{z}_{i}(s),\,s),\,\hat{\phi}^{(i)}(\bm{z}_{i}(s^{\prime}),\,s^{\prime})\}\rangle\,,\label{eq:GHn-def}
\end{align}
where $\hat{\phi}_{n}(\bm{z}_{n}(s),\,s)$ denotes the field operator
in the Heisenberg picture and the expectation value is taken with
respect to its initial thermal state.

Now it is useful to introduce the coarse-grained action 
\begin{align}
\mathcal{S}_{\textsc{cg}}[\bm{r},\,\bm{q}] & =\int_{t_{i}}^{t_{f}}\!ds\;\left\{ \mathcal{S}_{\chi}[\bm{r}+\frac{\bm{q}}{2}]-\mathcal{S}_{\chi}[\bm{r}-\frac{\bm{q}}{2}]+\int_{t_{i}}^{t_{f}}ds^{\prime}\;\Bigl[\bm{q}^{T}(s)\bm{G}_{R}(s,\,s^{\prime})\bm{r}(s^{\prime})+\frac{i}{2}\bm{q}^{T}(s)\bm{G}_{H}(s,\,s^{\prime})\bm{q}(s^{\prime})\Bigr]\right\} \,,\label{eq:SCG}
\end{align}
and apply the Hubbard-Stratonovich transformation~\citep{HHAoP}
to the term involving the noise kernel $\bm{G}_{H}(s,\,s^{\prime})$
in Eq. (\ref{eq:SCG}), i.e., 
\begin{equation}
\exp\left[-\frac{1}{2}\int_{t_{i}}^{t_{f}}\!ds\int_{t_{i}}^{t_{f}}\!ds^{\prime}\;\bm{q}^{T}(s)\bm{G}_{H}(s,\,s^{\prime})\bm{q}(s^{\prime})\right]=\int\!\mathcal{D}\bm{\xi}\;P[\bm{\xi}]\exp\left[i\int_{0}^{t_{f}}\!ds\;\bm{q}^{T}(s)\bm{\xi}(s)\right],\label{eq:HS}
\end{equation}
where $P[\bm{\xi}]$ is interpreted as the Gaussian probability distribution
functional of the $c$-number stochastic noise $\bm{\xi}(s)$, whose
first two moments are $\langle\xi_{k}(s)\rangle=0\,,\langle\xi_{k}(s)\xi_{l}(s^{\prime})\rangle=G_{H}^{(k)}(s,\,s^{\prime})\delta_{kl}.$
This allows to introduce the stochastic propagating function 
\begin{equation}
\mathcal{J}_{\bm{\xi}}(\bm{r}_{f},\,\bm{q}_{f},\,t_{f};\,\bm{r}_{i},\,\bm{q}_{i},\,t_{i})=\int_{\bm{q}(t_{i})=\bm{q}_{i}}^{\bm{q}(t_{f})=\bm{q}_{f}}\!\mathcal{D}\bm{q}\!\int_{\bm{r}(t_{i})=\bm{r}_{i}}^{\bm{r}(t_{f})=\bm{r}_{f}}\!\mathcal{D}\bm{r}\;\exp\Bigl\{ i\,\mathcal{S}_{\bm{\xi}}[\bm{r},\,\bm{q}]\Bigr\}\,,\label{eq:Jxi}
\end{equation}
where the stochastic effective action is defined as 
\begin{align}
\mathcal{S}_{\bm{\xi}}[\bm{r},\,\bm{q}] & =\int_{t_{i}}^{t_{f}}ds\left\{ \mathcal{S}\left[\bm{r}+\frac{\bm{q}}{2}\right]-\mathcal{S}\left[\bm{r}-\frac{\bm{q}}{2}\right]+\int_{t_{i}}^{t_{f}}ds^{\prime}\bm{q}^{T}(s)\bm{G}_{R}(s,\,s^{\prime})\bm{r}(s^{\prime})+\bm{q}^{T}(s)\bm{\xi}(s)\right\} .\label{eq:SSE}
\end{align}

\section{Functional method\label{sec:SPF}}

To evaluate the expectation value $O(t_{f},\bm{\xi}]$ of the quantity
$O$ for each realization of the stochastic noise, it is convenient
to introduce the functional method. We first note that $\operatorname{Tr}\{\hat{\rho}(t_{f})\}=\langle\!\langle\operatorname{Tr}\{\hat{\rho}_{\bm{\xi}}(t_{f})\}\rangle\!\rangle$,
and we can write 
\begin{align}
\operatorname{Tr}\Bigl\{\hat{\rho}_{\bm{\xi}}(t_{f})\Bigr\} & =\int\!d\bm{\chi}_{i}d\bm{\chi}_{i}^{\prime}\int\!d\bm{\chi}_{f}\;J_{\bm{\xi}}(\bm{\chi}_{f},\,\bm{\chi}_{f},\,t_{f};\,\bm{\chi}_{i},\,\bm{\chi}_{i}^{\prime})\rho(\bm{\chi}_{i},\,\bm{\chi}_{i}^{\prime})\nonumber \\
 & =\int\!d\bm{r}_{i}d\bm{q}{}_{i}\int\!d\bm{r}_{f}\;\mathcal{J}_{\bm{\xi}}(\bm{r}_{f},\,0,\,t_{f};\,\bm{r}_{i},\,\bm{q}_{i},\,t_{i})\varrho(\bm{r}_{i},\,\bm{q}_{i})\,,
\end{align}
where we have made the change of variables 
\begin{align}
\bm{\chi}_{f} & \mapsto\bm{r}_{f}\,, & \bm{\chi}_{i} & \mapsto\bm{r}_{i}+\frac{\bm{q}_{i}}{2}\,, & \bm{\chi}_{i}^{\prime} & \mapsto\bm{r}_{i}-\frac{\bm{q}_{i}}{2}\,,
\end{align}
and introduced the functional variants of the stochastic propagating
function $J_{\bm{\xi}}(\bm{\chi}_{f},\,\bm{\chi}_{f},\,t_{f};\,\bm{\chi}_{i},\,\bm{\chi}_{i}^{\prime})$
and the density matrix $\rho(\bm{\chi}_{i},\,\bm{\chi}_{i}^{\prime})$
by 
\begin{align}
J_{\bm{\xi}}(\bm{\chi}_{f},\,\bm{\chi}_{f},\,t_{f};\,\bm{\chi}_{i},\,\bm{\chi}_{i}^{\prime}) & =\mathcal{J}_{\bm{\xi}}(\bm{r}_{f},\,0,\,t_{f};\,\bm{r}_{i},\,\bm{q}_{i})\,, & \rho(\bm{\chi}_{i},\,\bm{\chi}_{i}^{\prime},) & =\varrho(\bm{r}_{i},\,\bm{q}_{i})\,,
\end{align}
respectively. Thus, the expectation value $p_{n}^{2}(t_{f},\,\bm{\xi}]$
in \eqref{eq:pn2} becomes 
\begin{align}
\operatorname{Tr}\Bigl\{\hat{p}_{n}^{2}\,\hat{\rho}_{\bm{\xi}}(t_{f})\Bigr\} & =\int\!d\bm{p}_{f}d\bm{\chi}_{f}d\bm{\chi}_{f}^{\prime}\;\braket{\bm{\chi}_{f}^{\prime}\big|\hat{p}_{n}^{2}\big|\bm{p}_{f}}\braket{\bm{p}_{f}\big|\bm{\chi}_{f}}\braket{\bm{\chi}_{f}\big|\hat{\rho}_{\bm{\xi}}(t_{f})\big|\bm{\chi}_{f}^{\prime}}\nonumber \\
 & =\frac{1}{(2\pi)^{N}}\int\!d\bm{p}_{f}\;p_{fn}^{2}\,e^{i\bm{p}_{f}\cdot(\bm{\chi}_{f}^{\prime}-\bm{\chi}_{f})}\int\!d\bm{\chi}_{f}d\bm{\chi}_{f}^{\prime}\;\braket{\bm{\chi}_{f}\big|\hat{\rho}_{\bm{\xi}}(t_{f})\big|\bm{\chi}_{f}^{\prime}}\nonumber \\
 & =\frac{1}{(2\pi)^{N}i^{2}}\frac{\partial^{2}}{\partial h_{n}^{2}}\int\!d\bm{p}_{f}\;e^{i\bm{p}_{f}\cdot(\bm{\chi}_{f}^{\prime}-\bm{\chi}_{f})+i\bm{p}_{f}\cdot\bm{h}}\int\!d\bm{\chi}_{f}d\bm{\chi}_{f}^{\prime}\braket{\bm{\chi}_{f}\big|\hat{\rho}(t_{f})\big|\bm{\chi}_{f}^{\prime}}\,\bigg|_{\bm{h}=0}\nonumber \\
 & =-\frac{\partial^{2}}{\partial h_{n}^{2}}\int\!d\bm{\chi}_{f}\;\braket{\bm{\chi}_{f}\big|\hat{\rho}(t_{f})\big|\bm{\chi}_{f}-\bm{h}}\,\bigg|_{\bm{h}=0}\nonumber \\
 & =-\frac{\partial^{2}}{\partial h_{n}^{2}}\int\!d\bm{r}_{i}d\bm{q}_{i}\int\!d\bm{r}_{f}\;\mathcal{J}_{\bm{\xi}}(\bm{r}_{f},\,\bm{h},\,t_{f};\,\bm{r}_{i},\,\bm{q}_{i})\varrho(\bm{r}_{i},\,\bm{q}_{i})\,\bigg|_{\bm{h}=0}\,.\label{eq:pn2-J}
\end{align}
Similarly, we find the expectation value $p_{n}(t_{f},\,\bm{\xi}]$
in \eqref{eq:pn2} expressed by 
\begin{align}
\operatorname{Tr}\Bigl\{\hat{p}_{n}\,\hat{\rho}_{\bm{\xi}}(t_{f})\Bigr\} & =-i\,\frac{\partial}{\partial h_{n}}\int\!d\bm{r}_{i}d\bm{q}_{i}\int\!d\bm{r}_{f}\;\mathcal{J}_{\bm{\xi}}(\bm{r}_{f},\,\bm{h},\,t_{f};\,\bm{r}_{i},\,\bm{q}_{i})\varrho(\bm{r}_{i},\,\bm{q}_{i})\,\bigg|_{\bm{h}=0}\,.\label{eq:pn-J}
\end{align}
as a derivative over a fictitious, external constant source $\bm{h}$.

It is straightforward to calculate $w_{\nu\to n}^{(\eta)}(t_{f},\,\bm{\xi}]$
in \eqref{eq:wn-nu-eta} by the same derivative approach. We need
to calculate the following quantity 
\begin{align}
\operatorname{Tr}\Bigl\{\hat{p}_{n}^{2}\bigl(\hat{\chi}_{n}-\hat{\chi}_{m}\bigr)^{\eta}\,\hat{\rho}_{\bm{\xi}}(t_{f})\Bigr\} & =-\frac{\partial^{2}}{\partial h_{n}^{2}}\int\!d\bm{\chi}_{f}\;\bigl(\chi_{fn}-\chi_{fm}\bigr)^{\eta}\,\braket{\bm{\chi}_{f}\big|\hat{\rho}_{\bm{\xi}}(t_{f})\big|\bm{\chi}_{f}-\bm{h}}\,\bigg|_{\bm{h}=0}\nonumber \\
 & =-\frac{\partial^{2}}{\partial h_{n}^{2}}\int\!d\bm{r}_{i}d\bm{q}_{i}\int\!d\bm{r}_{f}\,\Bigl[r_{fn}+\frac{h_{n}}{2}-r_{fm}-\frac{h_{m}}{2}\Bigr]^{\eta}\mathcal{J}_{\bm{\xi}}(\bm{r}_{f},\,\bm{h},\,t_{f};\,\bm{r}_{i},\,\bm{q}_{i})\varrho(\bm{r}_{i},\,\bm{q}_{i})\,\bigg|_{\bm{h}=0}\,,\label{eq:pn-2-un-um-eta}
\end{align}
and the quantity 
\begin{align}
\operatorname{Tr}\Bigl\{\bigl(\hat{\chi}_{n}-\hat{\chi}_{m}\bigr)^{\eta}\hat{p}_{n}^{2}\,\hat{\rho}_{\bm{\xi}}(t_{f})\Bigr\} & =-\frac{\partial^{2}}{\partial h_{n}^{2}}\int\!d\bm{\chi}_{f}\;\bigl(\chi_{n}-\chi_{m}\bigr)^{\eta}\braket{\bm{\chi}_{f}+\bm{h}\big|\hat{\rho}_{\bm{\xi}}(t_{f})\big|\bm{\chi}_{f}}\,\bigg|_{\bm{h}=0}\nonumber \\
 & =-\frac{\partial^{2}}{\partial h_{n}^{2}}\int\!d\bm{r}_{i}d\bm{q}_{i}\int\!d\bm{r}_{f}\;\Bigl[r_{fn}-\frac{h_{n}}{2}-r_{fm}+\frac{h_{m}}{2}\Bigr]^{\eta}\mathcal{J}_{\bm{\xi}}(\bm{r}_{f},\,\bm{h},\,t_{f};\,\bm{r}_{i},\,\bm{q}_{i})\varrho(\bm{r}_{i},\,\bm{q}_{i})\,\bigg|_{\bm{h}=0}\,.\label{eq:un-um-eta-pn-2}
\end{align}
Note the different signs in the right hand sides in Eq. (\ref{eq:pn-2-un-um-eta},
\ref{eq:un-um-eta-pn-2}) due to the differential positions of $\hat{p}_{n}$
on the left hand sides. Since the derivatives in Eqs.~\eqref{eq:pn-2-un-um-eta},
\eqref{eq:un-um-eta-pn-2} do not involve $h_{m}$ ($m\neq n$), so
we can safely set $h_{m}=0$ in these equation. Furthermore, the terms
proportional to the third power of $h_{n}$ and higher can be discarded
since only the second-order derivatives of $h_{n}$ are involved.
With these observations, we immediately find 
\begin{align}
\operatorname{Tr}\Bigl\{\bigl[\hat{p}_{n}^{2},\,\bigl(\hat{\chi}_{n}-\hat{\chi}_{m}\bigr)^{2}\bigr]\,\hat{\rho}_{\bm{\xi}}(t_{f})\Bigr\} & =-2\,\frac{\partial^{2}}{\partial h_{n}^{2}}\int\!d\bm{r}_{i}d\bm{q}_{i}\int\!d\bm{r}_{f}\;h_{n}\Bigl[r_{fn}-r_{fm}\Bigr]\,\mathcal{J}_{\bm{\xi}}(\bm{r}_{f},\,\bm{h},\,t_{f};\,\bm{r}_{i},\,\bm{q}_{i})\varrho(\bm{r}_{i},\,\bm{q}_{i})\,\bigg|_{\bm{h}=0}\nonumber \\
 & =-2\,\frac{\partial^{2}}{\partial h_{n}^{2}}\biggl\{ h_{n}\Bigl[\langle r_{fn}\rangle_{\bm{h}}-\langle r_{fm}\rangle_{\bm{h}}\Bigr]\biggr\}\,\bigg|_{\bm{h}=0}\nonumber \\
 & =-4\,\frac{\partial}{\partial h_{n}}\langle r_{fn}-r_{fm}\rangle_{\bm{h}}\,\bigg|_{\bm{h}=0}\,,\label{eq:Tr-p2-delu2}
\end{align}
where the expectation value $\langle\bullet\rangle$ is defined by
\begin{equation}
\langle\bullet\rangle\equiv\operatorname{Tr}\Bigl\{\bullet\;\hat{\rho}(t_{f})\Bigr\}=\int\!d\bm{r}_{i}d\bm{q}_{i}\int\!d\bm{r}_{f}\;\bigl(\bullet\bigr)\,\mathcal{J}_{\bm{\xi}}(\bm{r}_{f},\,\bm{h},\,t_{f};\,\bm{r}_{i},\,\bm{q}_{i})\varrho(\bm{r}_{i},\,\bm{q}_{i})\,,
\end{equation}
and the subscript $\bm{h}$ is a reminder that this expectation value
is still a function of $\bm{h}$. Similarly, We can find 
\begin{align}
\operatorname{Tr}\Bigl\{\bigl[\hat{p}_{n}^{2},\,\bigl(\hat{\chi}_{n}-\hat{\chi}_{m}\bigr)^{3}\bigr]\,\hat{\rho}_{\bm{\xi}}(t_{f})\Bigr\} & =-3\,\frac{\partial^{2}}{\partial h_{n}^{2}}\int\!d\bm{r}_{i}d\bm{q}_{i}\int\!d\bm{r}_{f}\;h_{n}\Bigl[r_{fn}-r_{fm}\Bigr]^{2}\,\mathcal{J}_{\bm{\xi}}(\bm{r}_{f},\,\bm{h},\,t_{f};\,\bm{r}_{i},\,\bm{q}_{i})\varrho(\bm{r}_{i},\,\bm{q}_{i})\,\bigg|_{\bm{h}=0}\nonumber \\
 & =-6\,\frac{\partial}{\partial h_{n}}\langle\bigl(r_{fn}-r_{fm}\bigr)^{2}\rangle_{\bm{h}}\,\bigg|_{\bm{h}=0}\,,
\end{align}
and 
\begin{equation}
\operatorname{Tr}\Bigl\{\bigl[\hat{p}_{n}^{2},\,\bigl(\hat{\chi}_{n}-\hat{\chi}_{m}\bigr)^{4}\bigr]\,\hat{\rho}_{\bm{\xi}}(t_{f})\Bigr\}=-8\,\frac{\partial}{\partial h_{n}}\langle\bigl(r_{fn}-r_{fm}\bigr)^{3}\rangle_{\bm{h}}\,\bigg|_{\bm{h}=0}\,.\label{eq:Tr-p2-delu4}
\end{equation}
Note that in Eqs.~(\ref{eq:pn2-J}-\ref{eq:Tr-p2-delu4}), oscillator
$n$ and $m$ are not necessarily ajacent to each other and interactions
beyond the nearest neigbhor interactions are also included. If we
take $m=\nu$ with $\nu=n\pm1$ we arrive at a simple, but general
expression for $w_{\nu\to n}^{(\eta)}(t_{f},\,\bm{\xi}]$ 
\begin{equation}
w_{\nu\to n}^{(\eta)}(t_{f},\,\bm{\xi}]=i\,\frac{\partial}{\partial h_{n}}\langle\bigl(r_{fn}-r_{f\nu}\bigr)^{\eta-1}\rangle_{\bm{h}}\,\bigg|_{\bm{h}=0}\,,
\end{equation}

In calculating Eqs.~(\ref{eq:pn2-J}-\ref{eq:Tr-p2-delu4}), we note
that it turns out convenient to introduce the generating functional
\begin{equation}
\mathcal{Z}_{\bm{\xi}}[\bm{j},\,\bm{j}^{\prime},\bm{h})=\int\!d\bm{r}_{i}d\bm{q}_{i}\int\!d\bm{r}_{f}\int_{\bm{r}(0)=\bm{r}_{i}}^{\bm{r}(t_{f})=\bm{r}_{f}}\!\mathcal{D}\bm{r}\int_{\bm{q}(0)=\bm{q}_{i}}^{\bm{q}(t_{f})=\bm{h}}\!\mathcal{D}\bm{q}\;e^{i\mathcal{S}_{\bm{\xi}}[\bm{r},\,\bm{q};\,\bm{j},\,\bm{j}^{\prime}]}\varrho(\bm{r}_{i},\,\bm{q}_{i})\,,\label{E:mggnsm}
\end{equation}
with a new action $\mathcal{S}_{\bm{\xi}}[\bm{r},\,\bm{q};\,\bm{j},\,\bm{j}^{\prime}]$
defined by 
\begin{equation}
\mathcal{S}_{\bm{\xi}}[\bm{r},\,\bm{q};\,\bm{j},\,\bm{j}^{\prime}]=\mathcal{S}_{\bm{\xi}}[\bm{r},\,\bm{q}]+\int_{0}^{t_{f}}\!ds\;\Bigl\{\bm{j}(s)\cdot\bm{r}(s)+\bm{j}^{\prime}(s)\cdot\bm{q}(s)\Bigr\}\,.\label{E:fghdjve}
\end{equation}
Apparently it is the stochastic effective action $\mathcal{S}_{\bm{\xi}}[\bm{r},\,\bm{q}]$
in \eqref{eq:SSE} attached with external sources $\bm{j}$ and $\bm{j}^{\prime}$.
With the help of this generating functional, the results in Eqs.~\eqref{eq:pn-Z}
and \eqref{eq:pn2-Z} then take very succinct forms.

To find the energy exchange between the oscillators, we need to evaluate
$w_{\nu\to n}^{(\eta)}(t_{f},\,\bm{\xi}]$ in \eqref{eq:wn-nu-eta}.
Generically it is equivalent to the evaluations of the multi-moment
of $\bm{\chi}$. We can do it also by the functional method. For example,
the second moment $\operatorname{Tr}\{\hat{\chi}_{m}\hat{\chi}_{n}\,\hat{\rho}_{\bm{\xi}}(t_{f})\}$
can be found by 
\begin{align}
\operatorname{Tr}\Bigl\{\hat{\chi}_{n}\hat{\chi}_{m}\,\hat{\rho}_{\bm{\xi}}(t_{f})\Bigr\} & =\int\!d\bm{\chi}_{f}\;\chi_{fn}\chi_{fm}\,\rho_{\bm{\xi}}(\bm{\chi}_{f},\,\bm{\chi}_{f},\,t_{f})=\langle r_{n}(t_{f})r_{m}(t_{f})\rangle_{\bm{h}=0}\,,\label{eq:2pointV}
\end{align}
and, similarly, three and four-point moments are given by 
\begin{align}
\operatorname{Tr}\Bigl\{\hat{\chi}_{n}\hat{\chi}_{m}\hat{\chi}_{k}\,\hat{\rho}_{\bm{\xi}}(t_{f})\Bigr\} & =\langle r_{n}(t_{f})r_{m}(t_{f})r_{k}(t_{f})\rangle_{\bm{h}=0}\,,\label{eq:3pointV}\\
\operatorname{Tr}\Bigl\{\hat{\chi}_{n}\hat{\chi}_{m}\hat{\chi}_{k}\hat{\chi}_{l}\,\hat{\rho}_{\bm{\xi}}(t_{f})\Bigr\} & =\langle r_{n}(t_{f})r_{m}(t_{f})r_{k}(t_{f})r_{l}(t_{f})\rangle_{\bm{h}=0}\,,\label{eq:4pointV}
\end{align}
with 
\begin{align}
\langle r_{n}(t_{f})r_{m}(t_{f})\rangle_{\bm{h}=0} & =-\frac{\delta^{2}\mathcal{Z}_{\bm{\xi}}[\bm{j},\,\bm{j}^{\prime},\bm{h})}{\delta j_{n}(t_{f})\delta j_{m}(t_{f})}\,\bigg|_{\bm{j}=\bm{j}^{\prime}=\bm{h}=0}\,,\\
\langle r_{n}(t_{f})r_{m}(t_{f})r_{k}(t_{f})r_{l}(t_{f})\rangle_{\bm{h}=0} & =+\frac{\delta^{4}\mathcal{Z}_{\bm{\xi}}[\bm{j},\,\bm{j}^{\prime},\bm{h})}{\delta j_{n}(t_{f})\delta j_{m}(t_{f})\delta j_{k}(t_{f})\delta j_{l}(t_{f})}\,\bigg|_{\bm{j}=\bm{j}^{\prime}=\bm{h}=0}\,.
\end{align}
Here we see that the generating functional defined in (\ref{eq:Z})
can be used to efficiently to evaluate the quantities such as $\operatorname{Tr}\{\hat{p}_{n_{1}}\cdots\hat{p}_{n_{a}}\hat{\chi}_{m_{1}}\cdots\hat{\chi}_{m_{b}}\,\hat{\rho}_{\bm{\xi}}(t_{f})\}$
or $\operatorname{Tr}\{\hat{\chi}_{m_{1}}\cdots\hat{\chi}_{m_{b}}\,\hat{\rho}_{\bm{\xi}}(t_{f})\,\hat{p}_{n_{1}}\cdots\hat{p}_{n_{a}}\}$
for any two arbitrary nonnegative integers $a$ and $b$.

\section{Explicit evaluation of the noninteracting stochastic propagating
function\label{sec:Z0-SE}}

Here we will derive the explicit form of the stochastic propagating
function for the linearly coupled oscillators. This will be the zeroth-order
term for the more general configuration considered in the paper, and
will serve as the basis for the perturbative calculations of the higher-order
corrections in various nonlinear coupling discussed earlier.

Taking variations of $\mathcal{S}_{\bm{\xi}}^{(0)}$ in \eqref{eq:S0SE}
with respect to $\bm{q}$ and $\bm{r}$ respectively, we obtains a
sets of differential equations for the classical trajectories $\bar{\bm{q}}$
and $\bar{\bm{r}}$, 
\begin{align}
\ddot{\bm{\bar{r}}}(s)-\int_{0}^{t_{f}}\!ds^{\prime}\;\bm{G}_{R}(s,\,s^{\prime})\cdot\bar{\bm{r}}(s^{\prime})+\bm{\Omega}^{2}\bar{\bm{r}}(s) & =\bm{j}^{\prime}(s)+\bm{\xi}(s)\,,\label{eq:r-diff}\\
\ddot{\bar{\bm{q}}}(s)-\int_{0}^{t_{f}}\!ds^{\prime}\;\bm{G}_{R}(s^{\prime},\,s)\cdot\bar{\bm{q}}(s^{\prime})+\bm{\Omega}^{2}\bar{\bm{q}}(s) & =\bm{j}(s)\,,\label{eq:q-diff}
\end{align}
with boundary conditions $\bar{\bm{r}}(t_{f})=\bm{r}_{f}$, $\bar{\bm{r}}(0)=\bm{r}_{i}$,
$\bar{\bm{q}}(t_{f})=\bm{q}_{f}$, $\bar{\bm{q}}(0)=\bm{q}_{i}$,
and $t_{i}=0$.

In fact the integrals in Eqs.~\eqref{eq:r-diff} and \eqref{eq:q-diff}
can be computed explicitly using Eq.~\eqref{eq:GHn-omega-thermal}.
We find 
\begin{align}
\int_{0}^{t_{f}}\!ds^{\prime}\;\bm{G}_{R}(s,\,s^{\prime})\cdot\bar{\bm{r}}(s^{\prime}) & =-2\gamma\left[\delta(0)\bar{\bm{r}}(s)+\dot{\bar{\bm{r}}}(s)\right]\,,\\
\int_{0}^{t_{f}}\!ds^{\prime}\;\bm{G}_{R}(s^{\prime},\,s)\cdot\bar{\bm{q}}(s^{\prime}) & =2\gamma\left[\delta(0)\bar{\bm{q}}(s)+\dot{\bar{\bm{q}}}(s)\right]\,,
\end{align}
with $\gamma=e^{2}/8\pi$. Eqs.~\eqref{eq:r-diff} and \eqref{eq:q-diff}
then become 
\begin{align}
\ddot{\bar{\bm{r}}}(s)+2\gamma\dot{\bar{\bm{r}}}(s)+\bm{\Omega}_{R}^{2}\cdot\bar{\bm{r}}(s) & =\bm{j}^{\prime}(s)+\bm{\xi}(s)\,,\label{eq:r-diff-gamma}\\
\ddot{\bar{\bm{q}}}(s)-2\gamma\dot{\bar{\bm{q}}}(s)+\bm{\Omega}_{R}^{2}\cdot\bar{\bm{q}}(s) & =\bm{j}(s)\,.\label{eq:q-diff-gamma}
\end{align}
where the divergent expression $\delta(0)$ is absorbed by $\omega_{0}$
in $\bm{\Omega}$ to give the renormalized value $\omega_{R}$, and
thus the renormalized frequency matrix $\bm{\Omega}_{R}^{2}$ is defined
as 
\begin{equation}
\bm{\Omega}_{R}^{2}\equiv\begin{bmatrix}\omega_{R}^{2}+\lambda_{2} & -\lambda_{2}\\
-\lambda_{2} & \omega_{R}^{2}+\lambda_{2}
\end{bmatrix}.\label{eq:Omega-Renormalized}
\end{equation}
Following the procedures outlined in~\citep{HHAoP,NENL1}, we find
the classical stochastic action $\bar{S}_{\bm{\xi}}^{(0)}[\bar{\bm{q}},\,\bar{\bm{r}},\,\bm{j},\,\bm{j}^{\prime}]$
given by 
\begin{align}
\bar{\mathcal{S}}_{\bm{\xi}}^{(0)}[\bm{j},\,\bm{j}^{\prime};\,\bm{r}_{f},\,\bm{q}_{f},\,t_{f};\,\bm{r}_{i},\,\bm{q}_{i}) & =\bm{q}_{f}^{T}\cdot\Bigl\{\dot{\bm{\alpha}}(t_{f})\cdot\bm{r}_{i}+\dot{\bm{\beta}}(t_{f})\cdot\bm{r}_{f}\Bigr\}-\bm{q}_{i}^{T}\cdot\Bigl\{\dot{\bm{\alpha}}(0)\cdot\bm{r}_{i}+\dot{\bm{\beta}}(0)\cdot\bm{r}_{f}\Bigr\}\nonumber \\
 & \qquad\qquad\qquad+\bm{q}_{f}^{T}\cdot\Bigl\{\dot{\bm{F}}(t_{f},\bm{j}^{\prime}]+\dot{\bm{F}}(t_{f},\bm{\xi}]\Bigr\}-\bm{q}_{i}^{T}\cdot\Bigl\{\dot{\bm{F}}(0,\,\bm{j}^{\prime}]+\dot{\bm{F}}(0,\,\bm{\xi}]\Bigr\}\label{eq:S0-xi-all}\\
 & +\int_{0}^{t_{f}}\!ds\;\bm{j}^{T}(s)\cdot\bm{\alpha}(s)\cdot\bm{r}_{i}+\int_{0}^{t_{f}}\!ds\;\bm{j}^{T}(s)\cdot\bm{\beta}(s)\cdot\bm{r}_{f}+\int_{0}^{t_{f}}\!ds\;\bm{j}^{T}(s)\cdot\Bigl\{\bm{F}(s,\,\bm{j}^{\prime}]+\bm{F}(s,\,\bm{\xi}]\Bigr\}\,.\nonumber 
\end{align}
Here we have introduced a couple of shorthand notations 
\begin{align}
\bm{\alpha}(t) & =\bm{D}_{1}(t)-\bm{D}_{2}(t)\cdot\bm{D}_{2}^{-1}(t_{f})\cdot\bm{D}_{1}(t_{f})\,,\qquad\qquad\qquad\bm{\beta}(t)=\bm{D}_{2}(t)\cdot\bm{D}_{2}^{-1}(t_{f})\,,\label{eq:beta}\\
\bm{F}(t,\,\bm{f}] & =\int_{0}^{t}\!ds\;\bm{D}_{2}(t-s)\cdot\bm{f}(s)-\bm{D}_{2}(t)\cdot\bm{D}_{2}^{-1}(t_{f})\cdot\int_{0}^{t_{f}}\!ds\;\bm{D}_{2}(t_{f}-s)\cdot\bm{f}(s)\,,\label{eq:Jprime}
\end{align}
among which $\bm{D}_{1}(t)$ and $\bm{D}_{2}(t)$ two fundamental
solutions to \eqref{eq:r-diff-gamma}, satisfying the initial conditions
$\bm{D}_{1}(0)=\mathbb{I}$, $\dot{\bm{D}}_{1}(0)=0$, $\bm{D}_{2}(0)=0$
and $\dot{\bm{D}}_{2}(0)=\mathbb{I}$.

Using the Laplace transform, we can find the explicit forms of $\bm{D}_{1}$
and $\bm{D}_{2}$ 
\begin{align}
\bm{D}_{1}(\mathfrak{s}) & =(\mathfrak{s}+2\gamma)[\mathfrak{s}^{2}+2\gamma\mathfrak{s}+\bm{\Omega}_{R}^{2}]{}^{-1}=\frac{\bm{W}_{1}(\mathfrak{s})}{\prod_{k=1}^{4}(\mathfrak{s}-\mathfrak{s}_{k})}\label{eq:D1-s-def}\\
\bm{D}_{2}(\mathfrak{s}) & =[\mathfrak{s}^{2}+2\gamma\mathfrak{s}+\bm{\Omega}_{R}^{2}]{}^{-1}=\frac{\bm{W}_{2}(\mathfrak{s})}{\prod_{k=1}^{4}(\mathfrak{s}-\mathfrak{s}_{k})}\label{eq:D2-s-def}
\end{align}
where $\bm{W}_{1}(\mathfrak{s})=(\mathfrak{s}+2\gamma)\bm{W}_{2}(\mathfrak{s})$,
\begin{equation}
\bm{W}_{2}(\mathfrak{s})=\begin{bmatrix}\mathfrak{s}^{2}+2\gamma\mathfrak{s}+\omega_{0}^{2}+\lambda_{2}, & \lambda_{2}\\
\lambda_{2}, & \mathfrak{s}^{2}+2\gamma\mathfrak{s}+\omega_{0}^{2}+\lambda_{2}
\end{bmatrix}\,,
\end{equation}
with $\mathfrak{s}_{1,\,2}=-\gamma\pm\sqrt{\gamma^{2}-\omega_{0}^{2}}$
and $\mathfrak{s}_{3,\,4}=-\gamma\pm\sqrt{\gamma^{2}-(\omega_{0}^{2}+2\lambda_{2})}$.
Since $\operatorname{Re}\mathfrak{s}_{k}<0$ for all $k=1$, $\cdots$,
4, all the matrix elements of $\bm{D}_{i}(t)$ exponentially decay
on time scale much large than $\gamma_{0}^{-1}$, where 
\begin{equation}
\gamma_{0}\equiv\min_{k}\lvert\operatorname{Re}\mathfrak{s}_{k}\rvert\,,\label{eq:gamma0}
\end{equation}
i.e., $\bm{D}_{1,2}(t)\sim e^{-\gamma_{0}t}$, which leads to 
\begin{equation}
\lim_{t\to\infty}\frac{d^{n}}{dt^{n}}\bm{D}_{i}(t)=0\,,\label{eq:Di-derivative}
\end{equation}
where $n=0,\,1,\,2\cdots$. This is an extremely useful property in
the context of the late-time dynamics of the oscillators. Essentially,
it predicts that the zeroth-order dynamics of the oscillators will
relax to a steady state, independent of the initial conditions.

Since from \eqref{eq:Jxi}, the integral gives 
\begin{align}
 & \quad\int\!d\bm{r}_{f}\;\exp\left\{ i\left[\bm{q}_{f}^{T}\cdot\dot{\bm{\beta}}(t_{f})-\bm{q}_{i}^{T}\cdot\dot{\bm{\beta}}(0)+\int_{0}^{t_{f}}\!ds\;\bm{j}^{T}(s)\cdot\bm{\beta}(s)\right]\cdot\bm{r}_{f}\right\} \nonumber \\
 & =\mathcal{N}^{-1}\delta^{(N)}(\,\bm{q}_{i}^{T}-\bm{q}_{f}^{T}\cdot\dot{\bm{D}}_{2}(t_{f})-\int_{0}^{t_{f}}\!ds\;\bm{j}^{T}(s)\cdot\bm{D}_{2}(s)\,)\,,\label{eq:int-drf}
\end{align}
with the normalization $\mathcal{N}^{-1}=(2\pi)^{N}\det\bm{D}_{2}(t_{f})$,
we arrive at 
\begin{align}
 & \quad\int\!d\bm{r}_{f}\;\mathcal{J}_{0\bm{\xi}}[\bm{j},\,\bm{j}^{\prime};\,\bm{r}_{f},\,\bm{h},\,t_{f};\,\bm{r}_{i},\,\bm{q}_{i})\nonumber \\
 & =\delta^{(N)}(\,\bm{q}_{i}^{T}-\bm{q}_{f}^{T}\cdot\dot{\bm{D}}_{2}(t_{f})-\int_{0}^{t_{f}}\!ds\;\bm{j}^{T}(s)\cdot\bm{D}_{2}(s)\,)\times\exp\Bigl\{ i\,\bar{\mathcal{S}}_{\bm{\xi}}^{(0)}[\bm{j},\,\bm{j}^{\prime};\,\bm{0},\,\bm{h},\,t_{f};\,\bm{r}_{i},\,\bm{q}_{i})\Bigr\}\,,
\end{align}
from Eq.~\eqref{eq:int-drf}. Thus, the stochastic generating function
$\mathcal{Z}_{\bm{\xi}}^{(0)}$ at late times reduces to 
\begin{equation}
\mathcal{Z}_{\bm{\xi}}^{(0)}[\bm{j},\,\bm{j}^{\prime},\,\bm{h})=\int\!d\bm{r}_{i}\;\exp\Bigl\{ i\,\bar{\mathcal{S}}_{\bm{\xi}}^{(0)}[\bm{j},\,\bm{j}^{\prime};\,\bm{0},\,\bm{h},\,t_{f};\,\bm{r}_{i},\,\bm{0})\Bigr\}\,\varrho(\bm{r}_{i},\,\int_{0}^{t_{f}}\!ds\;\bm{D}_{2}(s)\bm{j}(s))\,,\label{eq:Z0-xi-app}
\end{equation}
where we have removed terms that will vanish in the limit $t_{f}\to0$.
The stochastic action in this limit will take a simpler form 
\begin{align}
\bar{\mathcal{S}}_{\bm{\xi}}^{(0)}[\bm{j},\,\bm{j}^{\prime};\,\bm{0},\,\bm{h},\,t_{f};\,\bm{r}_{i},\,\bm{0}) & =\int_{0}^{t_{f}}\!ds\;\bm{j}^{T}(s)\cdot\bm{D}_{1}(s)\cdot\bm{r}_{i}+\bm{h}^{T}\cdot\int_{0}^{t_{f}}\!ds\;\dot{\bm{D}}_{2}(t_{f}-s)\cdot\bm{\xi}(s)+\bm{h}^{T}\cdot\int_{0}^{t_{f}}\!ds\;\dot{\bm{D}}_{2}(t_{f}-s)\cdot\bm{j}^{\prime}(s)\nonumber \\
 & \qquad+\int_{0}^{t_{f}}\!ds\int_{0}^{s}\!ds^{\prime}\;\bm{j}^{T}(s)\cdot\bm{D}_{2}(s-s^{\prime})\cdot\Bigl[\bm{j}^{\prime}(s^{\prime})+\bm{\xi}(s^{\prime})\Bigr]\,.\label{eq:S0-xi-app}
\end{align}

\section{Theorems on the functional derivatives\label{sec:Functional derivatives} }

In this section, we have gathered a few handy theorems about the functional
derivatives used in the calculations of the first-order corrections
to the energy current in the anharmonic chain.

\begin{thm}\label{thm:NKlessM} For $\nu+\kappa<\mu$ 
\begin{equation}
\frac{\partial^{\kappa}\delta^{\nu+\mu}\mathcal{Z}_{\bm{\xi}}^{(0)}[\bm{j},\,\bm{j}^{\prime},\,\bm{h})}{\partial h_{k_{1}}\partial h_{k_{2}}\cdots\partial h_{k_{\kappa}}\delta j_{n_{1}}(s_{1})\delta j_{n_{2}}(s_{2})\cdots\delta j_{n_{\nu}}(s_{\nu})\delta j'_{m_{1}}(s_{1}^{\prime})\delta j'_{m_{2}}(s_{2}^{\prime})\cdots\delta j'_{m_{\mu}}(s_{\mu}^{\prime})}\,\bigg|_{\bm{j}=\bm{j}^{\prime}=\bm{h}=0}=0\,,\label{eq:NlessM}
\end{equation}
where $\kappa$, $\nu$ and $\mu$ are non-negative integers. \end{thm}

\begin{proof} The Taylor's expanding of the term $\exp\{i\bar{\mathfrak{S}}^{(0)}[\bm{j},\,\bm{j}^{\prime},\,\bm{h})\}$
gives defined Eq. (\ref{eq:S0-frak}), 
\begin{equation}
\exp\Bigl\{ i\,\bar{\mathfrak{S}}^{(0)}[\bm{j},\,\bm{j}^{\prime},\,\bm{h})\Bigr\}=\sum_{n=0}^{\infty}\frac{i^{n}}{n!}\left[\int_{0}^{t_{f}}\!ds\int_{0}^{t_{f}}\!ds^{\prime}\;\bm{j}^{T}(s)\cdot\bm{\mathfrak{D}}(s-s^{\prime})\cdot\bm{j}^{\prime}(s^{\prime})+\bm{h}^{T}\cdot\int_{0}^{t_{f}}\!ds\;\dot{\bm{D}}_{2}(t_{f}-s)\cdot\bm{j}^{\prime}(s)\right]^{n}\,.\label{eq:eiSfrak}
\end{equation}
By counting the powers of $\bm{j}^{\prime}$, we see that after we
take the functional derivatives $\delta^{\mu}/\delta j_{m_{1}}^{\prime}(s_{1}^{\prime})\delta j_{m_{2}}^{\prime}(s_{2}^{\prime})\cdots\delta j_{m_{\mu}}^{\prime}(s_{\mu}^{\prime})$
and set $\bm{j}^{\prime}=0$, the only nonvanishing contributions
in the Taylor series is 
\begin{equation}
\frac{i^{\mu}}{\mu!}\left[\int_{0}^{t_{f}}\!ds\int_{0}^{t_{f}}\!ds^{\prime}\;\bm{j}^{T}(s)\cdot\bm{\mathfrak{D}}(s-s^{\prime})\cdot\bm{j}^{\prime}(s^{\prime})+\bm{h}^{T}\cdot\int_{0}^{t_{f}}\!ds\;\dot{\bm{D}}_{2}(t_{f}-s)\cdot\bm{j}^{\prime}(s)\right]^{\mu}\,.\label{eq:M-jprime}
\end{equation}
Then we immediate see that the total powers of $\bm{j}$ and $\bm{h}$
in the expansion of Eq.~\eqref{eq:M-jprime} must be $\mu$, too.
If we take the subsequent functional derivatives $\partial^{\kappa}\delta^{\nu}/\delta j_{n_{1}}(s_{1})\delta j_{n_{2}}(s_{2})\cdots\delta j_{n_{\nu}}(s_{\nu})$$\partial h_{k_{1}}\partial h_{k_{2}}\cdots\partial h_{k_{\kappa}}$
of \eqref{eq:M-jprime}, then we are left with term that are at least
$(\mu-\nu-\kappa)^{\text{th}}$ power of $\bm{j}$ or $\bm{h}$. Setting
$\bm{j}=\bm{j}'=\bm{h}=0$ gives vanishing results. \end{proof}

\begin{thm} For the case of $\nu+\kappa=\mu$, 
\begin{align}
 & \quad\frac{\partial^{\kappa}\delta^{\nu+\mu}\mathcal{Z}_{\bm{\xi}}^{(0)}[\bm{j},\,\bm{j}^{\prime},\,\bm{h})}{\partial h_{k_{1}}\partial h_{k_{2}}\cdots\partial h_{k_{\kappa}}\delta j_{n_{1}}(s_{1})\delta j_{n_{2}}(s_{2})\cdots\delta j_{n_{\nu}}(s_{\nu})\delta j_{m_{1}}^{\prime}(s_{1}^{\prime})\delta j_{m_{2}}^{\prime}(s_{2}^{\prime})\cdots\delta j_{m_{\mu}}^{\prime}(s_{\mu}^{\prime})}\,\bigg|_{\bm{j}=\bm{j}^{\prime}=\bm{h}=0}\nonumber \\
 & =\frac{i^{\mu}}{\mu!}\,C_{\kappa}^{\mu}\,\frac{\delta^{\nu+\mu}\left[{\displaystyle \int_{0}^{t_{f}}\!ds\int_{0}^{t_{f}}\!ds^{\prime}\;\bm{j}^{T}(s)\cdot\bm{\mathfrak{D}}(s-s^{\prime})\cdot\bm{j}^{\prime}(s^{\prime})}\right]^{\mu-\kappa}\left[{\displaystyle \bm{h}^{T}\cdot\int_{0}^{t_{f}}\!ds\;\dot{\bm{D}}_{2}(t_{f}-s)\cdot\bm{j}^{\prime}(s)}\right]^{\kappa}}{\delta j_{n_{1}}(s_{1})\delta j_{n_{2}}(s_{2})\cdots\delta j_{n_{\nu}}(s_{\nu})\delta j_{m_{1}}^{\prime}(s_{1}^{\prime})\delta j_{m_{2}}^{\prime}(s_{2}^{\prime})\cdots\delta j_{m_{\mu}}^{\prime}(s_{\mu}^{\prime})}\,\bigg|_{\bm{j}=\bm{j}^{\prime}=\bm{h}=0}\,,
\end{align}
which may be further evaluated analytically using the Wick's theorem.
\end{thm}

\begin{proof} The proof is similar to the proof for Theorem \ref{thm:NKlessM}.
After we take the functional derivatives $\delta^{\mu}/\delta j_{m_{1}}^{\prime}(s_{1}^{\prime})\cdots\delta j_{m_{\mu}}^{\prime}(s_{\mu}^{\prime})$
of the Taylor's expansion of Eq.~\eqref{eq:eiSfrak} and set $\bm{j}^{\prime}=0$,
we find the surviving terms in the expansion are 
\begin{equation}
\frac{i^{\mu}}{\mu!}\left[\int_{0}^{t_{f}}\!ds\int_{0}^{t_{f}}\!ds^{\prime}\;\bm{j}^{T}(s)\cdot\bm{\mathfrak{D}}(s-s^{\prime})\cdot\bm{j}^{\prime}(s^{\prime})+\bm{h}^{T}\cdot\int_{0}^{t_{f}}\!ds\;\dot{\bm{D}}_{2}(t_{f}-s)\cdot\bm{j}^{\prime}(s)\right]^{\mu}\,.\label{E:fbkrts}
\end{equation}
Applying the additonal derivatives $\partial^{\kappa}\delta^{\nu}/\partial h_{k_{1}}\partial h_{k_{2}}\cdots\partial h_{k_{\kappa}}\delta j_{n_{1}}(s_{1})\delta j_{n_{2}}(s_{2})\cdots\delta j_{n_{\mu}}(s_{\nu})$
on \eqref{E:fbkrts} gives 
\begin{equation}
\frac{i^{\mu}}{\mu!}\,C_{\kappa}^{\mu}\left[\int_{0}^{t_{f}}\!ds\int_{0}^{t_{f}}\!ds^{\prime}\;\bm{j}^{T}(s)\cdot\bm{\mathfrak{D}}(s-s^{\prime})\cdot\bm{j}^{\prime}(s^{\prime})\right]^{\mu-\kappa}\left[\bm{h}^{T}\cdot\int_{0}^{t_{f}}\!ds\;\dot{\bm{D}}_{2}(t_{f}-s)\cdot\bm{j}^{\prime}(s)\right]^{\kappa}\,.
\end{equation}
All the other terms in the binomial expansion are either of the order
$\mathcal{O}(\bm{j}^{\nu+1})$ or of the order $\mathcal{O}(\bm{h}^{\kappa+1})$.
They will vanish after either performing the functional derivative
$\delta^{\nu}/\delta j_{n_{1}}(s_{1})\delta j_{n_{2}}(s_{2})\cdots\delta j_{n_{\nu}}(s_{\nu})$
and setting $\bm{j}=0$, or taking the derivative $\partial^{\kappa}/\partial h_{k_{1}}\partial h_{k_{2}}\cdots\partial h_{k_{\kappa}}$
and setting $\bm{h}=0$. \end{proof}

When the times in the functional derivative are the same, we have
the following theorem \begin{thm}\label{thm:same-s} For $\mu\ge1$
and $\kappa<\mu$ 
\begin{equation}
\frac{\partial^{\kappa}\delta^{\nu+\mu}\mathcal{Z}_{\bm{\xi}}^{(0)}[\bm{j},\,\bm{j}^{\prime},\,\bm{h})}{\partial h_{k_{1}}\partial h_{k_{2}}\cdots\partial h_{k_{\kappa}}\delta j_{n_{1}}(s)\delta j_{n_{2}}(s)\cdots\delta j_{n_{\nu}}(s)\delta j_{m_{1}}^{\prime}(s)\delta j_{m_{2}}^{\prime}(s)\cdots\delta j_{m_{\mu}}^{\prime}(s)}\,\bigg|_{\bm{j}=\bm{j}^{\prime}=\bm{h}=0}=0\,.\label{eq:dZdj-same-s}
\end{equation}
\end{thm} \begin{proof} Again by counting the powers of $\bm{h}$
and $\bm{j}$ in the binomial expansion Eq. (\ref{eq:M-jprime}),
we note that the non-vanishing terms will likely be 
\begin{equation}
\sum_{m=0}^{\mu}\frac{i^{\mu}}{\mu!}\,C_{m}^{\mu}\left[\int_{0}^{t_{f}}\!ds\int_{0}^{t_{f}}\!ds^{\prime}\;\bm{j}^{T}(s)\bm{\mathfrak{D}}(s-s^{\prime})\cdot\bm{j}^{\prime}(s^{\prime})\right]^{\mu-m}\left[\bm{h}^{T}\cdot\int_{0}^{t_{f}}\!ds\;\dot{\bm{D}}_{2}(t_{f}-s)\cdot\bm{j}^{\prime}(s)\right]^{m}\,,\label{eq:bi-Mth}
\end{equation}
since the remaining terms in the expansion of \eqref{eq:M-jprime}
are of the order $\mathcal{(}\bm{h}^{\kappa+1})$, and they will vanish
after taking the derivatives $\partial^{\kappa}/\partial h_{k_{1}}\partial h_{k_{2}}\cdots\partial h_{k_{\kappa}}$
and setting $\bm{h}=0$. Next we observe that $j_{n_{i}}$ and $j_{m_{i}}$
have the same time argument $s$. Thus when we perform the functional
derivatives $\delta^{\nu}/\delta j_{n_{1}}(s)\delta j_{n_{2}}(s)\cdots\delta j_{n_{\nu}}(s)$
of Eq.~\eqref{eq:bi-Mth}, we will obtain terms proportional to $\bm{\mathfrak{D}}(0)$,
which by definition is zero. \end{proof}

Similarly, we have the following theorem

\begin{thm}\label{thm:same-s-tf} For $\mu\ge1$ and $\kappa+\lambda<\mu$,
we have 
\begin{equation}
\frac{\partial^{\kappa}\delta^{\lambda+\nu+\mu}\mathcal{Z}_{\bm{\xi}}^{(0)}[\bm{j},\,\bm{j}^{\prime},\,\bm{h})}{\partial h_{k_{1}}\partial h_{k_{2}}\cdots\partial h_{k_{\kappa}}\delta j_{l_{1}}(t_{f})\cdots\delta j_{l_{\lambda}}(t_{f})\delta j_{n_{1}}(s)\delta j_{n_{2}}(s)\cdots\delta j_{n_{\nu}}(s)\delta j_{m_{1}}^{\prime}(s)\delta j_{m_{2}}^{\prime}(s)\cdots\delta j_{m_{\mu}}^{\prime}(s)}\,\bigg|_{\bm{j}=\bm{j}^{\prime}=\bm{h}=0}=0\,.\label{eq:dZdj-same-s-tf}
\end{equation}
\end{thm} \begin{proof} The proof is similar to Theorem \ref{thm:same-s}.
When we apply $\kappa$ $\bm{j}^{\prime}$-functional derivatives
to the second pair of square brackets in \eqref{eq:bi-Mth} and $(\mu-\kappa)$
$\bm{j}^{\prime}$-functional derivatives to the first pair of square
brackets in \eqref{eq:bi-Mth}, if we would like to obtain a nonvanishing
result in the limits $\bm{j}=\bm{j}'=\bm{h}=0$, the only possibility
is to apply $(\mu-\kappa)$ $\bm{j}$-functional derivatives to the
first pair of square brackets in \eqref{eq:bi-Mth} as well.

This means that we have to choose $(\mu-\kappa)$ $\bm{j}$-functional
derivatives among $\delta/\delta j_{l_{1}}(t_{f})$, $\cdots$, $\delta/\delta j_{l_{\lambda}}(t_{f})$,
and $\delta/\delta j_{n_{1}}(s)$, $\cdots$, $\delta/\delta j_{n_{\nu}}(s)$.
Since $\mu-\kappa>\lambda$, this procedure unavoidably gives at a
term proportional to $\bm{\mathfrak{D}}(0)$, which by definition
is zero. \end{proof}

When all the functional derivatives are taken at the same time $t_{f}$,
the following lemma guarantees that the functional derivatives in
the late time limit $t_{f}\to\infty$ is independent of the initial
state

\begin{thm}[independence of initial state]\label{thm:dNZscrdNj}
Given $\mathscr{Z}^{(0)}$ in \eqref{eq:Z0-scr} and the property
of $\bm{D}_{i}$ in \eqref{eq:Di-derivative}, we can shown 
\begin{equation}
\lim_{t_{f}\to\infty}\frac{\delta^{\nu}\mathscr{Z}^{(0)}[\varrho_{i},\,\bm{j}]}{\delta j_{n_{1}}(t_{f})\delta j_{n_{2}}(t_{f})\cdots\delta j_{n_{\nu}}(t_{f})}=0\,.\label{eq:dNZscrdNj}
\end{equation}
\end{thm} \begin{proof} We observe that for $\nu=1$, the proof
is trivial 
\begin{align}
\frac{\delta\mathscr{Z}^{(0)}[\varrho_{i},\,\bm{j}]}{\delta j_{n}(t_{f})} & =i\int\!d\bm{r}_{i}\;\Bigl[\cancel{\bm{D}_{1}(t_{f})}\cdot\bm{r}_{i}\Bigr]_{n}\times\exp\Bigl\{ i\int_{0}^{t_{f}}\!ds\;\bm{j}^{T}(s)\cdot\bm{D}_{1}(s)\cdot\bm{r}_{i}\Bigr\}\,\varrho(\,\bm{r}_{i},\,J(t_{f})\,)\nonumber \\
 & \qquad\qquad+\int\!d\bm{r}_{i}\;\exp\Bigl\{ i\int_{0}^{t_{f}}\!ds\;\bm{j}^{T}(s)\cdot\bm{D}_{1}(s)\cdot\bm{r}_{i}\Bigr\}\frac{\delta}{\delta J(t_{f})}\varrho(\bm{r}_{i},\,J(t_{f})\,)\cdot\cancel{\bm{D}_{2}(t_{f})}\nonumber \\
 & =0\,,\label{eq:dZscrdj}
\end{align}
where 
\begin{equation}
J(t_{f})=\int_{0}^{t_{f}}\!ds\;\bm{D}_{2}(s)\cdot\bm{j}(s)\,.
\end{equation}
The vanishing result is the consequence of Eq. (\ref{eq:Di-derivative}).
Generalizing this observation straightforwardly to any integer $\nu$
concludes the proof \end{proof}

Another useful lemma is the following \begin{lem}\label{lem:decay-Laplace}
For signal $f(s)$ that\textcolor{blue}{{} }has no singularity on\textcolor{blue}{{}
$[0,\,+\infty)$} and decays as $s\to\infty$, i.e., ${\displaystyle \lim_{s\to\infty}\lvert f(s)\rvert=0}$,
we then have 
\begin{equation}
\lim_{t_{f}\to\infty}\int_{0}^{t_{f}}\!ds\;D(t_{f}-s)f(s)=0\,,\label{eq:int-dDidn}
\end{equation}
where $D(s)$ stands for any matrix element of the derivative or power
of the fundamental solution $\bm{D}_{i}(s)$ discussed in \ref{sec:Z0-SE}.
\end{lem}

\begin{proof} From the definition of the Laplace transform, we can
easily see that the possible poles of $f(\mathfrak{s})$ must lie
on the left half of the complex $\mathfrak{s}$ plane. Take $D(s)$
be the matrix element of $\bm{D}_{i}(s)$, and then by the Laplace
convolution theorem, we obtain 
\begin{align}
\lim_{t_{f}\to\infty}\int_{0}^{t_{f}}\!ds\;\Bigl[\bm{D}_{i}(t_{f}-s)\Bigr]_{jk}f(s) & =\lim_{t_{f}\to\infty}\frac{1}{2\pi i}\int_{\operatorname{Re}\mathfrak{s}_{0}-i\infty}^{\operatorname{Re}\mathfrak{s}_{0}+i\infty}\!d\mathfrak{s}\;\Bigl[\bm{D}_{i}(\mathfrak{s})\Bigr]_{jk}f(\mathfrak{s})\,e^{\mathfrak{s}t_{f}}\nonumber \\
 & =\lim_{t_{f}\to\infty}\frac{1}{2\pi i}\int_{\operatorname{Re}\mathfrak{s}_{0}-i\infty}^{\operatorname{Re}\mathfrak{s}_{0}+i\infty}\!d\mathfrak{s}\;\frac{[\bm{W}_{i}(\mathfrak{s})]_{jk}{f}(\mathfrak{s})e^{\mathfrak{s}t_{f}}}{\prod_{k=1}^{4}(\mathfrak{s}-\mathfrak{s}_{k})}\,,\label{eq:int-Di}
\end{align}
where $\operatorname{Re}\mathfrak{s}_{0}$ is a real number so that
the integration is well defined. Since all of the poles $\mathfrak{s}_{m}$
of the integrand lie on the left half of the complex $\mathfrak{s}$
plane, i.e., $\operatorname{Re}\mathfrak{s}_{m}<0$, the factor $e^{t_{f}\operatorname{Re}\mathfrak{s}_{m}}$
will approach zero exponentially in the limit $t_{f}\to\infty$. Thus
from the residue theorem, we conclude that \eqref{eq:int-Di} will
vanish in the same limit. Since any powers or times derivatives of
$\bm{D}_{i}(\mathfrak{s})$ have the same singular structure as $\bm{D}_{i}(\mathfrak{s})$,
the conclusion naturally extends to the case that $D(s)$ is powers
of time derivatives of any matrix element of $\bm{D}_{i}(s)$. \end{proof}

When we apply Theorem \ref{lem:decay-Laplace} in what follows, the
function $f$ may contain the Hadamard Green's function defined by
Eq. \eqref{eq:GHn-omega-thermal}, which is divergent without regularization.
This will render $f(s)$ no longer analytic on the $s\in[0,\,\infty)$
or at $s=\infty$. However, realistically there is always some regularization
imposed by physical constraint on the Hadamard Green's function, such
as the cutoff regularization, which will make it convergent. Therefore,
the Hadamard Green's function in what follows should be understood
as the regularized one and then we could apply Theorem \ref{lem:decay-Laplace}
without any further concerns. 

\begin{thm}\label{thm:decaying}

\begin{align}
 & \quad\lim_{t_{f}\to\infty}\int_{0}^{t_{f}}\!ds\;\langle\!\langle\,\Bigl[\dot{\bm{D}}_{2}(t_{f}-s)\Bigr]_{nr}\,\xi_{n}(t_{f})\,\frac{\delta^{3}\mathscr{Z}^{(0)}[\varrho_{i},\,\bm{j}]\exp\Bigl\{ i\,\bar{\mathscr{S}}_{\bm{\xi}}^{(0)}[\bm{j},\,\bm{h})\Bigr\}}{\delta j_{k}(s)\delta j_{l}(s)\delta j_{m}(s)}\,\rangle\!\rangle\,\bigg|_{\bm{j}=\bm{h}=0}\nonumber \\
 & =\lim_{t_{f}\to\infty}\int_{0}^{t_{f}}\!ds\;\langle\!\langle\,\Bigl[\dot{\bm{D}}_{2}(t_{f}-s)\Bigr]_{nr}\,\xi_{n}(t_{f})\,\frac{\delta^{3}\exp\Bigl\{ i\,\bar{\mathscr{S}}_{\bm{\xi}}^{(0)}[\bm{j},\,\bm{h})\Bigr\}}{\delta j_{k}(s)\delta j_{l}(s)\delta j_{m}(s)}\,\rangle\!\rangle\,\bigg|_{\bm{j}=\bm{h}=0}\,.\label{E:fhbkrt}
\end{align}
\end{thm} \begin{proof} In order to generate a non-vanishing results
after averaging on all the realization of $\bm{\xi}$, we keep terms
that are even powers in $\bm{\xi}$. Thus either one or all of the
functional derivatives among $\delta/\delta j_{k}(s)$, $\delta/\delta j_{l}(s)$,
$\delta/\delta j_{m}(s)$ should be applied to $\exp\bigl\{ i\,\bar{\mathscr{S}}_{\bm{\xi}}^{(0)}[\bm{j},\,\bm{h})\bigr\}$.
Now we show that the former case will give vanishing values.

Suppose we apply $\delta/\delta j_{k}(s)$ to $\exp\bigl\{ i\,\bar{\mathscr{S}}_{\bm{\xi}}^{(0)}[\bm{j},\,\bm{h})\bigr\}$,
and $\delta/\delta j_{l}(s)$, $\delta/\delta j_{m}(s)$ to $\mathscr{Z}^{(0)}[\varrho_{i},\,\bm{j}]$.
Then after taking the ensemble average, we obtain expressions like
\begin{align}
 & \quad\int_{0}^{t_{f}}\!ds\int_{0}^{t_{f}}\!ds^{\prime}\;\Bigl[\bm{D}_{2}(t_{f}-s)\Bigr]_{nr}\Bigl[\bm{\mathfrak{D}}(s-s^{\prime})\Bigr]_{kp}\langle\!\langle\xi_{n}(t_{f})\xi_{p}(s^{\prime})\rangle\!\rangle\Bigl[\bm{D}_{i}(s)\Bigr]_{l\alpha}\Bigl[\bm{D}_{j}(s)\Bigr]_{m\beta}\nonumber \\
\sim\epsilon_{1}(t_{f}) & \equiv\int_{0}^{t_{f}}\!ds\int_{0}^{t_{f}}\!ds^{\prime}\;D(t_{f}-s)\mathfrak{D}(s-s^{\prime})G_{H}(t_{f}-s^{\prime})D^{2}(s)\,,\label{eq:estimate}
\end{align}
where $[\bm{D}_{i}(s)]_{l\alpha}$ and $[\bm{D}_{j}(s)]_{m\beta}$
result from the application of $\delta/\delta j_{l}(s)$, $\delta/\delta j_{m}(s)$
on $\mathscr{Z}^{(0)}[\varrho_{i},\,\bm{j}]$. Now we make the change
of variables $t_{f}-s\to s$, $t_{f}-s^{\prime}\to s^{\prime}$, and
then Eq.~(\ref{eq:estimate}) becomes 
\begin{align}
\epsilon_{1}(t_{f}) & =\int_{0}^{t_{f}}\!ds\;D^{2}(t_{f}-s)\left[\int_{s}^{t_{f}}\!ds^{\prime}\;D(s)D(s^{\prime}-s)G_{H}(s^{\prime})\right]\,.
\end{align}
Since both $D(s)$ and $G_{H}(s)$ vanish at $s=\infty$, we know
${\displaystyle \lim_{t_{f}\to\infty}\int_{s}^{t_{f}}ds^{\prime}D(s)D(s^{\prime}-s)G_{H}(s^{\prime})}$
exists and we can safely set $t_{f}=\infty$, 
\begin{equation}
\lim_{t_{f}\to\infty}\int_{s}^{t_{f}}\!ds^{\prime}\;D(s)D(s^{\prime}-s)G_{H}(s^{\prime})=\int_{s}^{\infty}\!ds^{\prime}\;D(s)D(s^{\prime}-s)G_{H}(s^{\prime})\,.
\end{equation}
Therefore in the large $t_{f}$ limit, we find $\epsilon_{1}(t_{f})$
given by 
\begin{equation}
\lim_{t_{f}\to\infty}\epsilon_{1}(t_{f})=\lim_{t_{f}\to\infty}\int_{0}^{t_{f}}\!ds\;D^{2}(t_{f}-s)\left[\int_{s}^{\infty}ds^{\prime}D(s)D(s^{\prime}-s)G_{H}(s^{\prime})\right]\,.
\end{equation}
Since 
\begin{equation}
\lim_{s\to\infty}\int_{s}^{\infty}\!ds^{\prime}\;D(s)D(s^{\prime}-s)G_{H}(s^{\prime})=0\,,
\end{equation}
by Lemma \ref{lem:decay-Laplace}, we have 
\begin{equation}
\lim_{t_{f}\to\infty}\epsilon_{1}(t_{f})=0\,.
\end{equation}

In the second case, all the functional derivatives $\delta/\delta j_{k}(s)$,
$\delta/\delta j_{l}(s)$, $\delta/\delta j_{p}(s)$ apply to $\exp\bigl\{ i\,\bar{\mathscr{S}}_{\bm{\xi}}^{(0)}[\bm{j},\,\bm{h})\bigr\}$,
and $\mathscr{Z}^{(0)}[\varrho_{i},\,\bm{j}=0]=1$, so this gives
\eqref{E:fhbkrt}. This concludes the proof. \end{proof}

By similar procedures, because $\bm{\xi}$ has a Gaussian distribution,
we can readily show \begin{thm}\label{thm:decaying-2j} 
\begin{align}
 & \quad\lim_{t_{f}\to\infty}\int_{0}^{t_{f}}\!ds\;\langle\!\langle\,\Bigl[\dot{\bm{D}}_{2}(t_{f}-s)\Bigr]_{nr}\,\xi_{n}(t_{f})\,\frac{\delta^{2}\mathscr{Z}^{(0)}[\varrho_{i},\,\bm{j}]\exp\Bigl\{ i\,\bar{\mathscr{S}}_{\bm{\xi}}^{(0)}[\bm{j},\,\bm{h})\Bigr\}}{\delta j_{k}(s)\delta j_{l}(s)}\,\rangle\!\rangle\,\bigg|_{\bm{j}=\bm{h}=0}\nonumber \\
 & =\lim_{t_{f}\to\infty}\int_{0}^{t_{f}}\!ds\;\langle\!\langle\,\Bigl[\dot{\bm{D}}_{2}(t_{f}-s)\Bigr]_{nr}\,\xi_{n}(t_{f})\,\frac{\delta^{2}\exp\Bigl\{ i\,\bar{\mathscr{S}}_{\bm{\xi}}^{(0)}[\bm{j},\,\bm{h})\Bigr\}}{\delta j_{k}(s)\delta j_{l}(s)}\,\rangle\!\rangle\,\bigg|_{\bm{j}=\bm{h}=0}\nonumber \\
 & =0\,.
\end{align}
\end{thm}

Similar conclusions can be further generalized. We prove the following
theorems with the assistance of Lemma \ref{lem:decay-Laplace}. \begin{thm}\label{thm:decaying-2D-3j}
\begin{align}
 & \quad\lim_{t_{f}\to\infty}\int_{0}^{t_{f}}\!ds\int_{0}^{t_{f}}\!ds'\;\langle\!\langle\,\Bigl[\dot{\bm{D}}_{2}(t_{f}-s)\Bigr]_{nr}\Bigl[\dot{\bm{D}}_{2}(t_{f}-s^{\prime})\Bigr]_{np}\,\xi_{p}(s^{\prime})\,\frac{\delta^{3}\mathscr{Z}^{(0)}[\varrho_{i},\,\bm{j}]\exp\Bigl\{ i\,\bar{\mathscr{S}}_{\bm{\xi}}^{(0)}[\bm{j},\,\bm{h})\Bigr\}}{\delta j_{k}(s)\delta j_{l}(s)\delta j_{m}(s)}\,\rangle\!\rangle\,\bigg|_{\bm{j}=\bm{h}=0}\nonumber \\
 & =\lim_{t_{f}\to\infty}\int_{0}^{t_{f}}\!ds\int_{0}^{t_{f}}\!ds'\;\langle\!\langle\,\Bigl[\dot{\bm{D}}_{2}(t_{f}-s)\Bigr]_{nr}\Bigl[\dot{\bm{D}}_{2}(t_{f}-s^{\prime})\Bigr]_{np}\,\xi_{p}(s^{\prime})\,\frac{\delta^{3}\exp\Bigl\{ i\,\bar{\mathscr{S}}_{\bm{\xi}}^{(0)}[\bm{j},\,\bm{h})\Bigr\}}{\delta j_{k}(s)\delta j_{l}(s)\delta j_{m}(s)}\,\rangle\!\rangle\,\bigg|_{\bm{j}=\bm{h}=0}\,.
\end{align}
\end{thm}

\begin{proof} In order to obtain a non-vanishing results, we keep
terms that are even powers in $\bm{\xi}$. Thus, either one or all
of the three functional derivatives can be applied to $\exp\bigl\{ i\,\bar{\mathscr{S}}_{\bm{\xi}}^{(0)}[\bm{j},\,\bm{h}=0)\bigr\}$.
For the former case, let us apply, say, $\delta/\delta j_{k}(s)$
to $\exp\bigl\{ i\,\bar{\mathscr{S}}_{\bm{\xi}}^{(0)}[\bm{j},\,\bm{h}=0)\bigr\}$
and $\delta/\delta j_{l}(s)$, $\delta/\delta j_{m}(s)$ to $\mathscr{Z}^{(0)}[\varrho_{i},\,\bm{j}]$.

Then we will have expressions like 
\begin{align}
 & \quad\int_{0}^{t_{f}}\!ds\int_{0}^{t_{f}}\!ds'\int_{0}^{t_{f}}\!ds_{1}\;\Bigl[\dot{\bm{D}}_{2}(t_{f}-s)\Bigr]_{nr}\Bigl[\dot{\bm{D}}_{2}(t_{f}-s^{\prime})\Bigr]_{np}\Bigl[\bm{\mathfrak{D}}(s-s_{1})\Bigr]_{kq}\,\langle\!\langle\xi_{p}(s^{\prime})\xi_{q}(s_{1})\rangle\!\rangle\,\Bigl[\bm{D}_{i}(s)\Bigr]_{l\alpha}\Bigl[\bm{D}_{j}(s)\Bigr]_{m\beta}\nonumber \\
\sim\epsilon_{2}(t_{f}) & \equiv\int_{0}^{t_{f}}\!ds\;\dot{D}(t_{f}-s)\left[D^{2}(s)\int_{0}^{t_{f}}\!ds^{\prime}\;\int_{0}^{s}ds_{1}D(s-s_{1})G_{H}(s^{\prime}-s_{1})\right]\,.\label{eq:eps2}
\end{align}
Again since $D(s)$ and $G_{H}(s)$ vanish at $s=\infty$, so we find
\begin{equation}
\lvert\,\lim_{t_{f}\to\infty}\int_{0}^{t_{f}}\!ds^{\prime}\int_{0}^{s}\!ds_{1}\;D(s-s_{1})G_{H}(s^{\prime}-s_{1})\,\rvert=\lvert\,\int_{0}^{\infty}\!ds^{\prime}\int_{0}^{s}\!ds_{1}\;D(s-s_{1})G_{H}(s^{\prime}-s_{1})\,\rvert<\infty\,.
\end{equation}
Thus we obtain 
\begin{equation}
\lim_{t_{f}\to\infty}\epsilon_{2}(t_{f})=\lim_{t_{f}\to\infty}\int_{0}^{t_{f}}\!ds\;\dot{D}(t_{f}-s)\left[D^{2}(s)\int_{0}^{\infty}\!ds^{\prime}\int_{0}^{s}\!ds_{1}\;D(s-s_{1})G_{H}(s^{\prime}-s_{1})\right]\,.
\end{equation}
On the other hand, we know $D^{2}(s)$ decays as $s\to\infty$ and
\begin{equation}
\lim_{s\to\infty}\lvert\,\int_{0}^{\infty}\!ds'\int_{0}^{s}\!ds_{1}\;D(s-s_{1})G_{H}(s^{\prime}-s_{1})\,\rvert<\infty\,,
\end{equation}
so we conclude 
\begin{equation}
\lim_{s\to\infty}\left[D^{2}(s)\int_{0}^{t_{f}}\!ds'\int_{0}^{s}\!ds_{1}\;D(s-s_{1})G_{H}(s^{\prime}-s_{1})\right]=0\,.
\end{equation}
By Lemma \ref{lem:decay-Laplace}, we obtain 
\begin{equation}
\lim_{t_{f}\to\infty}\epsilon_{2}(t_{f})=0\,.
\end{equation}
\end{proof}

\begin{thm}\label{thm:decaying-2D-2j} 
\begin{align}
 & \quad\lim_{t_{f}\to\infty}\int_{0}^{t_{f}}\!ds\int_{0}^{t_{f}}\!ds'\;\langle\!\langle\Bigl[\dot{\bm{D}}_{2}(t_{f}-s)\Bigr]_{nr}\Bigl[\dot{\bm{D}}_{2}(t_{f}-s^{\prime})\Bigr]_{np}\,\xi_{p}(s^{\prime})\,\frac{\delta^{2}\mathscr{Z}^{(0)}[\varrho_{i},\,\bm{j}]\exp\Bigl\{ i\,\bar{\mathcal{S}}_{\bm{\xi}}^{(0)}[\bm{j},\,\bm{h})\Bigr\}}{\delta j_{k}(s)\delta j_{l}(s)}\rangle\!\rangle\,\bigg|_{\bm{j}=\bm{h}=0}\nonumber \\
 & =\lim_{t_{f}\to\infty}\int_{0}^{t_{f}}\!ds\int_{0}^{t_{f}}\!ds'\;\langle\!\langle\Bigl[\dot{\bm{D}}_{2}(t_{f}-s)\Bigr]_{nr}\Bigl[\dot{\bm{D}}_{2}(t_{f}-s^{\prime})\Bigr]_{np}\,\xi_{p}(s^{\prime})\,\frac{\delta^{2}\exp\Bigl\{ i\,\bar{\mathcal{S}}_{\bm{\xi}}^{(0)}[\bm{j},\,\bm{h})\Bigr\}}{\delta j_{k}(s)\delta j_{l}(s)}\rangle\!\rangle\,\bigg|_{\bm{j}=\bm{h}=0}\nonumber \\
 & =0
\end{align}
\end{thm}

\begin{thm}\label{thm:decaying-2D-1j} 
\begin{align}
 & \quad\lim_{t_{f}\to\infty}\int_{0}^{t_{f}}\!ds\int_{0}^{t_{f}}\!ds'\;\langle\!\langle\Bigl[\dot{\bm{D}}_{2}(t_{f}-s)\Bigr]_{nr}\Bigl[\dot{\bm{D}}_{2}(t_{f}-s^{\prime})\Bigr]_{np}\,\frac{\delta\mathscr{Z}^{(0)}[\varrho_{i},\,\bm{j}]\exp\Bigl\{ i\,\bar{\mathcal{S}}_{\bm{\xi}}^{(0)}[\bm{j},\,\bm{h})\Bigr\}}{\delta j_{k}(s)}\rangle\!\rangle\,\bigg|_{\bm{j}=\bm{h}=0}\nonumber \\
 & =\lim_{t_{f}\to\infty}\int_{0}^{t_{f}}\!ds\int_{0}^{t_{f}}\!ds'\;\langle\!\langle\Bigl[\dot{\bm{D}}_{2}(t_{f}-s)\Bigr]_{nr}\Bigl[\dot{\bm{D}}_{2}(t_{f}-s^{\prime})\Bigr]_{np}\,\frac{\delta\exp\Bigl\{ i\,\bar{\mathcal{S}}_{\bm{\xi}}^{(0)}[\bm{j},\,\bm{h})\Bigr\}}{\delta j_{k}(s)}\rangle\!\rangle\,\bigg|_{\bm{j}=\bm{h}=0}\nonumber \\
 & =0\,.
\end{align}
\end{thm}

Theorems \ref{thm:decaying}-\ref{thm:decaying-2D-1j} imply that
we can set $\mathscr{Z}^{(0)}[\varrho_{i},\,\bm{j}]$ in Eq.~\eqref{eq:Z0-xi}
to unity, which is equivalent to taking $\bm{j}=0$ in $\mathscr{Z}^{(0)}[\varrho_{i},\,\bm{j}]$,
without changing the heat transport in the late time limit. Since
the dependence of the initial state of the chain is included in $\mathscr{Z}^{(0)}[\varrho_{i},\,\bm{j}]$,
it then implies that these late-time results are independent of the
initial configurations of the chain.

\section{Useful identities for proving the NESS at the first order\label{sec:NESS-Identities}}

In this appendix, we will construct the identities useful for proving
the NESS at the first order in nonlinearity in Sec.~\ref{S:mgbfdmg}.

It will be convenient for the following discussions if we decompose
the right hand side of Eq.~\eqref{eq:Gamma-nr-klm-Frequency} into
three components and denote them by

\begin{equation}
\Gamma_{klm}^{nr}=\prescript{}{\text{I}}{\Gamma}_{klm}^{nr}+\prescript{}{\text{I\!I}}{\Gamma}_{klm}^{nr}+\prescript{}{\text{I\!I\!I}}{\Gamma}_{klm}^{nr}\,,
\end{equation}
and similar decompositions will be applied to $\tilde{\Gamma}_{klm}^{nr}$,
$\Upsilon_{\nu klm}^{nr}$, $\tilde{\Upsilon}_{\nu klm}^{nr}$ and
$\text{\ensuremath{\Lambda_{klm}^{n}}}$. as well. Furthermore, we
note that at late time the correlation matrix $\mathbf{C}$ in the
coincident limit actually takes the form 
\begin{equation}
\bm{\mathcal{C}}_{lm}=\int_{-\infty}^{\infty}\!\frac{d\omega}{2\pi}\;\Bigl[\tilde{\bm{\mathfrak{D}}}(\omega)\tilde{\bm{G}}_{H}(\omega)\tilde{\bm{\mathfrak{D}}}(-\omega)\Bigr]_{lm}\,.
\end{equation}
With the help of the properties of $\tilde{\bm{\mathfrak{D}}}(\omega)$
outlined in Eqs. (\ref{eq:Inverse1}-\ref{eq:D-omega-symmetry}),
we find 
\begin{eqnarray}
\bm{\mathcal{C}}{}_{lm} & = & \sum_{k=1}^{2}\int_{-\infty}^{\infty}\!\frac{d\omega}{2\pi}\;\Bigl[\tilde{\bm{\mathfrak{D}}}(\omega)\Bigr]_{lk}\Bigl[\tilde{\bm{\mathfrak{D}}}^{*}(\omega)\Bigr]_{km}\Bigl[\tilde{\bm{G}}_{H}(\omega)\Bigr]_{kk}\\
\prescript{}{\text{I}}{\Gamma}_{klm}^{nr} & = & -\bm{\mathcal{C}}{}_{lm}\int_{-\infty}^{\infty}\frac{d\omega}{2\pi}\omega\text{Im}\left([\tilde{\bm{\mathfrak{D}}}(\omega)]_{nr}\tilde{\bm{\mathfrak{D}}}(\omega)]_{nk}\right)\Bigl[\tilde{\bm{G}}_{H}(\omega)\Bigr]_{nn}\\
\prescript{}{\text{I}}{\tilde{\Gamma}}_{klm}^{nr} & = & 2\bm{\mathcal{C}}{}_{lm}\sum_{p=1}^{2}\int_{-\infty}^{\infty}\frac{d\omega}{2\pi}\omega^{2}\text{Re}\left([\tilde{\bm{\mathfrak{D}}}(\omega)]_{nr}[\tilde{\bm{\mathfrak{D}}}(\omega)]_{np}^{*}[\tilde{\bm{\mathfrak{D}}}(\omega)]_{pk}\right)\Bigl[\tilde{\bm{G}}_{H}(\omega)\Bigr]_{pp}\\
\prescript{}{\text{I}}{\Upsilon}_{\nu klm}^{nr} & = & -\bm{\mathcal{C}}{}_{lm}\sum_{p=1}^{2}\int_{-\infty}^{\infty}\frac{d\omega}{2\pi}\omega\text{Im}\left([\tilde{\bm{\mathfrak{D}}}(\omega)]_{nr}[\tilde{\bm{\mathfrak{D}}}(\omega)]_{\nu p}^{*}[\tilde{\bm{\mathfrak{D}}}(\omega)]_{pk}\right)\Bigl[\tilde{\bm{G}}_{H}(\omega)\Bigr]_{pp}\\
\prescript{}{\text{I}}{\tilde{\Upsilon}}_{\nu klm}^{nr} & = & \bm{\mathcal{C}}{}_{lm}\sum_{p=1}^{2}\int_{-\infty}^{\infty}\frac{d\omega}{2\pi}\omega\text{Im}\left([\tilde{\bm{\mathfrak{D}}}(\omega)]_{\nu r}[\tilde{\bm{\mathfrak{D}}}(\omega)]_{np}^{*}[\tilde{\bm{\mathfrak{D}}}(\omega)]_{pk}\right)\Bigl[\tilde{\bm{G}}_{H}(\omega)\Bigr]_{pp}\\
\prescript{}{\text{I}}{\Lambda}_{klm}^{n} & = & -\bm{\mathcal{C}}{}_{lm}\sum_{p=1}^{2}\int_{-\infty}^{\infty}\frac{d\omega}{2\pi}\omega\text{Im}\left([\tilde{\bm{\mathfrak{D}}}(\omega)]_{np}[\tilde{\bm{\mathfrak{D}}}^{*}(\omega)]_{pk}\right)\Bigl[\tilde{\bm{G}}_{H}(\omega)\Bigr]_{pp}\label{eq:Lambda-partI}
\end{eqnarray}
so we arrive at 
\[
-\prescript{}{\text{I}}{\Gamma}_{klm}^{nr}+2\gamma\prescript{}{\text{I}}{\tilde{\Gamma}}_{klm}^{nr}-\lambda_{2}\left(\prescript{}{\text{I}}{\Upsilon}_{\nu klm}^{nr}+\prescript{}{\text{I}}{\tilde{\Upsilon}}_{\nu klm}^{nr}\right)=\bm{\mathcal{C}}_{lm}\int_{-\infty}^{\infty}\!\frac{d\omega}{2\pi}\;\left\{ \mathcal{K}_{\nu k}^{nr}\Bigl[\tilde{\bm{G}}_{H}(\omega)\Bigr]_{nn}+\mathcal{L}_{\nu k}^{nr}(\omega)\Bigl[\tilde{\bm{G}}_{H}(\omega)\Bigr]_{\nu\nu}\right\} \,,
\]
where 
\begin{align}
\mathcal{K}_{\nu k}^{nr}(\omega) & \equiv\omega\text{Im}\left([\tilde{\bm{\mathfrak{D}}}(\omega)]_{nr}\tilde{\bm{\mathfrak{D}}}(\omega)]_{nk}\right)+4\gamma\omega^{2}\text{Re}\left([\tilde{\bm{\mathfrak{D}}}(\omega)]_{nr}[\tilde{\bm{\mathfrak{D}}}(\omega)]_{nn}^{*}[\tilde{\bm{\mathfrak{D}}}(\omega)]_{nk}\right)\nonumber \\
 & +\omega\left[\lambda_{2}\text{Im}\left([\tilde{\bm{\mathfrak{D}}}(\omega)]_{nr}[\tilde{\bm{\mathfrak{D}}}(\omega)]_{\nu n}^{*}[\tilde{\bm{\mathfrak{D}}}(\omega)]_{nk}\right)-\text{Im}\left([\tilde{\bm{\mathfrak{D}}}(\omega)]_{\nu r}[\tilde{\bm{\mathfrak{D}}}(\omega)]_{nn}^{*}[\tilde{\bm{\mathfrak{D}}}(\omega)]_{nk}\right)\right]\label{eq:calK}\\
\mathcal{L}_{\nu k}^{nr}(\omega) & \equiv4\gamma\omega^{2}\text{Re}\left([\tilde{\bm{\mathfrak{D}}}(\omega)]_{nr}[\tilde{\bm{\mathfrak{D}}}(\omega)]_{n\nu}^{*}[\tilde{\bm{\mathfrak{D}}}(\omega)]_{\nu k}\right)+\nonumber \\
 & +\omega\lambda_{2}\left[\text{Im}\left([\tilde{\bm{\mathfrak{D}}}(\omega)]_{nr}[\tilde{\bm{\mathfrak{D}}}(\omega)]_{\nu\nu}^{*}[\tilde{\bm{\mathfrak{D}}}(\omega)]_{\nu k}\right)-\text{Im}\left([\tilde{\bm{\mathfrak{D}}}(\omega)]_{\nu r}[\tilde{\bm{\mathfrak{D}}}(\omega)]_{n\nu}^{*}[\tilde{\bm{\mathfrak{D}}}(\omega)]_{\nu k}\right)\right]\label{eq:calL}
\end{align}
Let us now using Eqs. (\ref{eq:off-diag-ME}-\ref{eq:Inverse2}) to
simplify $\mathcal{K}_{\nu k}^{nr}(\omega)$ and $\mathcal{L}_{\nu k}^{nr}(\omega)$.
Eliminate $[\tilde{\bm{\mathfrak{D}}}(\omega)]_{\nu n}^{*}$ in Eq.
\eqref{eq:calK} with Eq. \eqref{eq:Inverse1}, we obtain 
\begin{align}
\mathcal{K}_{\nu k}^{nr}(\omega) & =\omega\text{Im}\left\{ (-\omega^{2}+2i\gamma\omega+\omega_{R0}^{2}+\lambda_{2})[\tilde{\bm{\mathfrak{D}}}(\omega)]_{nr}[\tilde{\bm{\mathfrak{D}}}(\omega)]_{nn}^{*}[\tilde{\bm{\mathfrak{D}}}(\omega)]_{nk}\right\} \nonumber \\
 & -\omega\lambda_{2}\text{Im}\left([\tilde{\bm{\mathfrak{D}}}(\omega)]_{\nu r}[\tilde{\bm{\mathfrak{D}}}(\omega)]_{nn}^{*}[\tilde{\bm{\mathfrak{D}}}(\omega)]_{nk}\right)
\end{align}
Upon using Eq. \eqref{eq:Inverse1} for $r=n$ and Eq. \eqref{eq:Inverse2}
for $r=\nu$, one immediately obtains 
\begin{eqnarray}
\mathcal{K}_{\nu k}^{nn}(\omega) & = & \omega\lambda_{2}\text{Im}\left([\tilde{\bm{\mathfrak{D}}}(\omega)]_{nn}^{*}[\tilde{\bm{\mathfrak{D}}}(\omega)]_{nk}\right)\label{eq:K-Id1}\\
\mathcal{K}_{\nu k}^{n\nu}(\omega) & = & 0\label{eq:K-Id2}
\end{eqnarray}
Similarly, eliminating $[\tilde{\bm{\mathfrak{D}}}(\omega)]_{\nu\nu}^{*}$
in Eq. \eqref{eq:calL} with Eq. \eqref{eq:Inverse2} yields 
\begin{align}
\mathcal{L}_{\nu k}^{nr}(\omega) & =\omega\text{Im}\left\{ (-\omega^{2}+2i\gamma\omega+\omega_{R0}^{2}+\lambda_{2})[\tilde{\bm{\mathfrak{D}}}(\omega)]_{nr}[\tilde{\bm{\mathfrak{D}}}(\omega)]_{n\nu}^{*}[\tilde{\bm{\mathfrak{D}}}(\omega)]_{\nu k}\right\} \nonumber \\
 & -\omega\lambda_{2}\text{Im}\left([\tilde{\bm{\mathfrak{D}}}(\omega)]_{\nu r}[\tilde{\bm{\mathfrak{D}}}(\omega)]_{n\nu}^{*}[\tilde{\bm{\mathfrak{D}}}(\omega)]_{\nu k}\right)
\end{align}
Then again using Eq. \eqref{eq:Inverse1} for $r=n$ and Eq. \eqref{eq:Inverse2}
for $r=\nu$, one obtain 
\begin{align}
\mathcal{L}_{\nu k}^{nn}(\omega) & =\omega\text{Im}\left([\tilde{\bm{\mathfrak{D}}}(\omega)]_{n\nu}^{*}[\tilde{\bm{\mathfrak{D}}}(\omega)]_{\nu k}\right)\label{eq:L-Id1}\\
\mathcal{L}_{\nu k}^{n\nu}(\omega) & =0\label{eq:L-Id2}
\end{align}
Now from Eqs. (\ref{eq:K-Id1}, \ref{eq:K-Id2}, \ref{eq:L-Id1},
\ref{eq:L-Id2}), one immediately obtains 
\begin{equation}
-\prescript{}{\text{I}}{\Gamma}_{nnn}^{nn}+2\gamma\prescript{}{\text{I}}{\tilde{\Gamma}}_{nnn}^{nn}-\lambda_{2}\left(\prescript{}{\text{I}}{\Upsilon}_{\nu nnn}^{nn}+\prescript{}{\text{I}}{\tilde{\Upsilon}}_{\nu nnn}^{nn}\right)=0\label{E:kjrkthe1}
\end{equation}
\begin{equation}
-\prescript{}{\text{I}}{\Gamma}_{klm}^{n\nu}+2\gamma\prescript{}{\text{I}}{\tilde{\Gamma}}_{klm}^{n\nu}-\lambda_{2}\left(\prescript{}{\text{I}}{\Upsilon}_{\nu klm}^{n\nu}+\prescript{}{\text{I}}{\tilde{\Upsilon}}_{\nu klm}^{n\nu}\right)=0\label{E:kjrkthe2}
\end{equation}
Furthermore, Eqs. (\ref{eq:K-Id1}, \ref{eq:K-Id2}, \ref{eq:L-Id1},
\ref{eq:L-Id2}) also indicate that 
\begin{equation}
-\prescript{}{\text{I}}{\Gamma}_{\nu lm}^{nn}+2\gamma\prescript{}{\text{I}}{\tilde{\Gamma}}_{\nu lm}^{nn}-\lambda_{2}\left(\prescript{}{\text{I}}{\Upsilon}_{\nu\nu lm}^{nn}+\prescript{}{\text{I}}{\tilde{\Upsilon}}_{\nu\nu lm}^{nn}\right)=\prescript{}{\text{I}}{\Lambda}_{\nu lm}^{n}\label{E:kjrkthe3}
\end{equation}
The procedures work for cases I\!I and I\!I\!I because they differ
by permutations of the indices, so Eqs.~\eqref{E:kjrkthe1}--\eqref{E:kjrkthe3}
hold in general without the left subscripts I, I\!I and I\!I\!I.
Therefore, we see that Eq. \eqref{E:kjrkthe1} leads to Eq. \eqref{eq:NESS-Id0}
while Eq. \eqref{E:kjrkthe2}leads to Eqs. (\ref{eq:NESS-Id1}-\ref{eq:NESS-Id2}).
Finally, Eq. \eqref{E:kjrkthe3} proves Eqs. (\ref{eq:NESS-Id3}-\ref{eq:NESS-Id4}).


\begin{thebibliography}{10}
\bibitem{GalCoh}G. Gallavotti and E. G. D. Cohen, ``Dynamical ensembles
in nonequilibrium statistical mechanics'', Phys. Rev. Lett. \textbf{74},
2694 (1995).

\bibitem{Evans} D. J. Evans, E. G. D. Cohen and G. P. Morriss, ``Probability
of second law violations in shearing steady states'', Phys. Rev.
Lett. \textbf{71}, 2401 (1993);\\
 D. J. Evans and D. J. Searles, ``Equilibrium microstates which generate
second law violating steady states'', Phys. Rev. E \textbf{50}, 1645
(1994); \\
 D. J. Evans and D. J. Searles, ``The fluctuation theorem", Adv.
Phys. \textbf{51}, 1529 (2002).

\bibitem{Kurchan} J. Kurchan, ``Fluctuation theorem for stochastic
dynamics'', J. Phys. A \textbf{31}, 3719 (1998).

\bibitem{LebSpoFluc} J. L. Lebowitz and H. Spohn, ``A Gallavotti-Cohen-Type
symmetry in the large deviation functional for stochastic dynamics'',
J. Stat. Phys. \textbf{95}, 333 (1999).

\bibitem{Seifert} U. Seifert, ``Stochastic thermodynamics, fluctuation
theorems and molecular machines'', Rep. Prog. Phys. \textbf{75},
126001, (2012).

\bibitem{Jar97} C. Jarzynski, ``Nonequilibrium equality for free
energy differences'', Phys. Rev. Lett. \textbf{78}, 2690 (1997);
\\
 C. Jarzynski, ``Equilibrium free-energy differences from nonequilibrium
measurements: a master-equation approach'', Phys. Rev. E \textbf{56},
5018 (1997).

\bibitem{Crook} G. E. Crooks, ``Entropy production fluctuation theorem
and the nonequilibrium work relation for free energy differences'',
Phys. Rev. E \textbf{60}, 2721 (1999).

\bibitem{FlucThmRev} M. Esposito, U. Harbola and S. Mukamel, ``Nonequilibrium
fluctuations, fluctuation theorems, and counting statistics in quantum
systems'', Rev. Mod. Phys. \textbf{81}, 1665 (2009); \\
 M. Campisi, P. Hänggi and P. Talkner, ``Quantum fluctuation relations:
foundations and applications'', Rev. Mod. Phys. \textbf{83}, 771
(2011).

\bibitem{JarRev} C. Jarzynski, ``Equalities and inequalities: irreversibility
and the second law of thermodynamics at the nanoscale'', Ann. Rev.
Cond. Mat. Phys. \textbf{2}, 329 (2011).

\bibitem{SubHu} Y. Subaşi and B. L. Hu, ``Quantum and classical
fluctuation theorems from a decoherent-histories open-system analysis'',
Phys. Rev. E \textbf{85}, 011112 (2012).

\bibitem{Mahler} J. Gemmer, M. Michel and G. Mahler, ``Quantum Thermodynamics
- Emergence of Thermodynamic Behavior within Composite Quantum Systems'',
2nd Edition (Springer Verlag, Berlin, 2004).

\bibitem{Kosloff} R. Kosloff, ``Quantum thermodynamics: a dynamical
viewpoint'', Entropy \textbf{15}, 2100 (2013).

\bibitem{Anders} S. Vinjanampathy and J. Anders, ``Quantum thermodynamics'',
Contemp Phys. \textbf{57}, 545 (2016).

\bibitem{QTDQI} S. Deffner and S. Campbell, ``Quantum Thermodynamics:
An Introduction to the Thermodynamics of Quantum Information'', (Morgan
\& Claypool, 2019); {[}arXiv:1907.01596{]}.

\bibitem{NESSclHOC} Z. Rieder, J. L. Lebowitz and E. Lieb, ``Properties
of a harmonic crystal in a stationary non-equilibrium state'', J.
Math. Phys. \textbf{8}, 1073 (1967); \\
 A. Casher and J. L. Lebowitz, ``Heat flow in regular and disordered
harmonic chains'', J. Math. Phys. \textbf{12}, 1701 (1971); \\
 A. J. O'Connor and J. L. Lebowitz, ``Heat conduction and sound transmission
in isotopically disordered harmonic crystals'', J. Math. Phys. \textbf{15},
692 (1974); \\
 H. Spohn and J. L. Lebowitz, ``Stationary non-equilibrium states
of infinite harmonic systems'', Commun. Math. Phys. \textbf{54},
97 (1977).

\bibitem{NESSclAHO} J.-P. Eckmann, C.-A. Pillet and L. Rey-Bellet,
``Non-equilibrium statistical mechanics of anharmonic chains coupled
to two heat baths at different temperatures'', Commun. Math. Phys.
\textbf{201}, 657 (1999); \\
 J.-P. Eckmann and M. Hairer, ``Non-equilibrium statistical mechanics
of strongly anharmonic chains of oscillators'', Commun. Math. Phys.
\textbf{212}, 105 (2000); \\
 L. Rey-Bellet and L. E. Thomas, ``Exponential convergence to non-equilibrium
stationary states in classical statistical mechanics'', Commun. Math.
Phys. \textbf{225}, 305 (2002).

\bibitem{HHAoP} J.-T. Hsiang and B. L. Hu, ``Nonequilibrium steady
state in open quantum systems: influence action, stochastic equation
and power balance'', Ann. of Phys. \textbf{362}, 139 (2015).

\bibitem{Motz} T. Motz, J. Ankerhold and J. T. Stockburger, ``Currents
and fluctuations of quantum heat transport in harmonic chains'',
New J. Phys. \textbf{19} 053013 (2017);\\
 T. Motz, M. Wiedmann, J. T. Stockburger and J. Ankerhold, ``Rectification
of heat currents across nonlinear quantum chains: a versatile approach
beyond weak thermal contact'', New J. Phys. \textbf{20}, 113020 (2018).

\bibitem{FPU} E. Fermi, J. Pasta and S. Ulam, ``Studies of nonlinear
problems'', Los Alamos report, LA-1940 (1955).

\bibitem{FPUrev} D. K. Campbell, P. Rosenau and G. Zaslavsky, ``Introduction:
the Fermi-Pasta-Ulam problem: the first fifty years'', Chaos \textbf{1},
015101 (2005); \\
 G. Gallavotti, ``The Fermi-Pasta-Ulam Problem: A Status Report'',
Lecture Notes in Physics, Vol. \textbf{728}, (Springer Verlag, Berlin,
2008).

\bibitem{Montroll} F. T. Hioe and E. W. Montroll, ``Quantum theory
of anharmonic oscillators. I. Energy levels of oscillators with positive
quartic anharmonicity'', J. Math. Phys. \textbf{16}, 1945 (1975).

\bibitem{Bender} C. M. Bender and S. Boettcher, ``Quasi-exactly
solvable quartic potential'', J. Phys. A \textbf{31}, L273 (1998).

\bibitem{cgea} B. L. Hu, Lectures at the Seventh International Latin-American
Symposium on General Relativity (SILARG VII). Proceeding appeared
as ``Relativity and Gravitation: Classical and Quantum'', edited
by J. D'Olivio et al (World Scientific, Singapore, 1991); \\
 Y. Zhang, Ph.D thesis (University of Maryland, 1990);\\
 E. Calzetta, B. L. Hu and F. D. Mazzitelli, ``Coarse-grained effective
action and renormalization group theory in semiclassical gravity and
cosmology'', Phys. Rep. \textbf{352}, 459 (2001).

\bibitem{JH1} P. R. Johnson and B. L. Hu, ``Stochastic theory of
relativistic particles moving in a quantum field: scalar Abraham-Lorentz-Dirac-Langevin
equation, radiation reaction, and vacuum fluctuations'', Phys. Rev.
D \textbf{65}, 065015 (2002).

\bibitem{GH1} C. R. Galley and B. L. Hu, ``Self-force with a stochastic
component from radiation reaction of a scalar charge moving in curved
spacetime'', Phys. Rev. D \textbf{72}, 084023 (2005).

\bibitem{CRV} E. Calzetta, A. Roura and E. Verdaguer, ``Stochastic
description for open quantum systems'', Physica A \textbf{319}, 188
(2003).

\bibitem{NENL1} J.-T. Hsiang and B. L. Hu, ``Nonequilibrium nonlinear
open quantum systems I: Functional perturbative analysis of a weakly
anharmonic oscillator'', Phys. Rev. D (2020); {[}arXiv:1912.12803{]}.

\bibitem{NLFDR} J.-T. Hsiang and B. L. Hu, ``Fluctuation-dissipation
relation from the nonequilibrium dynamics of a nonlinear open quantum
system'', Phys. Rev. D (2020); {[}arXiv:2002.07694{]}.

\bibitem{HPZ93} B. L. Hu, J. P. Paz and Y. Zhang, ``Quantum Brownian
motion in a general environment. II. nonlinear coupling and perturbative
approach'', \textbf{47}, 1576 (1993).

\bibitem{Fourier} F. Bonetto, J. L. Lebowitz and L. Rey-Bellet, ``Fourier
law: a challenge to theorists'' in Mathematical Physics 2000, (Imperial
College Press, London 2000); {[}arXiv:math-ph/0002052{]}.

\bibitem{LLP} S. Lepri, R. Livi and A. Politi, ``Heat conduction
in chains of nonlinear oscillators'', Phys. Rev. Lett. \textbf{78},
1869 (1997); ``On the anomalous thermal conductivity of one-dimensional
lattices'', Europhys. Lett. \textbf{43}, 271 (1998).

\bibitem{LLPrep} S. Lepri, R. Livi and A. Politi,``Thermal conduction
in classical low-dimensional lattices'', Phys. Rep. \textbf{377},
1 (2003).

\bibitem{Dhar} A. Dhar, ``Heat transport in low-dimensional systems'',
Adv. Phys. \textbf{57}, 457 (2008).

\bibitem{LiRen} N. Li, J. Ren, L. Wang, G. Zhang, P. Hänggi and B.
Li, ``Phononics: Manipulating heat flow with electronic analogs and
beyond'', Rev. Mod. Phys. \textbf{84}, 1045 (2012).

\bibitem{DSH} A. Dhar, K. Saito and P. Hänggi, ``Nonequilibrium
density-matrix description of steady-state quantum transport'', Phys.
Rev. E \textbf{85}, 011126 (2012).

\bibitem{KadBay} L. P. Kadanoff and G. Baym, "Quantum Theory of
Many-Particle Systems", (Dover, 1962).

\bibitem{Rammer} J. Rammer, ``Quantum Field Theory of Non-equilibrium
States'' (Cambridge University Press, Cambridge, 2007).

\bibitem{MacWanGuo} J. Maciejko, J. Wang, and H. Guo, ``Time-dependent
quantum transport far from equilibrium: An exact nonlinear response
theory'', Phys. Rev. B \textbf{74}, 085324 (2006).

\bibitem{WALT} J.-S.Wang, B. K. Agarwalla, H. Li and J. Thingna,
``Nonequilibrium Green's function method for quantum thermal transport'',
Front. Phys. \textbf{9}, 673 (2014).

\bibitem{PJSB03}Pilgram, S., Jordan, A.N., Sukhorukov, E.V. and Büttiker,
M., 2003. Stochastic path integral formulation of full counting statistics.
Physical Review Letters, 90(20), p.206801.

\bibitem{JSP04}Jordan, A.N., Sukhorukov, E.V. and Pilgram, S., 2004.
Fluctuation statistics in networks: A stochastic path integral approach.
Journal of mathematical physics, 45(11), pp.4386-4417.

\bibitem{ctp} J. S. Schwinger, ``Brownian motion of a quantum oscillator'',
J. Math. Phys. (N.Y.) \textbf{2}, 407 (1961); \\
 L. Keldysh, ``Diagram technique for nonequilibrium processes'',
Zh. Eksp. Teor. Fiz. \textbf{47}, 1515 (1964); {[}JETP \textbf{20},
1018 (1964){]}; \\
 K.-C. Chou, Z.-B. Su, B.-L. Hao and L. Yu, ``Equilibrium and non-equilibrium
formalisms made unified'', Phys. Rep. \textbf{118}, 1 (1985).

\bibitem{CalHu08} E. A. Calzetta and B.-L. B. Hu, ``Nonequilibrium
Quantum Field Theory'', (Cambridge University Press, Cambridge, 2008).

\bibitem{Kamenev} A. Kamenev, ``Field Theory of Non-equilibrium
Systems'', (Cambridge University Press, Cambridge, 2011).

\bibitem{Weiss} U. Weiss, ``Quantum Dissipative Systems, 4th Edition''.
(World Scientific, New Jersey, 2012).

\bibitem{IF} R. P. Feynman and F. L. Vernon, ``The theory of a general
quantum system interacting with a linear dissipative system '', Ann.
Phys. (N.Y.) \textbf{24}, 118 (1963); \\
 A. O. Caldeira and A. J. Leggett, ``Path integral approach to quantum
Brownian motion'', Physica A \textbf{121}, 587 (1983);\\
 H. Grabert, P. Schramm and G. L. Ingold, ``Quantum Brownian motion,
the functional integral approach'', Phys. Rep. \textbf{168}, 115
(1988).

\bibitem{HPZ92} B. L. Hu, J. P. Paz and Y. Zhang, ``Quantum Brownian
motion in a general environment: exact master equation with nonlocal
dissipation and colored noise'', Phys. Rev. D \textbf{45}, 2843 (1992).

\bibitem{Zhao} H. Zhao, Z. Wen, Y. Zhang and D. Zheng, ``Dynamics
of solitary wave scattering in the Fermi-Pasta-Ulam model'', Phys.
Rev. Lett. \textbf{94},025507 (2005); \\
 H. Zhao, ``Identifying diffusion processes in one-dimensional lattices
at thermal equilibrium'', Phys. Rev. Lett. \textbf{96}, 140602 (2006);
\\
 S. Chen, Y. Zhang, J. Wang and H. Zhao, ``Diffusion of heat, energy,
momentum and mass in one-dimensional systems'', Phys. Rev. E \textbf{87},
032153 (2013); \\
 H. Zhao and H. Zhao, ``Brownian motion: from kinetics to hydrodynamics'',
{[}arxiv:1706.00779{]};\\
 D. Xiong, ``Heat perturbation spreading in the Fermi-Pasta-Ulam-beta
system with next-nearest-neighbor coupling: Competition between phonon
dispersion and nonlinearity'', Phys. Rev. E \textbf{95}, 062140 (2017).

\bibitem{QBE} H. Spohn, ``The phonon Boltzmann equation, properties
and link to weakly anharmonic lattice dynamics'', J. Stat. Phys.
\textbf{124}, 1041 (2006).

\bibitem{Spohn} C. B. Mendl and H. Spohn, ``Dynamic correlators
of Fermi-Pasta-Ulam chains and nonlinear fluctuating hydrodynamics'',
Phys. Rev. Lett. \textbf{111}, 230601 (2013);\\
 H. Spohn, ``Nonlinear fluctuating hydrodynamics for anharmonic chains'',
J. Stat. Phys. \textbf{154}, 1191 (2014); \\
 S. G. Das, A. Dhar, K. Saito, C. B. Mendl and H. Spohn, ``Numerical
test of hydrodynamic fluctuation theory in the Fermi-Pasta-Ulam chain'',
Phys. Rev. E \textbf{90}, 012124 (2014); \\
 C. B. Mendl and H. Spohn, ``Current fluctuations for anharmonic
chains in thermal equilibrium'', J. Stat. Mech. 2015, P03007 (2015).

\bibitem{DharPC} A. Dhar, private communication 2014;\\
 A. Dhar and K. Wagh, ``Equilibration problem for the generalized
Langevin equation'', Europhys. Lett. \textbf{79}, 60003 (2007).

\bibitem{ZhaoPC} H. Zhao, private communication 2014.

\bibitem{FZZ19} W. Fu, Y. Zhang and H. Zhao, `` Universal scaling
of the thermalization time in one-dimensional lattices'', Phys. Rev.
E \textbf{100}, 010101(R) (2019).

\bibitem{WFZZ19} Z. Wang, W. Fu, Y. Zhang and H. Zhao, ``Arbitrarily
weak nonlinearity can destroy the Anderson localization''; {[}arXiv:1903.09502v2{]}.

\bibitem{POC18} L. Pistone, M. Onorato and S. Chibbaro, ``Thermalization
in the discrete nonlinear Klein-Gordon chain in the wave-turbulence
framework'', Euro. Phys. Lett. 121 44003 (2018).

\bibitem{PCBLO} L. Pistone, S. Chibbaro, M. D. Bustamante, Y. L'vov
and M. Onorato, ``Universal route to thermalization in weakly-nonlinear
one-dimensional chains'', Mathematics in Engineering \textbf{1},
672 (2019); {[}arXiv:1812.08279{]}.

\bibitem{CDJ13}A. Chantasri, J. Dressel, and A. N. Jordan. \textquotedbl Action
principle for continuous quantum measurement.\textquotedbl{} Physical
Review A 88.4 (2013): 042110.

\bibitem{CJ15}A. Chantasri, and A. N. Jordan. \textquotedbl Stochastic
path-integral formalism for continuous quantum measurement.\textquotedbl{}
Physical Review A 92.3 (2015): 032125.

\bibitem{StoMak} J. T. Stockburger and C. H. Mak, ``Stochastic Liouvillian
algorithm to simulate dissipative quantum dynamics with arbitrary
precision'' J. Chem. Phys. 110, 4983 (1999).

\bibitem{JJA-NatComm}Léger, S., Puertas-Martínez, J., Bharadwaj,
K., Dassonneville, R., Delaforce, J., Foroughi, F., ... \& Hasch-Guichard,
W. (2019). Observation of quantum many-body effects due to zero point
fluctuations in superconducting circuits. Nature communications, 10(1),
1-8.

\bibitem{JJA-Nature}Kuzmin, R., Mencia, R., Grabon, N., Mehta, N.,
Lin, Y. H., \& Manucharyan, V. E. (2019). Quantum electrodynamics
of a superconductor--insulator phase transition. Nature Physics,
15(9), 930-934.

\bibitem{Paz} Martinez, Esteban A., and Juan Pablo Paz. ``Dynamics
and thermodynamics of linear quantum open systems.'' Physical review
letters 110.13 (2013): 130406. 
\end{thebibliography}
\end{document}